\documentclass[11pt]{article}

\pdfoutput=1

\usepackage[letterpaper, portrait, margin=1in]{geometry}
\usepackage[colorlinks=true,linkcolor=blue,citecolor=ForestGreen]{hyperref}
\usepackage{algorithm}
\usepackage[noend]{algpseudocode}
\usepackage{url}
\usepackage{amsmath,amssymb,amsthm}
\usepackage{thmtools,thm-restate}
\usepackage[noabbrev,capitalise,nameinlink]{cleveref}
\usepackage{mathtools}
\usepackage{xspace}
\usepackage{verbatim}
\usepackage{mathrsfs}
\usepackage[usenames,dvipsnames,svgnames,table]{xcolor}
\usepackage{pgf}
\usepackage[dvipsnames]{xcolor}
\usepackage{todo}
\usepackage{tabularx}
\usepackage{enumitem}

\usepackage{dsfont}

\newcommand*\ie{i.\kern.1em e.\ }
\newcommand*\eg{e.\kern.1em g.\ }

\theoremstyle{plain}
\newtheorem{theorem}{Theorem}[section]
\newtheorem{lemma}[theorem]{Lemma}
\newtheorem{fact}[theorem]{Fact}
\newtheorem{proposition}[theorem]{Proposition}
\newtheorem{claim}[theorem]{Claim}

\theoremstyle{definition}
\newtheorem{observation}[theorem]{Observation}

\newtheorem{definition}[theorem]{Definition}
\newtheorem{remark}[theorem]{Remark}
\newtheorem{example}[theorem]{Example}

\theoremstyle{plain}

\newcommand{\ignore}[1]{}

\DeclareMathOperator{\poly}{poly}

\newcommand{\dist}{\mathsf{dist}}

\newcommand{\Var}[1]{\mathrm{Var} \left[ #1 \right]}
\newcommand{\Varu}[2]{\underset{ #1 } {\mathrm{Var}} \left[ #2 \right]}
\newcommand{\Varuc}[3]{\underset{ #1 } {\mathrm{Var}} \left[ #2 \;\; \left| \;\; #3 \right. \right]}

\newcommand{\Ex}[1]{\bE \left[ #1 \right]}
\newcommand{\Exu}[2]{\underset{#1} \bE \left[ #2 \right] }
\newcommand{\Exuc}[3]{\underset{#1} \bE \left[ #2 \;\; \left| \;\; #3
\right.\right] }

\renewcommand{\Pr}[1]{\bP \left[ #1 \right]} 
\newcommand{\Pru}[2]{\underset{ #1 }\bP \left[ #2 \right]}

\newcommand{\define}{\vcentcolon=}

\newcommand{\floor}[1]{\ensuremath{\lfloor #1 \rfloor}}
\newcommand{\ceil}[1]{\ensuremath{\lceil #1 \rceil}}

\DeclarePairedDelimiter{\abs}{\lvert}{\rvert}

\newcommand{\ind}[1]{\mathds{1} \left[ #1 \right] }

\newcommand{\zo}{\{0,1\}}

\newcommand{\cA}{\ensuremath{\mathcal{A}}}

\newcommand{\cD}{\ensuremath{\mathcal{D}}}
\newcommand{\cE}{\ensuremath{\mathcal{E}}}

\newcommand{\cG}{\ensuremath{\mathcal{G}}}

\newcommand{\cI}{\ensuremath{\mathcal{I}}}

\newcommand{\cP}{\ensuremath{\mathcal{P}}}

\newcommand{\cS}{\ensuremath{\mathcal{S}}}

\newcommand{\cU}{\ensuremath{\mathcal{U}}}

\newcommand{\cY}{\ensuremath{\mathcal{Y}}}
\newcommand{\cX}{\ensuremath{\mathcal{X}}}

\newcommand{\bB}{\ensuremath{\mathbb{B}}}

\newcommand{\bE}{\ensuremath{\mathbb{E}}}

\newcommand{\bN}{\ensuremath{\mathbb{N}}}
\newcommand{\bP}{\ensuremath{\mathbb{P}}}
\newcommand{\bR}{\ensuremath{\mathbb{R}}}

\newcommand{\bZ}{\ensuremath{\mathbb{Z}}}

\usepackage{tgpagella}
\usepackage{float}
\usepackage{afterpage}
\usepackage{accents}
\usepackage{bm}
\usepackage{centernot}
\usepackage[nottoc]{tocbibind}
\usepackage{tocloft}
\restylefloat{table}
\usepackage{etoolbox}
\usepackage{nicematrix}
\usepackage{caption}
\usepackage{xspace}

\newcommand\lequestion{\stackrel{\mathclap{\normalfont\mbox{\tiny ?}}}{\le}}
\newcommand\gtquestion{\stackrel{\mathclap{\normalfont\mbox{\tiny ?}}}{>}}
\newcommand\gequestion{\stackrel{\mathclap{\normalfont\mbox{\tiny ?}}}{\ge}}
\newcommand\eqquestion{\stackrel{\mathclap{\normalfont\mbox{\tiny ?}}}{=}}

\newcommand{\rep}{\mathsf{rep}}
\newcommand{\SAMP}{\mathsf{SAMP}}
\newcommand{\LABEL}{\mathsf{LABEL}}

\newcommand{\GammaStar}{\Gamma^*_{\epsilon/4}}

\newcommand{\induced}[2]{{#1}_{\lfloor{#2}\rfloor}}

\newcommand{\CC}{\mathcal{CC}} 
\newcommand{\CONN}{\mathcal{C}} 
\newcommand{\BOX}{\mathcal{B}}
\newcommand{\CONV}{\mathcal{CV}}

\newcommand{\Poi}{\mathsf{Poi}}

\newcommand{\TV}{\mathsf{TV}}

\newcommand{\nstar}{n^*_{\epsilon/4}}

\newcommand{\EMD}{\mathsf{EMD}}
\newcommand{\qcell}{\mathsf{qcell}}
\newcommand{\qreject}{\mathsf{qreject}}

\newcommand{\cost}{\mathsf{cost}}
\newcommand{\diam}{\mathsf{diam}}

\DeclareMathOperator{\smallinterval}{\mathsf{small}}
\DeclareMathOperator{\largeinterval}{\mathsf{large}}

\newcommand{\bPhi}{\boldsymbol{\Phi}}
\newcommand{\cyc}{\mathsf{cycle}}
\newcommand{\pth}{\mathsf{path}}

\DeclarePairedDelimiterX{\infdivx}[2]{(}{)}{%
  #1\;\delimsize\|\;#2%
}
\newcommand{\RelativeConcentrationSymbol}{\chi}
\newcommand{\RelativeConcentration}[1]{\RelativeConcentrationSymbol(#1)}
\newcommand{\RelativeConcentrationT}[2]{\RelativeConcentrationSymbol_{#2}(#1)}

\DeclareMathOperator{\interior}{\mathsf{int}}
\newcommand{\closure}[1]{\mathsf{cl}({#1})}

\title{Distribution Testing with a Confused Collector\footnote{Some results in this paper appeared
together with a number of others in an earlier preprint \cite{FHparity}. We have expanded
upon our earlier results and included them here, separate from the remainder of \cite{FHparity},
after receiving feedback on the preprint.}}

\author{%
  Renato Ferreira Pinto Jr.\thanks{Partly funded by an NSERC Canada Graduate Scholarship Doctoral
  Award.}\\
  University of Waterloo\\
  \texttt{r4ferrei@uwaterloo.ca}
\and Nathaniel Harms\thanks{Partly funded by an NSERC Postdoctoral Fellowship, and the Swiss State
Secretariat for Education, Research and Innovation (SERI) under contract number MB22.00026.  Some of
this work was done while the author was at the University of Waterloo.}\\
  EPFL \\
  \texttt{nathaniel.harms@epfl.ch}}

\date{}

\begin{document}
\maketitle

\begin{abstract}
We are interested in testing properties of distributions with systematically mislabeled samples. Our
goal is to make decisions about unknown probability distributions, using a sample that has been
collected by a \emph{confused collector}, such as a machine-learning classifier that has not learned
to distinguish all elements of the domain. The confused collector holds an unknown clustering of the
domain and an input distribution $\mu$, and provides two oracles: a \emph{sample oracle} which
produces a sample from $\mu$ that has been labeled according to the clustering; and a
\emph{label-query oracle} which returns the label of a query point $x$ according to the clustering.

Our first set of results shows that identity, uniformity, and equivalence of distributions can be
tested efficiently, under the earth-mover distance, with remarkably weak conditions on the confused
collector, even when the unknown clustering is \emph{adversarial}. This requires defining a variant
of the distribution testing task (inspired by the recent \emph{testable learning} framework of
Rubinfeld \& Vasilyan), where the algorithm should test a joint property of the distribution and its
clustering.  As an example, we get efficient testers when the distribution tester is allowed to
reject if it detects that the confused collector clustering is ``far'' from being a decision tree.

The second set of results shows that we can sometimes do significantly better when 
the clustering is \emph{random} instead of adversarial. For certain one-dimensional random
clusterings, we show that uniformity can be tested under the TV distance using $\widetilde
O\left(\frac{\sqrt n}{\rho^{3/2} \epsilon^2}\right)$ samples and \emph{zero} queries, where $\rho
\in (0,1]$ controls the ``resolution'' of the clustering. We improve this to $O\left(\frac{\sqrt
n}{\rho \epsilon^2}\right)$ when queries are allowed.
\end{abstract}

\thispagestyle{empty}
\setcounter{page}{0}
\newpage

\thispagestyle{empty}
\setcounter{page}{0}
\newpage
{\small
\setcounter{tocdepth}{2} 
\tableofcontents
}
\thispagestyle{empty}
\setcounter{page}{0}
\newpage
\setcounter{page}{1}

\section{Introduction}

We are interested in the problem of making decisions about an unknown probability distribution, using only samples from that distribution which have been collected and labeled by an entity who
may not be capable of accurate distinctions between elements of the domain. Consider some informal
but illustrative examples:
\begin{enumerate}[itemsep=0pt,leftmargin=1.5em]
\item We wish
to make a decision about the distribution of woodland flora based on a sample that was tabulated by
a research assistant who cannot differentiate between black spruce and white spruce, or between red
maple and sugar maple, and has counted the spruces together and the maples together by
mistake\footnote{We thank ecologist Prof.~Julie Messier for these examples of trees frequently
mistaken by students.}.
\item Our sample of a distribution of images has been labeled by an algorithm, such as a machine
learning classifier, that fails to distinguish between, say, red squirrels and grey squirrels.
For instance, the algorithm might be represented by a decision tree, and we are not certain if the
decision tree has sufficient granularity for our task.
\item Sample points have been truncated due to rounding, or hashed in an effort to save space,
possibly causing collisions.
\item Random environmental conditions prevent the collector from making perfect distinctions. Say
we wish to make decisions about the distribution of fossils over time, but we are unable
to distinguish between fossils from year $x$ and year $y$ unless those years are separated by some
random geological event that leaves traces in the rock. Note that this is different from having a
sample corrupted by random noise, since the mislabelling is \emph{systematic}, applying to all
sample points in the same way.
\end{enumerate}
These types of constraints are essentially unavoidable in practice, and they also occur in
theoretical analysis of distribution-free property testing algorithms (which we explain briefly in
\cref{section:discussion}).  When faced with a situation like in these examples, we call the
collector of the sample a \emph{confused collector}, and our goal is to design \emph{distribution
testing} algorithms that work even when faced with a confused collector, which we will define
formally below.

Distribution testing is a fundamental type of statistical task, where the goal is to determine
whether an unknown probability distribution $\mu$ belongs to a property $\cP$, or is
\emph{$\epsilon$-far} from the property $\cP$, meaning that its distance to any distribution $\nu
\in \cP$ is at least $\epsilon$; the distance metric depends on the problem but is usually assumed
to be the total-variation (TV) distance. The tester should make this decision using a random sample
from $\mu$ that is as small as possible while allowing it to succeed with
high probability (usually probability $2/3$). See \cite{Can20} for a recent survey.

To our knowledge, the most closely related work on distribution testing does not capture the
phenomenon we are interested in; they assume that the tester either: can gain perfect knowledge of
the sample using queries to the sample points \cite{GR22}; sees random noise applied independently
to each sample point \cite{AKR19}; sees samples labeled by a permutation of the correct labels
\cite{CW21}; sees samples through a privacy mechanism \cite{GR18,She18,ACT19,ACFT19,ACT20,ACFST21};
sees ``truncated'' samples restricted to a subset of the domain \cite{DGTZ18,DNS23}; or solves the
related but nearly opposite task of testing if the input \emph{can} be clustered to match some known
target \cite{CW20}. Recent work in learning theory \cite{FKKT21} states that, while there is
extensive applied learning literature on the type of mislabeled samples we describe, little is known
theoretically; they study statistical learning problems (\eg Gaussian mean estimation) in a model
similar in spirit to what we will define, but fundamentally different in the details (see
\cref{section:discussion} for a comparison).

In this paper we will define a general model for the confused collector and show that distribution
testing tasks can be performed efficiently under remarkably weak restrictions on the confused
collector, even when the collector is \emph{adversarial}. Then, we will show how to significantly
improve some of these results when the distinctions made by the collector are \emph{random}.\\

\noindent
\textbf{Modelling the Confused Collector.}
The idea is that the confused collector holds a \emph{clustering} of the domain, and labels each
sample point $x$ with a representative of its cluster.  For a fixed domain $\cX$, a
\emph{clustering} of the domain is a pair $(\Gamma, \rep)$ consisting of a partition $\Gamma = \{
\Gamma_1, \Gamma_2, \dotsc, \Gamma_k\}$ of $\cX$ into some number $k$ of \emph{cells} $\Gamma_i$,
together with a choice of \emph{representatives} of each cell, $\rep : [k] \to \cX$ such that
$\rep(i) \in \Gamma_i$ for each $i \in [k]$. We define $\gamma : \cX \to [k]$ as the function that
assigns to each point $x$ the index of its cell, so that $x \in
\Gamma_{\gamma(x)}$. The input to the distribution testing algorithm consists of a clustering
$(\Gamma,\rep)$ and one or more distributions\footnote{In \cref{section:discussion} we also
briefly discuss a setting where each input distribution is held by a different collector.} $\mu_1, \mu_2, \dotsc$. The inputs are held by the
confused collector, who provides the algorithm with access via the following oracles:

\newcommand{\slabel}{\mathsf{label}}
\begin{enumerate}[itemsep=0pt,topsep=0.5em,leftmargin=1.5em]
\item \emph{Clustered-Sample Oracle.}
For each distribution $\mu_i$ in the input, the confused collector provides 
access to $\mu_i$ via a \emph{clustered-sample oracle} $\SAMP(\Gamma,\rep,\mu_i)$. On request,
this oracle produces an independent sample point of the form $\rep(\gamma(x))$ where $x
\sim \mu_i$; \ie, the oracle provides the algorithm with the label of $x$, 
defined as the representative of the cluster that contains $x$.

\item\emph{Label Oracle.} Thinking of the confused collector as an entity (\eg a machine learning
classifier) that has labeled the sample, it is sensible to allow the tester to ask the collector
about its clustering. The confused collector provides access to a \emph{label oracle}
$\LABEL(\Gamma, \rep)$ which, on query $x \in \cX$, answers with $\rep(\gamma(x))$, the
representative of the cluster containing $x$. Unlike typical property testing models, we think of
queries as being cheap, relative to samples. For example, we may have black-box access to the
algorithm that provided labels for a sample, without having the ability to request additional
samples, or we can ask our research assistant about their clustering without sending them back to
the forest.

\end{enumerate}
Not much can be done if we allow the confused collector to hold \emph{any} clustering, while making
the same demands on the tester as in the standard model. If a sample of woodland flora were to be
labeled by the authors of this paper, not many interesting properties could be tested under the
resulting partition into only 2 or 3 cells, even if the tester exactly learns the clustering. It is
important to observe that two distributions $\mu$ and $\nu$ are indistinguishable to the tester if
they can be transformed into each other by transporting mass within individual cells of the
clustering, and therefore an adversary can force two distributions with TV distance 1 to be
indistinguishable to the tester, making standard distribution testing impossible.\\

\noindent
\textbf{Results and Organization.}
This paper is clustered into two parts. Part I shows how to sidestep the impossibility we just
described, even when the clustering is \emph{adversarial}, by defining a natural relaxation of the
distribution testing task, where the algorithm tests a joint property of the distribution and its
clustering. We present general baseline upper bounds for the \emph{identity} and \emph{equivalence
testing} tasks, under adversarial clusterings.  One example application, motivated by the
example where the confused collector is a machine learning classifier, is that the relaxed testing
tasks can be done efficiently when the tester expects the clustering to be realized by a
\emph{decision tree} and can test this assumption.

These results provide a foundation for Part II, where we present more involved technical results
showing how to significantly improve upon the above baseline results in certain special cases where
we assume \emph{randomized} clusterings.  Motivated by the ``environmental randomness'' example and
applications to property testing, we show that certain random clusters of the domain $[n]$ allow the
standard (non-relaxed) \emph{uniformity testing} task to be accomplished efficiently with
\emph{zero} label queries, and even more efficiently with queries.

\subsection{Part I: Adversarial Clustering}

We begin our exploration by establishing that natural testing tasks can be performed efficiently
even with an \emph{adversarial} confused collector. It is not obvious \emph{a priori} that
non-trivial results are even possible under adversarial clusterings, because the adversary can
provide a clustering that makes distribution testing impossible. The simple message of Part I is
that the tester can detect if this has occurred: the algorithm tests both the input distribution
\emph{and} the clustering.  Inspired by recent work on \emph{testable learning} \cite{RV23}, we
allow the tester to reject if it deems the clustering unsuitable, while requiring the tester to both
(1) accept ``good'' clusterings and (2) function properly as a distribution tester whenever it
accepts the clustering, regardless of whether it was truly ``good''\!\!.

We give general upper bounds for the basic \emph{uniformity}, \emph{identity}, and \emph{equivalence
testing} problems. In \emph{uniformity testing}, the tester checks if the input
distribution is uniform. In \emph{identity testing}, the tester knows a target distribution $\nu$
and tests whether the input $\mu$ is equal to $\nu$ or $\epsilon$-far from $\nu$.  In
\emph{equivalence testing} (sometimes called \emph{closeness testing}), the tester is given two
input distributions $\mu$ and $\nu$ and decides if $\mu = \nu$, or if they are $\epsilon$-far from
each other.
Before formally defining the testing task, we give one informal application of our results,
motivated by the example where the tester should work properly when the sample is labeled by a
``good'' machine learning classifier:

\begin{theorem}[Informal]
\label{thm:intro-cv-dt-informal}
Suppose the collector is promised to provide a clustering of the constant-dimensional cube $[0,1]^d$
into \emph{convex} cells of radius at least $\delta$. Then uniformity can be tested using
$\poly(\tfrac{1}{\epsilon} \cdot \log(1/\delta))$ samples and queries, under the earth-mover
distance, if the tester must accept clusterings realized by decision trees (with nodes of the form
``$x_i < t$?'') that put most of the input distribution in low-diameter cells,
and may otherwise reject the clustering.
\end{theorem}

\subsubsection{The Testing Task}

Let us explain how to arrive at a suitable definition of distribution testing when faced with an
adversarially-chosen clustering.  Consider the simple problem of distinguishing distributions $\mu$
supported on a single element, from the distributions $\nu$ which have probability mass $1/2$ on
each of two elements. This is trivial in the standard distribution testing model, but impossible if
the confused collector holds any clustering with a cell of size at least 2, because the two supports
of $\mu$ may lie in the same cell. One may consider three ways to fix this problem:\\

\noindent
\emph{Change the distance metric.} We should not use TV distance to define our testing task, because
the adversary can choose the clustering to ``hide'' arbitrarily large TV distances.  Assuming the
domain $\cX$ is equipped with an ambient metric denoted by $\dist(\cdot,\cdot)$, we instead
use the \emph{earth-mover} (or \emph{Wasserstein}) distance on distributions, defined as
\[
  \EMD_\dist(\mu,\nu) \define \inf_\pi \sum_{x,y \in \cX} \pi(x,y) \cdot \dist(x,y) \,,
\]
where the infimum is taken over all couplings $\pi$ of $\mu$ and $\nu$. This means that transporting
mass within a cell of the clustering cannot transform a distribution $\mu$ into another one, $\nu$,
that is \emph{far} in $\EMD_\dist$, as long as the cell is small with respect to the ambient
metric $\dist$.\\

\noindent
\emph{Reject bad clusterings.} By itself, using $\EMD$ does not solve our problems, since
there is no guarantee that the cells are small. One solution is to demand that the confused
collector use a low-diameter clustering, but this is too strong: we may
not trust that the collector has acquiesced to our demands; low-diameter clusterings
of the entire (large) domain may be costly for the collector to compute; and low-diameter
clusterings do not depend on the input probability distribution -- the collector may itself have
learned the clustering from the input distribution, and we may be satisfied with a collector who
makes poor distinctions on low-probability elements.  Instead, we allow the tester to reject
clusterings not belonging to a defined set of ``good'' clusterings, while requiring it to succeed on
any clustering that passes the test.\\

\noindent
\emph{Restrict the clusterings.} Some types of clusterings make the task of rejecting bad
clusterings infeasible. Using label queries, the tester can learn the entire
clustering by querying every point in the domain, but we want sublinear query complexity. If, for
example, the clusters are not even required to be \emph{connected}, then it might become infeasible
to efficiently detect bad clusterings. But, by making reasonable assumptions on the clusterings
(including, say, that the cells of the clustering are connected or convex), it becomes feasible to
check whether the given clustering is suitable.\\

This leads to our definition of distribution testing. For convenience, we assume that the 
metric $\dist(\cdot,\cdot)$ is normalized, meaning it has diameter 1. 

\newcommand{\ClusterReject}{\textsf{CLUSTER-REJECT}\xspace}
\newcommand{\clusterreject}{\textsf{cluster-reject}\xspace}
\newcommand{\Accept}{\textsf{ACCEPT}\xspace}
\newcommand{\accept}{\textsf{accept}\xspace}
\newcommand{\Reject}{\textsf{REJECT}\xspace}
\newcommand{\reject}{\textsf{reject}\xspace}

\newcommand{\CTest}{\mathsf{CT}}
\newcommand{\DTest}{\mathsf{DT}}
\begin{definition}[Testing with an Adversarial Confused Collector]
\label{def:adversarial-testing}
Fix a domain $\cX$, let $\cP$ be a property of distributions over $\cX$, and let $d(\cdot,\cdot)$ be
a metric on probability distributions over $\cX$. Let $\cU$ be a class (\emph{universe}) of
clusterings, and let $\cD$ be a set of ``good'' clustering-distribution pairs $((\Gamma,\rep), \mu)$
where $(\Gamma, \rep) \in \cU$ and $\mu$ is an arbitrary distribution. Let $\cA$ be an algorithm
with clustered-sample and label oracle access to the input $((\Gamma, \rep), \mu)$, whose possible
outputs are \Accept, \Reject, and \ClusterReject.  Then we say $\cA$ is a
\emph{$(\cU,\cD,\epsilon,\delta)$-distribution tester for $\cP$ under metric $d$} if it satisfies
the following on every input $((\Gamma,\rep), \mu)$ where $(\Gamma,\rep) \in \cU$:
\begin{enumerate}[leftmargin=1.3em]
\item If $((\Gamma,\rep), \mu) \in \cD$ (\ie $(\Gamma,\rep)$ is a ``good'' clustering for $\mu$)
then the output of $\cA$ is \ClusterReject with probability at most $\delta$;
\item If $\mu \in \cP$ then the output is in $\{\Accept, \ClusterReject\}$
  with probability at least $1-\delta$;
\item If $d(\mu,\cP) > \epsilon$ then the output is in $\{\Reject, \ClusterReject\}$
  with probability at least $1-\delta$,
\end{enumerate}
where the probabilities are over the randomness of $\cA$ and the responses to the
oracle calls.
Note that this definition permits standard boosting: an algorithm satisfying the above
conditions, with (say) $\delta \le 1/6$, may be boosted to any $\delta'$ by taking a majority
vote of $\Theta(\log(1/\delta'))$ runs.
\end{definition}

Different problems may lead to different notions of ``good'' clustering. In this paper the ``good''
clusters will be those belonging to a chosen subclass $\cG \subseteq \cU$ which also satisfy a
``high-probability of low-diameter'' (HPLD) condition, which gives the following instance of the
testing task: 

\begin{definition}[Diameter-Guarded Testing]
\label{def:diameter-guarded}
For a finite unit-diameter metric space $(\cX, \dist)$, property $\cP$ of distributions over $\cX$,
universe $\cU$ of clusterings, and subset $\cG \subseteq \cU$, we say that an algorithm $\cA$ is a
\emph{$(\cU,\cG,\Delta)$-diameter-guarded $(\epsilon,\delta)$-tester} for $\cP$ if it is a $(\cU,
\cG_{\Delta,\epsilon}, \epsilon, \delta)$-distribution tester under $\EMD_\dist$, where
\[
\cG_{\Delta,\epsilon} \define \left\{
    \left((\Gamma, \rep),\mu\right) \;|\; (\Gamma, \rep) \in \cG,\;
      \Pru{\bm{x} \sim \mu}{\diam_\dist(\Gamma_{\gamma(\bm x)}) > \Delta}
                                \le c \cdot \epsilon \right\} 
\]
with $c \define \frac{1}{384\ln(24)}$ (defined this way for convenience in the analysis).
If no $\delta$ is specified, it is assumed to be $\delta = 1/6$.
\end{definition}

The above definition does not capture the equivalence testing task, where there are two input
distributions, but the adaptation is straightforward (see \cref{remark:equivalence-good}).
We have chosen ``good'' to mean HPLD ($\Pru{x \sim \mu}{\diam_\dist(\Gamma_\gamma{\bm{x})}) >
\Delta} \leq c\cdot\epsilon$) over two possible alternatives:

The first alternative is the less permissive condition that $\Gamma$ is a low-diameter clustering,
\ie $\diam_\dist(\Gamma_i) < \Delta$ for all clusters $\Gamma_i$. This makes the algorithm's job
easier but allows it to \textsf{CLUSTER-REJECT} in many cases where we wish for it to work: if the
input $\mu$ is concentrated on a small fraction of the domain, it should be acceptable for the
clustering to be coarse elsewhere.

The second alternative is the more permissive condition that $\Gamma$ has low \emph{average}
diameter, \ie $\Exu{\bm{x} \sim \mu}{\diam_\dist(\Gamma_{\gamma(\bm x)})} \leq \Delta$. Some of our
applications can be strengthened in this way (essentially when the diameter of any cluster in $\cU$
can be efficiently estimated\footnote{The key modification would be in the diameter tester from
\cref{lemma:clustering-tester}.}), but we chose HPLD as an option that is both reasonable under our
motivation and feasible for many applications.

\subsubsection{Results}

We prove a general lemma (\cref{lemma:metric-general}) that reduces the complexity of identity and
equivalence testing to the query complexity of two subroutines:
\begin{enumerate}[itemsep=0pt,leftmargin=0pt]
\item[] \emph{Cell discovery.} Given a representative point $r$, output an approximation of
its cell $\Gamma_{\gamma(r)}$.
\item[] \emph{Cell rejection.} Given a representative point $r$ and two parameters
$t_1 < t_2$, distinguish between the case where the cell $\Gamma_{\gamma(r)}$ has diameter at most
$t_1$, or at least $t_2$.
\end{enumerate}
The query complexities of these subroutines depend on the geometry of the underlying metric space as
well as the ``universe'' $\cU$ of clusterings that the confused collector is promised to provide,
and the ``good'' clustering geometry $\cG$ that the tester is required to accept. Instead of stating
the general result, we summarize the main applications to the following classes of clusterings:
\begin{enumerate}[itemsep=0pt]
\item[$\CONN$] is the class of \emph{connected} clusterings of the hypergrid $[n]^d$. A
clustering is connected if each of its cells is a connected subset of the standard hypergrid graph
on vertices $[n]^d$, where $x,y \in [n]^d$ have an edge when $\|x-y\|_1 = 1$.
\item[$\CC$] is the class of \emph{connected convex} clusterings of the hypergrid $[n]^d$, where each
cell is both \emph{connected} and \emph{convex}. A subset $S \subseteq [n]^d$ is convex if it is
equal to its convex hull.
\item[$\BOX$] is the class of \emph{axis-aligned box} clusterings of the hypergrid $[n]^d$, where
each cell is an axis-aligned box. Note that $\BOX \subseteq \CC$, and that decision-tree
clusterings are a subclass of $\BOX$.
\item[$\CONV_\delta$] is the class of \emph{$\delta$-convex} clusterings of the continuous cube
$[0,1]^d$, where each cell is a convex set that is guaranteed to contain the $\ell_2$-ball of
radius $\delta$ around its representative point.
\item[$\BOX_\delta$] is the class of clusterings of the continuous cube $[0,1]^d$, where each cell is
an axis-aligned box, and is guaranteed to contain the $\ell_2$-ball of radius $\delta$ around
its representative point.
\end{enumerate}
The most difficult instances of $(\cU,\cG,\Delta)$-diameter-guarded testing are when $\cU$
is as inclusive as possible (the universe of clusterings is large) and $\cG \subseteq \cU$ is as
inclusive as possible (the algorithm is required to function as a distribution tester on a wider
class of inputs without cluster-rejecting).  See \cref{table:results} for a summary of the
quantitative bounds. \\

\noindent
\textbf{Connected clusters.}
The most difficult case for testers on domain $[n]^d$ is when the clusters are only promised to be
connected, and it must accept any connected HPLD clustering. Plugging in a simple cell rejection
subroutine, we show that $m(\epsilon) = 2^{O(d)} \cdot \widetilde O(\epsilon^{-\max\{2,
\frac{2d}{3}\}})$ samples and $q(\epsilon) = O(d\epsilon^{d-2} n^{d-1})$ queries suffice for
equivalence testing, which is sublinear in the domain size $n^d$. Since we are testing under $\EMD$,
it may be most natural to consider small $\epsilon = O(n^{-c})$ for constants $c > 0$.  It may seem
odd that the number of queries \emph{decreases} as $\epsilon \to 0$ (when $d > 2$), but this is
simply because the query complexity is a balance of the \emph{number} of cells that must be checked
(which increases as $\epsilon \to 0$) and the \emph{size} of the cells that must be accepted (which
decreases).

The connectivity promise seems too weak to get sublinear query complexity for cell discovery, so the
best sublinear-query result that we get is by reducing identity testing to equivalence
testing.  One may always reduce identity testing to equivalence testing in the confused collector
model, by incurring an additive $m(\epsilon)$ label query cost to simulate clustered-sample requests
to the known target distribution $\nu$. \\

\noindent
\textbf{Convex clusters.}
A reasonable, yet still very weak, condition to place on the confused collector is that its
clustering has convex cells (and on the hypergrid $[n]^d$, we also keep the condition that it is
connected). In the hypergrid $[n]^d$, we are now able to get sublinear query complexity for cell
discovery, giving a better bound on identity testing than for merely connected cells. In the
continuous domain, we use subroutines for convex optimization with membership oracles \cite{LSV20}
to implement the cell rejection procedure to give results for testing equivalence. \\

\begin{figure}[t]
\begin{NiceTabular}{c|c|c|l}[cell-space-limits=0.2em]
 $(\cU, \cG, \Delta)$ & \textbf{Identity} & \textbf{Equivalence} & \textbf{Theorem} \\
Domain $[n]^d$
 & $m(\epsilon) = \widetilde O(\epsilon^{-\max\{2,\frac{d}{2}\}})$
 & $m(\epsilon) = \widetilde O(\epsilon^{-\max\{2,\frac{2d}{3}\}})$
 & \\
\hline
$(\CONN,\CONN, \epsilon/8d^{1/p})$
  & --
  & $q(\epsilon) = O(\epsilon^{d-2} n^{d-1})$
  & \cref{thm:intro-c-c}\\
\hline
$(\CC, \CC, \epsilon/8d^{1/p})$
  & $q(\epsilon) = O(m(\epsilon) \cdot n^{d-1})$
  & $q(\epsilon) = O(\epsilon^{d-2} n^{d-1})$
  & \cref{thm:intro-cc-cc} \\
\hline
$(\BOX, \BOX, \epsilon/8)$
  & $q(\epsilon) = O( m(\epsilon) \cdot \log n )$
  & $q(\epsilon) = O( \frac{1}{\epsilon} \cdot \log n )$
  & \cref{thm:intro-bb} \\
\hline
\Block{1-1}{$(\CC, \BOX, \epsilon/8)$ \\ (domain $[n]^2$)}
  &  --
  & $q(\epsilon) = O( \frac{1}{\epsilon} \cdot \log n )$
  & \cref{thm:intro-cc-b} \\
\hline
\hline
Domain $[0,1]^d$ & \textbf{Uniformity} & \textbf{Equivalence} & \\
\hline
$(\CONV_\delta, \BOX_\delta, \epsilon/16)$
  & $q(\epsilon) = O( m(\epsilon) \cdot \poly\log\frac{1}{\delta} )$
  & $q(\epsilon) = \widetilde O( \frac{1}{\epsilon} \cdot \poly\log\frac{1}{\delta} )$
  & \cref{thm:intro-cv-b-unif} \\
\hline
$(\CONV_\delta, \CONV_\delta, \epsilon/16 d^{1/p})$
  & --
  & $q(\epsilon) = \widetilde O(\frac{1}{\epsilon}\cdot \poly \log \frac{1}{\delta})$
  & \cref{thm:intro-cv-cv} \\
\end{NiceTabular}
\captionof{table}{Summary of applications in Part I. The bounds are stated for $d =
O(1)$, and the normalized $\ell_p$ metric with $p \geq 1$.  The
sample complexity is given by $m(\epsilon)$ at the top of each column, and
$q(\epsilon)$ is the query complexity. The main difference between each setting is the promised cell
geometry $\cU$, the ``good'' cell geometry $\cG$, the  ``good'' diameter $\Delta$
(with higher values requiring the algorithm to accept the clustering more often),
and the query complexity, which is sublinear.}
\label{table:results}
\end{figure}

\noindent
\textbf{Axis-aligned box clusters.} One of our main motivations was for clusterings computed by
decision trees with nodes of the form ``$x_i < t$?'', which are a special case of axis-aligned box
clusterings; our most interesting results in Part I are for this class. The easiest result is when
the cells are \emph{promised} to be axis-aligned boxes in $[n]^d$: using binary search,
the cells can be learned exactly. More interesting is when the clusterings are \emph{not} promised
to be boxes, but we only demand that the tester pass the box clusters, so that it detects if the
clustering is too far away from a box clustering.  In the continuous domain $[0,1]^d$, we show that,
as long as the clustering is promised to be \emph{convex}, the algorithm can either learn a good
enough approximation of the cells to test uniformity, or it can reject the clustering. We get a
similar result for the discrete domain $[n]^2$. Crucially, the algorithm is not required to exactly
learn the cells, or even verify that they are exactly boxes, which would be expensive.\\

\noindent
\textbf{Threshold metrics.}
A simple but helpful example to understand the model, is that we can apply our general result using
\emph{threshold $\ell_p$ metrics}, where $\dist(x,y) = \max\{\ell_p(x,y), t\}$ for some threshold
$t$. This allows to interpolate between the $\EMD_{\ell_p}$ metric and TV distance (see
\cref{thm:capped-bb}).\\

\noindent
\textbf{Techniques.}
We consider the problem itself to be the main contribution of Part I since it was not clear to us in
advance that adversarial clusterings allow for any interesting algorithmic tasks, and the problem
requires careful definitions. First, since we include results for both identity and uniformity testing, let
us note that the standard reduction from identity testing to uniformity testing \cite{Gol16} does
\emph{not} hold in the confused collector model\footnote{A ``reduction'' in this model would require
transforming one (unknown) clustering of the domain into a clustering of another domain (with
simulated label queries), and samples from the clustered distribution into samples from a different
clustered distribution. It is possible to transform the domain in a way that follows the reduction
of \cite{Gol16} but, due to the clustering, it does not actually change the observed samples at
all.}.

Our algorithms will require a variety of techniques.  The first step of the algorithm is to verify
the HPLD condition, which is done by repeatedly sampling a point $x \sim \mu$ from the
clustered-sample oracle, and using label queries to check if the diameter of its (unknown) cell is
large. This depends on the geometry of the cells: we use either ad-hoc algorithms or, for convex
cells in $[0,1]^d$, an application of convex optimization \cite{LSV20}.

The second step of the identity and uniformity testing algorithms is to sample points $x \sim \mu$
from the clustered-sample oracle and discover an approximation of its cluster $\Gamma_{\gamma(x)}$.
We may then simulate a sample from an auxiliary $\mu^\bullet$ distribution, by resampling from
$\Gamma_{\gamma(x)}$ according to the \emph{target} (known) distribution $\nu$. We then compare the
target $\nu$ to the auxiliary distribution $\mu^\bullet$ using an EMD tester similar to that of
\cite{DNNR11}, which reduces to standard identity and equivalence testing bounds for TV distance
\cite{CDVV14,VV17}. Cell discovery depends on the geometry of both the \emph{promised} clusterings,
and of the \emph{good} clusterings; the most illustrative examples are when the universe $\cU$ is
promised to be either connected and convex in the grid $[n]^2$, or convex in $[0,1]^d$, and the
``good'' clusterings $\cG$ are the decision trees (axis-aligned boxes). In each case, we can provide
sample access to the auxiliary distribution $\mu^\bullet$ \emph{without} exactly learning the cells,
by making use of the algorithm's ability to output $\ClusterReject$ if the cell is too far away from
being a box. These examples show, in particular, that the cluster-rejection ability of the algorithm
can be more powerful and interesting than simply testing the HPLD condition.

\subsection{Part II: Random Clustering}

If the confused collector has a \emph{random} clustering instead of an adversarial clustering, one
hopes to improve upon the results in Part I, and we show that this is true for certain random
clusterings. We now wish for the tester to be correct with high probability, over both
the samples \emph{and} the clustering. Specifically, for a property $\cP$, input $\mu$, parameter
$\delta$, and distribution $\cD$ over the class $\cU$ of clusterings, we require that the testing
algorithm $A$ satisfies:
\begin{itemize}[itemsep=0pt]
\item If $\mu \in \cP$, then $\Pru{(\bm{\Gamma},\bm{\rep}) \sim \cD;\; A}{A \text{ accepts}} \geq
1-\delta$; and
\item If $\mu$ is $\epsilon$-far from $\cP$, then $\Pru{(\bm{\Gamma},\bm{\rep}) \sim \cD;\; A}{A
\text{ rejects}} \geq 1-\delta$.
\end{itemize}
It may sometimes be natural to allow the algorithm to $\ClusterReject$, but this will not be
necessary at present. When the algorithm is not allowed to $\ClusterReject$, standard probability
boosting techniques \emph{do not work}: the algorithm has no control over the clustering, which is
fixed, and therefore the error probability depends on the distribution over clusterings.

We will focus on random clusterings defined as follows.
The domain is $[n]$, which we think of as vertices of a path (or cycle) $G = ([n],E)$.
For parameter $\rho \in (0,1]$, which we call the \emph{resolution}, the distribution $\cU_\rho$
over clusterings is defined by taking a random subgraph $\bm H$ of $G$ where each edge is deleted
independently with probability $\rho$, and the vertices are clustered by their connected components
in $\bm H$.  In other words, each consecutive pair of the domain is ``separated'' into different
cells with probability $\rho$, so that the resolution $\rho$ controls the granularity of the
clustering. 

Our motivations for this choice of random clustering are, first, that it captures a type of
``environmental randomness'' in the collection of samples, like the example where fossils from 
different years can be distinguished only if a random geological event occurred between those years.
We think of ``environmental randomness'' as being uncorrelated with the input probability
distribution $\mu$, whereas other natural random clusterings (\eg a machine-learning classifier) may
depend on $\mu$.  The one-dimensional clusterings we study are a simple and natural starting point for
understanding how environmental randomness affects distribution testing tasks.

The second motivation is that these random clusterings occur in the study of certain
distribution-free property testing problems. These problems are outside the scope of the current
paper and will be explained formally in future work\footnote{A partial treatment of this occurrence
is given in a preprint \cite{FHparity} which contains some results from  Part II of this paper, and
will be elaborated upon in future work.}, but we give a simplified description in
\cref{remark:dist-free-testing}.

\subsubsection{Results}

We give results for testing \emph{uniformity} of distributions under these random clusters, where
the algorithm should accept the uniform distribution over $[n]$ and reject any distribution over
$[n]$ that is $\epsilon$-far from uniform in TV distance. From \cref{thm:intro-bb} in Part I, we
obtain a tester using $\widetilde O(\tfrac{1}{\epsilon^2})$ samples and $\widetilde
O(\tfrac{1}{\epsilon^2})$ queries under the $\EMD$ distance, but we now hope to use the TV distance
(which corresponds roughly to testing under $\EMD$ with parameter $\approx \epsilon/n$). The optimal
sample complexity for testing uniformity in the standard distribution testing model is $\Theta(\sqrt
n /\epsilon^2)$ \cite{Pan08,VV17}.\\

\noindent
\emph{Na\"ive benchmark algorithm.}
A reasonable algorithm that one might first propose is as follows.  Suppose that the
tester uses queries to exactly learn the clustering $(\Gamma, \rep)$ of the domain $[n]$. It may
then define $\nu^*$ as the distribution over the representatives of the clusters $r_i \define
\rep(\Gamma_i)$ obtained by sampling $x$ according to the uniform distribution over $[n]$ and then
taking the representative of its cluster, $\rep(\Gamma_{\gamma(x)})$. Define $\mu^*$ as the
distribution received from the confused collector, so that $\mu^*$ is the distribution over
representatives $\rep(\Gamma_{\gamma(x)})$ when $x \sim \mu$ is sampled from the unknown input
distribution. Then the algorithm runs an identity test on $\mu^*$ and $\nu^*$.

A back-of-the-envelope calculation of the complexity of this algorithm is as follows.  The expected
number of clusters is $\Theta(\rho n)$, so we require at most $O(\rho n \log n)$ queries to learn
them with binary search, and the identity test will be performed on a domain of size $O(\rho n)$.
When $\mu$ goes from the original distribution to the clustered distribution $\mu^*$, it begins with
TV distance $\epsilon$ from uniform, and might shrink to TV distance $\rho \epsilon$ from the target
$\nu^*$, forcing us to set the distance parameter for the identity test no larger than $\rho
\epsilon$. (An example of this shrinkage is obtained by adding $\epsilon/n$ mass to each odd element
of the uniform distribution, and subtracting $\epsilon/n$ from each even element, so that the $\pm
\epsilon/n$ perturbations in each cluster balance out, except for an expected surplus of $\approx
\pm \epsilon/2n$ for each of the $\rho n$ clusters.) Plugging in the identity testing bound, we get
sample complexity $O\left(\frac{\sqrt{\rho n}}{\rho^2 \epsilon^2}\right) = O\left(\frac{\sqrt
n}{\rho^{3/2} \epsilon^2}\right)$ and query complexity $O(\rho n \log n)$.

Our first result, which is the main technical challenge of this paper, shows that we can in fact get
the same bound on the sample complexity while using \emph{zero} queries. This is important because,
for some of the ``environmental randomness'' and property testing motivations we are interested in,
it may not be possible to pose queries to the confused collector.

\begin{theorem}[Main theorem; see \cref{thm:confused-collector-main} for specific requirements on $\rho, \epsilon$.]
\label{thm:intro-confused-collector-main}
Let $G$ be a path or cycle on $n$ vertices and let $\rho, \epsilon \in (0,1]$ satisfy $\rho
\geq \widetilde \Omega(n^{-1/5} \epsilon^{-4/5})$. Then $\epsilon$-testing uniformity under TV
distance with the confused collector can be done using $\widetilde
O\left(\frac{\sqrt n}{\rho^{3/2} \epsilon^2}\right)$ samples and zero queries.
\end{theorem}

Since the algorithm has no control over the clustering, improvements to the error probability can be
achieved by improving the resolution parameter.  The proof is summarized in
\cref{section:intro-proof-overview} below. It does not use a reduction to identity testing, and is
instead a direct analysis of a generalization of the standard uniformity tester, which requires a
new technical lemma on the concentration of random quadratic forms.  The cycle has a cleaner
analysis, but the path and cycle cases do not appear to directly reduce to one another, which is why
both are included in the theorem.

With the sample complexity of the natural benchmark algorithm matched by a zero-query algorithm, one
may wonder if queries can still be helpful, but it is not clear how to improve the algorithm or
analysis of the benchmark: using the instance-optimal tester of \cite{VV17} does not immediately
improve the analysis.  However, we show it is indeed possible to improve the sample complexity when
queries are allowed.  The algorithm is simple but, to us, much less natural than the benchmark.  It
is the same as the benchmark algorithm except that it tests identity \emph{only on the singleton
clusters} (\ie clusters of size 1), crucially using the instance-optimal algorithm of \cite{VV17}.

\begin{theorem}
\label{thm:intro-ignorance-is-bliss}
Let $G$ be a path or cycle on $n$ vertices, and let $\rho, \epsilon$ satisfy $\rho \geq
\Omega((\epsilon n)^{-1/4})$. Then testing uniformity under TV distance with the
confused collector requires $O\left(\frac{\sqrt n}{\rho \epsilon^2}\right)$ samples and
$O(\rho n \log n)$ queries.
\end{theorem}
The natural parameter regime is where $\rho = o(1)$; for the allowed constants $\delta$ in the
theorem, setting $\rho = n^{-\delta}$ produces a sublinear $O(n^{1-\delta} \log n)$ query
complexity.  



\subsubsection{Proof Overview.}
\label{section:intro-proof-overview}
We now give an overview of the proof of our main technical \cref{thm:intro-confused-collector-main},
which shows that there is a zero-query tester that matches the sample complexity of the
benchmark algorithm.  Let us review the standard uniformity tester \cite{GR00,DGPP19} (see also
\cite{Can22}).  Let $\mu$ be the input distribution over $[n]$. For a sample $S$ of size $m$, let
$X_i$ be the multiplicity of element $i$ in $S$. The tester counts the number of ``collisions'' in
the sample: it computes $Y \define \frac{1}{m(m-1)} \sum_{i=1}^n X_i (X_i - 1)$, and rejects if this
is too large. This works because $\Ex{\bm{Y}} = \mu^\top \mu = \|\mu \|_2^2$, which is large when $\mu$ is
far from uniform. Now we describe the zero-query tester for the confused collector.  For input
distribution $\mu$ on domain $\bZ_n$ (which are the vertices of the path or cycle), we use the
standard Poissonization technique, so that element $j$ occurs in the sample with multiplicity $\bm
T_j \sim \Poi(m \mu_j)$ independently of the other elements. Recall that $\bm H$ is the random
subgraph of the path or cycle that determines the clusters of the domain, and redefine $\bm X_i$ as
the number of sample points contained in the $i^{th}$ connected component of $\bm H$, which the
tester cannot distinguish: the $\bm X_i$ variables remain Poisson, but they are not independent.
The tester computes a ``collision count'', as in the standard algorithm:
\[
  \bm Y \define \frac{1}{m^2} \sum_i \bm X_i (\bm X_i - 1)
    = \frac{1}{m^2}\left(\bm T^\top \bm \Phi \bm T  - \|\bm T\|_1 \right) \,,
\]
where $\bm \Phi$ is the random Boolean matrix with $\bm \Phi_{i,j} = 1$ iff $i,j$ belong to the same
connected component of $\bm H$. The expected value is $\Ex{\bm Y} = \mu^\top \phi \mu$ where
$\phi = \Ex{\bm \Phi}$, and we show that this is again large when $\mu$ is far from uniform, using
spectral analysis of the matrix $\phi$ which is either Toeplitz (for paths) or circulant (for
cycles).

To complete the analysis, we require concentration of measure for the random quadratic form $\bm
T^\top \bm \Phi \bm T - \|\bm T\|_1$. This is similar to Hanson-Wright inequalities, except that
Hanson-Wright inequalities apply to fixed matrices $\Phi$ whereas our matrix $\bm \Phi$ is random.
We prove the following concentration inequality in terms of a concentration measure
$\RelativeConcentration{\mu}$, which roughly satisfies $\RelativeConcentration{\mu} \lesssim
\|\mu\|_\infty / \rho$ (but is less stringent in general) and which the algorithm can separately
test:

\begin{lemma}[Informal; see \cref{lemma:concentration-inequality-t}]
    \label{lemma:intro-informal-concentration}
    Let $\delta \in (0, 1)$ be a constant.
    Let $\mu$ be a
    probability distribution over $\bZ_n$, and
    suppose $\rho \ge \Omega(n^{-\delta})$ and $m \le \poly(n)$. Then for all $\tau > 0$,
    \[
        \Pr{\abs*{\bm{Y} - \Ex{\bm{Y}}} \ge \tau}
        \le \frac{\|\mu\|_2^2}{\rho \tau^2}
                \cdot \max\left\{ \RelativeConcentration{\mu}, 1/m \right\}^2
                \cdot O(\log^2 n) \,.
    \]
\end{lemma}

\subsection{Discussion and Open Problems}
\label{section:discussion}
\label{remark:dist-free-testing}

\noindent
\textbf{Related work.} \cite{FKKT21} studied statistical learning (\eg Gaussian mean
estimation) in a model with similar motivation to ours. The difference is that \cite{FKKT21} chooses
a random clustering for each sample point independently and provides an explicit representation of
the cell containing each point, whereas in our model the sample points are all labeled by the
\emph{same} clustering, the algorithm is not given the cell explicitly, and the algorithm can make
label queries. If we suppose that the input distribution $\mu$ is held by \emph{multiple} confused
collectors, and each sample point is labeled by a \emph{random} collector, then we nearly recover
the model of \cite{FKKT21}, except that we allow label queries and do not receive explicit
representations of the cells. In this interpretation, \cite{FKKT21} require that the collectors
jointly hold an ``information-preserving'' clustering (which approximately preserves TV distance
between any pair of distributions), which is unnecessary in our model.

Our model also shares some conceptual similarities with the ``huge object'' model
\cite{GR22}, where the algorithm must test properties of a distribution over $\zo^n$ by taking
samples and, for each sampled element $x \in \zo^n$, querying a subset of its bits. In both models,
the algorithm has incomplete information about the sample, and the goal is to test with respect to
the earth-mover distance. The difference is that, in the huge object model, the algorithm has incomplete
information about the sample because the objects are too large to observe entirely, 
but perfect knowledge could be obtained if enough queries are used. In the confused collector
model, the incomplete information is due to imprecise classification of the sample and may not be
possible to acquire.

We remark also that several distribution testing algorithms in the literature use random clustering
as a step in the tester, including algorithms for testing in the huge objects model \cite{CFGMS22}
(which assign each $x \in \zo^n$ to a cluster determined by a random subset of bits); and algorithms
for testing with communication constraints, which compress the domain to save space (see \eg
\cite{ACH+20} and the survey \cite{Can20}). For example, \cite{ACH+20} gives results for testing
identity on unstructured domains, using random clusterings produced with few bits of randomness
shared between parties. In these settings, the algorithm has control over the clustering, whereas we
are interested in the case where the clustering is not controlled.  \\

\noindent
\textbf{Questions.} We would like to know more about how to manage \emph{multiple} confused
collectors. One may consider a setting where the input $\mu$ is held by multiple collectors (as
above), or where the collectors hold separate inputs $\mu_1, \dotsc, \mu_t$ (as studied in
\cite{LRR13,LRR14,AS20} for the standard model). Each sample point could be drawn from a random
collector, or a chosen collector; or one might pay cost $k$ to get a single sample point labeled by
$k$ of the collectors. How efficiently can one detect and select collectors with complementary
expertise, and combine their expertise?

In Part I, we reduce identity or equivalence testing in the clustered domain, to the identity or
equivalence testing in the standard model, but this is not as easy for properties like monotonicity
(\eg \cite{BKR04,AGPRY19}) or $k$-histograms (\eg \cite{ILR12,Can16,CDGR18,CDKL22}), which are not
preserved by clustering; it would be interesting to study these.

An open problem of \cite{DNNR11} is to find tight bounds on the sample complexity for estimating
$\EMD$ between distributions on $[0,1]^d$. This appears to be open still, and it would be helpful in
the confused collector model to have optimal bounds on tolerant testing under $\EMD$.

The random clustering model in Part II is tailored to the path and cycle, and
may not be sensible for more general graphs.  It would be interesting to know which natural random
clusterings of graphs allow for efficient zero-query algorithms, as in
\cref{thm:intro-confused-collector-main}. We also wonder what properties of distributions,
beyond uniformity, admit testers (with and without queries) under TV distance; \eg is identity
testing possible for some non-trivial class of target distributions $\nu$?

Finally, it would be interesting to investigate instance-optimal identity testing \cite{VV17}
(see also \cite{DK16,BCG19}) in the confused collector model, since cell discovery etc. can
be tailored to the known distribution.\\

\noindent
\textbf{Distribution-free property testing.}
One of our original motivations for studying the confused collector model is its relation to some
distribution-free property testing problems. This will be the subject of future work and is outside
the scope of this paper, but the connection boils down to this:

In distribution-free testing of functions $f : \cX \to \zo$, it becomes necessary to test a joint
property of the input distribution and the set $f^{-1}(1)$, which we can think of as a union of
connected components; if $f : \bR \to \zo$, then $f^{-1}(1)$ is a union of intervals.  But, given
two samples $x,y \in f^{-1}(1)$, it is not possible to know whether they came from the same or
different intervals, unless another sample $z$ occurs in $f^{-1}(0)$ between $x,y$. Testing a joint
property of the input distribution with $f^{-1}(1)$ therefore must be done with the tester seeing
only a ``coarsening'' of $f^{-1}(1)$. The confused collector model allows us to study this
phenomenon in a simpler setting from first principles, and there is in fact a formal connection
between the problem just described and the random clusterings in Part II of this paper, which will
be elaborated in future work.  \\

\section{Preliminaries}
\subsection{General Notation}
\label{section:framework-notation}

We will often use the notations $\gequestion$, \;$\lequestion$, \;$\eqquestion$\;, etc., within
proofs, when stating an (in)equality that will be established later on in the proof.

We write $\log x$ for the natural logarithm of $x$. $\bN$ denotes the set of positive
integers, \ie it does not include 0. For any $x$, we write $\bZ_{> x}$ for the set of integers
greater than $x$, and $\bZ_{< x}, \bZ_{\geq x}, \bZ_{\leq x}$ are defined similarly.
We denote random variables by boldface symbols, \eg $\bm{X}$.
We write $x = a \pm b$ as a shorthand for $a-b \le x \le a+b$. For an event $E$, $\ind{E}$ is the
indicator variable for $E$, which takes value 1 if and only if $E$ occurs.

For a distance metric $\dist(\cdot, \cdot)$ on a domain $\cX$, an element $y \in \cX$, and a set $X
\subseteq \cX$, we write
\[
  \dist( y, X ) \define \inf_{x \in X} \dist(y,x) \,.
\]
For a probability distribution $\mu$ over (countable) domain $\cX$ and any set $S \subseteq \cX$, we
write $\mu[S] = \sum_{x \in S} \mu(x)$.

For two probability distributions, $\mu$ over domain $\cX$ and $\nu$ over domain $\cY$, we write
$\Pi(\mu, \nu)$ for the set of all couplings between $\mu$ and $\nu$, namely joint distributions
$\pi$ over domain $\cX \times \cY$ whose marginals are (respectively) $\mu$ and $\nu$.

\begin{definition}[Earth Mover's Distance]
    \label{def:emd}
    Let $(\cX, \dist)$ be a metric space, and let $\mu, \nu$ be probability distributions over
    $\cX$. The earth mover's distance (EMD) between $\mu$ and $\nu$ is
    \[
        \EMD_\dist(\mu, \nu) \define
            \inf_{\pi \in \Pi(\mu, \nu)} \Exu{(\bm{x}, \bm{y}) \sim \pi}{\dist(\bm{x}, \bm{y})} \,.
    \]
\end{definition}

\begin{remark}
    When $(\cX, \dist)$ is a finite metric space, the infimum in \cref{def:emd} is attained, since
    the EMD may be equivalently formulated as a linear program whose feasible region is bounded (and
    nonempty). Therefore we may replace the infimum with a minimum.
\end{remark}

\subsection{Notation for Clusterings}

We require some specialized notation for clusterings. Let $(\Gamma, \rep)$ be a clustering of domain
$\cX$, where $\rep : \cX \to \bN$.

\begin{definition}[Induced Distributions]
\label{def:induced-distributions}
Let $\mu$ be any probability distribution over $\cX$ and let $(\Gamma, \rep)$ be a clustering of
$\cX$. We will write $\induced{\mu}{\Gamma}$ for the distribution over the range of $\rep$
\emph{induced} by $\mu$, which is defined as the distribution of the random variable
$\rep(\gamma(\bm{x}))$ when $\bm{x} \sim \mu$. Observe that for a representative point $r =
\rep(\Gamma_i)$, $\induced{\mu}{\Gamma}(r) = \mu[\Gamma_i]$.

In Part II, the specific representatives of the clusters are not important, so we will abuse
notation in the following way.  For any clustering $(\Gamma, \rep)$ with $\Gamma = \{ \Gamma_1,
\dotsc, \Gamma_k \}$, we write $\induced{\mu}{\Gamma}$ for the distribution over \emph{indices} of
cells, so that for each $i \in [k]$, $\induced{\mu}{\Gamma}(i) \define \mu[\Gamma_i]$ is the total
probability mass of cell $\Gamma_i$.
\end{definition}

\section{Part I: Adversarial Clustering}

\subsection{Cell Discovery and Cell Rejection with Queries}

Our testers will require subroutines for two tasks, \emph{cell-discovery} and \emph{cell-rejection}.
A cell-discovery algorithm will, given a membership oracle to an unknown cell $H$ and its
representative $h \in H$, output an approximation to the cell $H$. A cell-rejection algorithm will,
given the same inputs and diameter thresholds $t_1 < t_2$, reject the cell $H$ if $\diam(H) > t_2$
and accept if $\diam(H) < t_1$.

\begin{definition}[Cell Discovery]
    \label{def:cell-discovery}
    Fix domain $\cX$, let $\cU$ be a class of clusterings and let $\cG \subseteq \cU$.
    Let $\alpha \in (0, 1]$, $\rho \in [0, 1)$, and let $\nu$ be a distribution over $\cX$. An
    \emph{$(\alpha,\nu,\rho)$-cell-discovery algorithm} $D$ for $(\cU,\cG)$ is a randomized
    algorithm defined as follows. The inputs to the
    algorithm $D$ are $(H,h)$ where $H \in \Gamma$ is a cell and $h \in H$ is its representative $h
    = \rep(H)$, according to an (unknown) clustering $(\Gamma, \rep) \in \cU$, where access to $H$
    is only given implicitly via the label queries $\LABEL(\Gamma, \rep)$ while $h \in H$ is given
    explicitly. The output of the algorithm is either $\clusterreject$, or a \emph{container} $C
    \subseteq \cX$, satisfying the following requirements with probability at least $1-\rho$ over
    the randomness of $D$:
    \begin{enumerate}
        \item If $(\Gamma, \rep) \in \cG$, then the output is a container $C$ such that $H \subseteq
            C$, and moreover $\nu(H) \ge \alpha \cdot \nu(C)$. When $\alpha = 1$, we strengthen this
            requirement to $H = C$.
        \item If the output is a container $C$, then it must satisfy $H \subseteq C$.
    \end{enumerate}
    The \emph{query cost} of $D$ on $(H,h)$, denoted $\cost_D(H,h)$, is the number of label queries
    used by $D$. We then define the query cost of \emph{$(\alpha,\nu,\rho)$-cell discovery for
    $(\cU, \cG)$} as
    \[
        \qcell(\cU, \cG, \alpha, \nu, \rho)
        \define \min_D \max_{(\Gamma, \rep) \in \cU} \max_{H \in \Gamma} \cost_D(H, \rep(H)) \,,
    \]
    where the minimum is over all cell-discovery algorithms $D$. When $\alpha=1$, we may remove the
    dependence on $\nu$ and write simply $\qcell(\cU, \cG, \rho)$. We also omit $\rho$ from the
    notation when $\rho=0$.
\end{definition}

\begin{definition}[Cell Rejection]
    \label{def:cell-rejection}
    Let $(\cX, \dist)$ be any metric space, let $\cU$ be a class of clusterings of $\cX$, and let
    $\cG \subseteq \cU$. For $t_1, t_2 \in \bR$ with $0 \le t_1 \le t_2$ and $\rho \in [0, 1)$, a
    $(t_1,t_2,\rho)$-\emph{cell-rejection algorithm} $R$ for $(\cU,\cG)$ is a randomized algorithm
    defined as follows. The inputs to the algorithm $R$
    are $(H,h)$ where $H \in \Gamma$ is a cell and $h \in H$ is its representative $h = \rep(H)$
    according to an (unknown) clustering $(\Gamma, \rep) \in \cU$, where access to $H$ is only given
    implicitly via the label queries $\LABEL(\Gamma, \rep)$ while $h \in H$ is given explicitly. The
    algorithm outputs either $\bot$, \textsf{accept}, or \textsf{reject}, satisfying the following
    conditions with probability at least $1-\rho$ over the randomness of $R$:
    \begin{enumerate}
        \item If $(\Gamma, \rep) \in \cG$ then the output is not $\bot$;
        \item If $\diam_\dist(H) \le t_1$ then the output is either $\bot$ or \textsf{accept}; and,
        \item If $\diam_\dist(H) > t_2$ then the output is either $\bot$ or \textsf{reject}.
    \end{enumerate}
    The query cost $\cost_R(H)$ is the number of membership queries made by $R$ on input $(H,h)$.
    Similar to above, we define
    \[
        \qreject(\cU,\cG,t_1,t_2,\rho)
        \define \min_R \max_{(\Gamma, \rep) \in \cU} \max_{H \in \Gamma} \cost_R(H, \rep(H)) \,,
    \]
    where the minimum is over all $(t_1,t_2,\rho)$-cell-rejection algorithms $R$ for $(\cU,\cG)$. As
    above, we omit $\rho$ from the notation when $\rho=0$.
\end{definition}

It is easy to verify the following relation between these quantities, which holds due to the
strengthened condition on cell discovery when $\alpha=1$.

\begin{fact}
\label{fact:rejection-to-discovery}
For any $\cU, \cG$ and $t_1, t_2$, it holds that $\qreject(\cU,\cG,t_1,t_2) \leq \qcell(\cU,\cG)$.
\end{fact}

\subsection{General Upper Bound}

We give a general result for identity and equivalence testing in an arbitrary (finite) metric space,
with the sample and query complexity expressed in terms of abstract bounds on the query complexity
of cell discovery and rejection, as well as the optimal size of a low-diameter partitioning of
space. We will apply the general result to specific cases in the next section.

The result requires a lemma about EMD. A similar but not equivalent statement appears in
\cite{IT03,DNNR11}, so for the sake of completeness we sketch a proof in
\cref{appendix:proof-of-tv-lemma}.

\begin{restatable}{lemma}{lemmaemdtvdiameter}
    \label{lemma:emd-tv-diameter}
    Let $(\cX, \dist)$ be a finite metric space, let $\mu, \nu$ be probability distributions over
    $\cX$, and let $\Gamma$ be a clustering of $\cX$. Then
    \[
        \EMD_\dist(\mu, \nu) \le
            \diam(\cX) \cdot \dist_\TV(\induced{\mu}{\Gamma}, \induced{\nu}{\Gamma})
            + \Exu{\bm{x} \sim \mu}{\diam(\Gamma_{\gamma(\bm{x})})} \,.
    \]
\end{restatable}

\noindent
Our general result will require a (meta-)algorithm for testing the diameter, described in the next
lemma. Let us first describe the parameters in the statement of this lemma.
\begin{itemize}[leftmargin=1em,itemsep=0pt]
\item $\epsilon$ is the EMD distance parameter that determines which distributions should be
rejected by the tester.
\item $\Delta$ is the diameter threshold we wish to use for diameter-guarded testing. We wish to
detect whether there are too many clusters with diameter larger than $\Delta$. We need $\Delta <
\epsilon/2$ because larger $\Delta$ incurs larger inaccuracy in EMD through clustering.
\item $\beta$ is a parameter which determines the ``gap'' between the acceptance and rejection
conditions for the clustering itself. For $\beta \in (0,1)$, write $c_\beta \define \frac{\beta}{96
\ln(24)}$, which is a constant that is convenient for the proof, and note that $c_{1/4}$ is the
constant $c$ from \cref{def:diameter-guarded}.
\end{itemize}

\begin{lemma}
    \label{lemma:clustering-tester}
    Let $(\cX, \dist)$ be a finite unit-diameter metric space, let $\cU$ be a class of clusterings,
    let $\cG \subseteq \cU$ be a subclass, and let $0 < 2\Delta < \epsilon < 1/2$, $\delta \in
    (0,1/2)$, and $\beta \in (0,1)$ satisfy $\Delta \le \beta \epsilon / 2$.
    Then there exists an algorithm \textsc{TestClustering} with
    sample complexity $m(\epsilon,\delta) = O(\log(1/\delta)/\beta\epsilon)$ and query complexity
    $q(\epsilon,\delta) = O\left( m(\epsilon,\delta) \cdot
    \qreject(\cU,\cG,\Delta,\beta\epsilon/2,\rho) \right)$, where $\rho = \frac{1}{24
    m(\epsilon,\delta)}$, satisfying the following: for every clustering $(\Gamma, \rep) \in \cU$
    and probability distribution $\mu$ over $\cX$,
    \begin{enumerate}
        \item If $(\Gamma, \rep) \in \cG$ and $\Pru{\bm{x} \sim \mu}{\diam(\Gamma_{\gamma(\bm{x})})
            > \Delta} \le c_\beta \cdot \epsilon$, \textsc{TestClustering} accepts with
            probability at least $1-\delta$.
        \item If $\Exu{\bm{x} \sim \mu}{\diam(\Gamma_{\gamma(\bm{x})})} > \beta \epsilon$,
            \textsc{TestClustering} rejects with probability at least $1-\delta$.
    \end{enumerate}
\end{lemma}
\begin{proof}
    We prove the claim for $\delta = 1/12$; the general case follows by standard boosting.
    Let $k \define \frac{2}{\beta}\ln(24)$. The algorithm makes $s = \ceil{k/\epsilon}$
    requests to $\SAMP(\Gamma, \rep, \mu)$ to obtain the representatives $r_1, \dotsc, r_s$ where
    for each $i \in [s]$, $r_i \define \rep(\gamma(x_i))$ is the representative of cell
    $\Gamma_{\gamma(x_i)}$ for a randomly drawn $x_i \sim \mu$.

    For each $r_i$, algorithm runs the $(\Delta, \beta\epsilon/2, \rho)$-cell rejection algorithm on
    input $(\Gamma_i, r_i)$. The algorithm outputs \emph{reject} if the cell rejection algorithm
    rejects or outputs $\bot$ on any of the $r_i$.

    The sample and query complexity claims are immediate, so now we show correctness. First, suppose
    $(\Gamma, \rep) \in \cG$ and $\Pru{\bm{x} \sim \mu}{\diam(\Gamma_{\gamma(\bm{x})}) > \Delta} \le
    c_\beta \cdot \epsilon$. Then the probability that some $x_i$ satisfies
    $\diam(\Gamma_{\gamma(x_i)}) > \Delta$ is, by the union bound, at most
    \[
        \left\lceil \frac{k}{\epsilon} \right\rceil \cdot c_\beta \cdot \epsilon
        \leq 2k \cdot c_\beta
        = \frac{4\ln(24)}{\beta}\cdot \frac{\beta}{96\ln(24)}
        = \frac{1}{24} \,.
    \]
    When this event does not occur, each $x_i$ satisfies $\diam(\Gamma_{\gamma(x_i)}) \le \Delta$,
    and since $(\Gamma, \rep) \in \cG$, the cell rejection algorithm accepts with probability
    $1-\rho$. Hence the algorithm accepts except with probability of failure at most $1/24 +
    \rho \cdot s \le 1/12$.

    On the other hand, suppose $\Exu{\bm{x} \sim \mu}{\diam(\Gamma_{\gamma(\bm{x})})} >
    \beta\epsilon$.
    We claim that
    $\Pru{\bm{x} \sim \mu}{\diam(\Gamma_{\gamma(\bm{x})}) > \beta\epsilon/2} > \beta\epsilon/2$.
    Suppose for a contradiction that this is not the case. Then
    \begin{align*}
        \Exu{\bm{x} \sim \mu}{\diam(\Gamma_{\gamma(\bm{x})})}
        &\le \Pru{\bm{x} \sim \mu}{\diam(\Gamma_{\gamma(\bm{x})}) > \beta\epsilon/2} \cdot \diam(\cX)
            + \Pru{\bm{x} \sim \mu}{\diam(\Gamma_{\gamma(\bm{x})}) \le \beta\epsilon/2}
                \cdot \beta\epsilon/2 \\
        &\le \frac{\beta\epsilon}{2} \cdot 1 + 1 \cdot \frac{\beta\epsilon}{2} = \beta\epsilon\,,
    \end{align*}
    a contradiction, so the claim holds. Then, the probability that $\diam(\Gamma_{\gamma(x_i)}) \le
    \beta\epsilon/2$ for every $x_i$ is at most
    \[
        \left( 1 - \frac{\beta\epsilon}{2} \right)^{\ceil{k/\epsilon}}
        \le e^{- \frac{\beta\epsilon}{2} \cdot \left\lceil \frac{k}{\epsilon} \right\rceil}
        \le e^{-\beta k/2}
        = \frac{1}{24} \,.
    \]
    When this event does not occur, we have $\diam(\Gamma_{\gamma(x_i)}) > \beta\epsilon/2$ for some
    $x_i$, and in this case the cell rejection algorithm rejects or outputs $\bot$ except with
    probability at most $\rho < 1/24$. Thus the algorithm rejects except with probability of
    failure at most $1/12$.
\end{proof}

\begin{remark}
\label{remark:equivalence-good}
When performing identity testing, with a universe $\cU$ of clusterings and a subclass $\cG \subseteq
\cU$, we defined in the introduction the set of ``good'' clusterings for input $\mu$:
\[
  \cG_{\Delta,\epsilon} \define \Big\{
    \left((\Gamma, \rep),\mu\right) \;|\; (\Gamma, \rep) \in \cG,\;
      \Pru{\bm{x} \sim \mu}{\diam(\Gamma_{\gamma(\bm x)}) > \Delta} \le c \cdot \epsilon
      \Big\} \,,
\]
where $c \define c_{1/4}$ in the definition of $c_\beta$ above.  For
equivalence testing, \cref{def:adversarial-testing} does not suffice to define the problem, since it
does not apply to properties of pairs of distributions, but it is clear how to adapt the definition.
Then we must also define the set of ``good'' clusterings for diameter-guarded equivalence testing,
for inputs $\mu$ and $\nu$, as follows:
\begin{align*}
  \cG_{\Delta,\epsilon} &\define \Big\{
    \left((\Gamma, \rep),\mu,\nu\right) \;|\; (\Gamma, \rep) \in \cG,\;\\
      &\qquad
      \Pru{\bm{x} \sim \mu}{\diam(\Gamma_{\gamma(\bm x)}) > \Delta} \le c \cdot \epsilon
      \wedge
      \Pru{\bm{x} \sim \nu}{\diam(\Gamma_{\gamma(\bm x)}) > \Delta} \le c \cdot \epsilon
      \Big\} \,,
\end{align*}
\end{remark}

We now prove our general result. Below, $\rho_d$ and $\rho_r$ are our requirements on the failure of
probability $\rho$ of cell rejection and discovery; in many cases these subroutines will be
deterministic and we may take $\rho=0$.
\begin{restatable}{lemma}{lemmaintrogeneral}
\label{lemma:metric-general}
Let $(\cX, \dist)$ be any finite unit-diameter metric space, let $\cU$ be a class of clusterings,
and let $\cG \subseteq \cU$ be a subclass.  Then for $\epsilon > 0$, $\alpha \in (0,1]$, and $\Delta
\le \epsilon/8$,
\begin{enumerate}

\item Let $m^\mathsf{id}_\EMD(\epsilon/2)$ be the sample complexity of testing identity
under $\EMD_\dist$ on domain $\cX$, distance parameter $\epsilon/2$ and failure probability
$1/12$. Then $(\cU,\cG,\Delta)$-diameter-guarded $\epsilon$-testing identity of distributions
on $\cX$ under $\EMD_\dist$, where $\nu$ is the target (known) distribution, requires at most
\[
m(\epsilon) = O\left(m^\mathsf{id}_\EMD(\epsilon/2) + \frac{1}{\epsilon} \right)
\]
samples and
\[
q(\epsilon) = O\left(m^\mathsf{id}_\EMD(\epsilon/2) \cdot \left(\frac{1}{\alpha}
  + \qcell(\cU,\cG,\alpha,\nu,\rho_d)\right)
  + \frac{1}{\epsilon}\cdot\qreject(\cU,\cG,\Delta,\epsilon/8,\rho_r)\right)
\]
queries, where $\rho_d = \frac{b}{m^\mathsf{id}_\EMD(\epsilon/2)}$ and $\rho_r = b \epsilon$ for a
sufficiently small universal constant $b > 0$.

\item Let $m^\mathsf{equiv}_\EMD(\epsilon/2)$ be the sample complexity of testing equivalence of
distributions under $\EMD_\dist$ on domain $\cX$, with distance parameter $\epsilon/2$ and failure
probability $1/12$. Then $(\cU,\cG,\Delta)$-diameter-guarded $\epsilon$-testing equivalence of
distributions on $\cX$
under $\EMD_\dist$ requires at most
\[
m(\epsilon) = O\left(m^\mathsf{equiv}_\EMD(\epsilon/2) + \frac{1}{\epsilon}\right)
\]
samples and
\[
  q(\epsilon) = O\left(\frac{1}{\epsilon} \cdot \qreject(\cU,\cG,\Delta,\epsilon/8,\rho_r) \right)
\]
queries, where $\rho_r = b \epsilon$ for a sufficiently small universal constant $b > 0$.
\end{enumerate}
\end{restatable}

\begin{proof}[Proof of \cref{lemma:metric-general}, Part 1]
Let $\nu$ be the known target distribution over $\cX$, so that the algorithm must distinguish
between $\mu = \nu$ and $\EMD_\dist(\mu,\nu) > \epsilon$.

Before defining the algorithm, we define an auxiliary distribution. We let
$\mu^\bullet$ be the distribution over points $\bm{x} \in \cX$ obtained as follows: first sample
$\bm{y} \sim \mu$ and then sample $\bm{x}$ from the conditional distribution $\nu |
\Gamma_{\gamma(\bm y)}$, \ie $\nu$ conditioned on the cell $\Gamma_{\gamma(\bm y)}$; if this
distribution does not exist for $\Gamma_{\gamma(\bm y)}$ (\ie $\nu(\Gamma_{\gamma(\bm y)}) = 0$)
then choose $\bm{x}$ uniformly randomly from $\Gamma_{\gamma(\bm y)}$.

The algorithm is now defined as follows:
\begin{enumerate}
    \item Run \textsc{TestClustering} with $\beta=1/4$, error parameter $\epsilon$, and error
      probability $\delta = 1/12$, and output $\ClusterReject$ if that algorithm rejects; otherwise,
      continue.
  \item Attempt to draw $m \define m^{\mathsf{id}}_{\EMD}(\epsilon/2)$ independent samples from
      $\mu^\bullet$, using the procedure defined in \cref{claim:container-sampling} below. If the
      procedure outputs $\reject$, output $\Reject$, and if it outputs $\clusterreject$ then
      output $\ClusterReject$. Otherwise, continue\footnote{This second step in the algorithm is the
reason that we do not define our distribution tester as a pair of ``cluster-tester'' and
``distribution-tester'' algorithms, which is how the tester-learners of \cite{RV23} are defined.}.
    \item Using the $m$ samples drawn in the previous step, simulate the identity tester under
        $\EMD_\dist$ on domain $\cX$, with target distribution $\nu$ and input distribution
        $\mu^\bullet$, with error parameter $\epsilon/2$ and error
        probability $1/12$.  Output $\Accept$ or $\Reject$ according to the output of the
        identity tester.
\end{enumerate}
To complete step 2 of the algorithm, we require the following sampling procedure:
\begin{claim}
\label{claim:container-sampling}
There is a randomized procedure which takes parameter $m$ and whose output is either $\reject$,
$\clusterreject$, or $m$ independent random samples from $\mu^\bullet$, satisfying the following
conditions with probability at least $11/12$ over the randomization of the procedure:
\begin{enumerate}
\item If the input $\left((\Gamma, \rep),\mu\right)$ satisfies $(\Gamma, \rep) \in \cG$, then the
    output is not $\clusterreject$; and
\item If $\mu = \nu$, then the output is either $\clusterreject$ or $m$ independent samples from
    $\mu^\bullet$.
\end{enumerate}
The procedure makes at most $m$ requests to $\SAMP(\Gamma,\rep,\mu)$ and at most
$m\cdot\left(\qcell(\cU, \cG,\alpha,\nu,\rho_d) + \frac{24}{\alpha}\right)$ requests to the
$\LABEL(\Gamma,\rep)$ oracle, where we require $\rho_d \le \frac{1}{24 m}$.
\end{claim}
\begin{proof}[Proof of claim]
Let $k = 24m/\alpha$. The algorithm initiates a counter $c = 0$.  For each $i \in [m]$,
perform the following. Sample $\bm{r_i}$ from $\SAMP(\Gamma,\rep,\mu)$ which has value $\bm{r_i} =
\rep(\Gamma_{\gamma(\bm y_i)})$ for $\bm y_i \sim \mu$. Then run an $(\alpha,\nu,\rho_d)$-cell
discovery algorithm for $(\cU, \cG)$ with input $\bm{r}_i$, which outputs either $\clusterreject$ or
a container $C
\subseteq \cX$. If it outputs
$\clusterreject$ then the procedure also outputs $\clusterreject$. If its output $C$ satisfies
$\nu(C) = 0$, the procedure outputs $\reject$. Otherwise, repeat the following:
\begin{enumerate}
\item Choose $\bm{x} \sim \nu|C$ (\ie sample $\bm x$ from $\nu$ conditioned on the
container $C$).
\item Using 1 label query, check if $\bm{x} \in \Gamma_{\gamma(\bm r_i)}$. If so, output $\bm{x_i} = \bm{x}$
and stop repeating.
\item Increment the counter $c$. If $c > k$ then output $\clusterreject$.
\end{enumerate}
Suppose that $(\Gamma, \rep) \in \cG$. We need to show that the output is not $\clusterreject$,
which happens either when the cell discovery algorithm outputs $\clusterreject$ or when the counter
$c$ grows too large. For each $i \in [m]$ the cell discovery algorithm will output a container
$C_i$ such that $\Gamma_{\gamma(\bm{r_i})} \subseteq C_i$, except with failure probability
$\rho_d$, so this holds for every $i$ except with probability at most $m \rho_d \le 1/24$ (for
sufficiently small choice of constant $b > 0$). By
the conditions on the cell discovery algorithm, it must be that $\nu(\Gamma_{\gamma(\bm{r_i}}))
\geq \alpha \cdot \nu(C_i)$ for each $i \in [m]$. For each $i \in [m]$, let $\bm{\ell}_i$ be the
number of times the inner loop is executed, or zero if $\nu(C_i) = 0$. Suppose $\nu(C_i) > 0$;
then each time $\bm{x} \in \Gamma_{\gamma(\bm r_i)}$ is chosen inside the loop, it is contained
in $\Gamma_{\bm r_i}$ with probability $\frac{\nu(\Gamma_{\gamma(\bm r_i}))}{\nu(C_i)} \geq
\alpha$. Therefore, for each $i \in [m]$, $\Ex{\bm{\ell}_i} \leq \frac{1}{\alpha}$. So the total
number of times the counter is incremented is at most $\Ex{\sum_{i=1}^m \bm{\ell}_i} \le
m/\alpha$. Therefore, by Markov's inequality, the probability that the algorithm outputs
$\clusterreject$ is at most $\frac{m/\alpha}{k} = 1/24$. Hence the algorithm only outputs
$\clusterreject$ with probability at most $2 \cdot 1/24 = 1/12$.

Now suppose that $\mu = \nu$ and suppose that the algorithm does not output $\clusterreject$, so
that for each $i$ the cell discovery algorithm outputs a container $C_i$. Since $\mu=\nu$, it
must always be the case that $\nu(\Gamma_{\gamma(\bm{r_i})}) > 0$, because
$\Gamma_{\gamma(\bm{r_i})}$ contained a sample drawn from $\mu = \nu$. Thus as long as each
$C_i$ satisfies $\Gamma_{\gamma(\bm{r}_i)} \subseteq C_i$, we will have $\nu(C_i) > 0$. Thus the
algorithm outputs $\reject$ only with probability at most $m \rho_d \le 1/24$ for sufficiently small
choice of constant $b > 0$.
\end{proof}

The sample and query complexity bounds now follow from \cref{lemma:clustering-tester} and 
\cref{claim:container-sampling}. It remains to prove correctness of the algorithm.

The first requirement of \cref{def:adversarial-testing} is satisfied because, when
$((\Gamma,\rep),\mu) \in \cG_{\Delta,\epsilon}$, it holds that $\Pru{\bm{x} \sim
\mu}{\diam_\dist(\Gamma_{\gamma(\bm{x})}) > \Delta} \le c \cdot \epsilon$, so the first step of
the algorithm outputs $\ClusterReject$ only with probability at most $1/12$, by
\cref{lemma:clustering-tester}, and the second step of the algorithm only outputs $\ClusterReject$
with probability at most $1/12$ due to \cref{claim:container-sampling}.

Now, suppose $\Exu{\bm{x} \sim \mu}{\diam(\Gamma_{\gamma(\bm{x})})} > \epsilon/2 >
\epsilon/4$. Then the first step of the algorithm outputs $\ClusterReject$ with probability at least
$11/12$ by \cref{lemma:clustering-tester}, in which case the second and third requirements of
\cref{def:adversarial-testing} are also satisfied.

    On the other hand, suppose $\Exu{\bm{x} \sim \mu}{\diam(\Gamma_{\gamma(\bm{x})})} \le
    \epsilon/2$. We consider two cases. First, suppose $\mu = \nu$. Except with failure probability
    $1/12$, the sampling procedure from \cref{claim:container-sampling} outputs $\clusterreject$ or
    produces $m$ independent random samples from $\mu^\bullet$. In the first case, the algorithm
    outputs $\ClusterReject$. In the second case, we may simulate the identity tester in step 3,
    which will accept with probability at least $11/12$, since $\mu^\bullet = \mu = \nu$ in this
    case (which holds by definition). Therefore, the algorithm outputs $\ClusterReject$ or $\Accept$
    except with failure probability at most $1/6$, as desired.

    In the second case, suppose $\EMD_\dist(\mu,\nu) > \epsilon$. Recalling that $\diam(\cX) = 1$,
    we use the triangle inequality and \cref{lemma:emd-tv-diameter}, as follows.  Recall that
    $\induced{\mu}{\Gamma} = \induced{\mu^\bullet}{\Gamma}$ by definition. Then:
    \begin{align*}
        \EMD_\dist(\mu, \nu)
        &\le \EMD_\dist(\mu, \mu^\bullet) + \EMD_\dist(\mu^\bullet, \nu) \\
        &\le \dist_\TV(\induced{\mu}{\Gamma}, \induced{\mu^\bullet}{\Gamma})
           + \Exu{\bm{x} \sim \mu}{\diam(\Gamma_{\gamma(\bm x)})} + \EMD_\dist(\mu^\bullet, \nu) \\
        &\le \EMD_\dist(\mu^\bullet,\nu) + \frac{\epsilon}{2} \,.
    \end{align*}
    Therefore $\EMD_\dist(\mu^\bullet,\nu) > \epsilon/2$. If the algorithm outputs $\ClusterReject$
    in either the first or second step, then we are done. Otherwise, the second step of the
    algorithm either outputs $\reject$ or produces $m$ independent random samples from
    $\mu^\bullet$.  If it outputs $\reject$, then whole algorithm outputs $\Reject$, and we are
    done. Otherwise, we simulate the identity tester under $\EMD_\dist$, which rejects with
    probability at least $11/12$. Therefore the probability to output $\Accept$ is at most $1/12$ as
    desired.
\end{proof}

The proof of the equivalence testing result follows the same outline, which we sketch below.

\begin{proof}[Proof sketch of \cref{lemma:metric-general}, Part 2]
    Let $\mu$ and $\nu$ be the two input distributions to the equivalence tester and let
    $(\Gamma,\rep)$ be the clustering. Recall that $\induced{\mu}{\Gamma}$ and
    $\induced{\nu}{\Gamma}$ are the distributions over the representative points induced by $\mu$
    and $\nu$, respectively, which are also the distributions over the points received from
    $\SAMP(\Gamma,\rep,\mu)$ and $\SAMP(\Gamma,\rep,\nu)$.  The algorithm proceeds as follows:
    \begin{enumerate}
        \item Run \textsc{TestClustering} twice, once with input distribution $\mu$ and once with
            input distribution $\nu$, both with $\beta=1/4$, error parameter $\epsilon$, and $\delta
            = 1/12$. If either run rejects, output $\ClusterReject$.
        \item Using $m^\mathsf{equiv}_\EMD(\epsilon/2)$ samples, run the equivalence tester under
            $\EMD_\dist$ on inputs $\induced{\mu}{\Gamma}$ and $\induced{\nu}{\Gamma}$.
    \end{enumerate}
    The sample and query complexity are immediate. The correctness of the algorithm follows by
    nearly an identical argument as in the case for identity testing. We require only the following
    inequality when $\EMD_\dist(\mu,\nu) > \epsilon$, and $\Exu{\bm{x} \sim
    \mu}{\diam(\Gamma_{\gamma(\bm x)})}, \Exu{\bm{x} \sim \nu}{\diam(\Gamma_{\gamma(\bm x)})} \le
    \epsilon/4$ (where the latter condition is guaranteed by step 1),  using
    \cref{lemma:emd-tv-diameter}:
    \begin{align*}
      \epsilon < \EMD_\dist(\mu,\nu)
        &\leq \EMD_\dist(\mu,\induced{\mu}{\Gamma})
            + \EMD_\dist(\induced{\mu}{\Gamma},\induced{\nu}{\Gamma})
            + \EMD_\dist(\induced{\nu}{\Gamma},\nu) \\
        &\leq \dist_\TV(\induced{\mu}{\Gamma},\induced{\mu}{\Gamma})
            + \Exu{\bm{x} \sim \mu}{\diam(\Gamma_{\gamma(\bm x)})} \\
          &\qquad + \dist_\TV(\induced{\nu}{\Gamma},\induced{\nu}{\Gamma})
            + \Exu{\bm{x} \sim \nu}{\diam(\Gamma_{\gamma(\bm x)})} \\
          &\qquad+ \EMD_\dist(\induced{\mu}{\Gamma},\induced{\nu}{\Gamma}) \\
        &\leq \EMD_\dist(\induced{\mu}{\Gamma},\induced{\nu}{\Gamma}) + \epsilon/2 \,.
    \end{align*}
  This guarantees the second step will reject appropriately.
\end{proof}

To apply the general \cref{lemma:metric-general}, it is necessary to plug in values for for the
optimal identity or equivalences testers under $\EMD$, as well as query algorithms for cell
discovery and rejection. A general bound for testing identity and equivalence, in terms of the size
of the optimal small-diameter covering of the domain, is as follows, and we will give an improved
bound for domains $[n]^d$ and $[0,1]^d$. In \cref{section:applications} we give bounds on cell
discovery and rejection in specific cases.

\begin{lemma}
\label{lemma:general-emd-testers}
Let $(\cX, \dist)$ be any unit-diameter metric space. For any $\epsilon > 0$, let $\GammaStar$ be
any clustering of $\cX$ into cells of diameter at most $\epsilon/4$ and write $\nstar =
|\GammaStar|$. Then there is an $\epsilon$-tester for identity of distributions on domain
$\cX$, under $\EMD_\dist$, with error probability $1/12$ and sample complexity
\[
  m^\mathsf{id}_\EMD(\epsilon) = O\left(\frac{\sqrt{\nstar}}{\epsilon^2}\right) \,,
\]
and there is an $\epsilon$-tester for equivalence of distributions on domain $\cX$, under
$\EMD_\dist$, with error probability $1/12$ and sample complexity
\[
  m^\mathsf{equiv}_\EMD(\epsilon) = O\left(\frac{\sqrt{\nstar}}{\epsilon^2} +
\frac{(\nstar)^{2/3}}{\epsilon^{4/3}}\right) \,.
\]
\end{lemma}
\begin{proof}
Write $\Gamma^* \define \GammaStar$, let $\nu^* \define \induced{\nu}{\Gamma^*}$ and $\mu^* \define
\induced{\mu}{\Gamma^*}$.  We must first prove that $\dist_\TV(\mu^*,\nu^*) > \epsilon/4$ whenever
$\EMD_\dist(\mu,\nu) > \epsilon$. This is proved using the triangle inequality and
\cref{lemma:emd-tv-diameter}:
\begin{align*}
\epsilon < \EMD_\dist(\mu,\nu)
  &\leq \EMD_\dist(\nu,\nu^*) + \EMD_\dist(\nu^*,\mu^*) + \EMD_\dist(\mu^*,\mu) \\
  &\leq \Exu{\bm{x} \sim \nu}{\diam(\Gamma^*_{\gamma(\bm x)})} 
      + \Exu{\bm{x} \sim \nu}{\diam(\Gamma^*_{\gamma(\bm x)})} 
      + \Exu{\bm{x} \sim \mu}{\diam(\Gamma^*_{\gamma(\bm x)})} + \dist_\TV(\nu^*,\mu^*) \\
  &\leq \dist_\TV(\nu^*,\mu^*) + 3\epsilon/4 \,.
\end{align*}
We may now conclude the proof by applying the known upper bounds on identity and equivalence
testing, in \cite{VV17} and \cite{CDVV14} respectively, on domain size $\nstar$, which give the
expressions in the lemma statement.
\end{proof}

\subsection{Applications}
\label{section:applications}

We now apply our general \cref{lemma:metric-general} to some specific cases of interest. We start
with a simple example of testing on the line $[n]$ under a ``threshold metric''\!, which shows how
to use thresholds to interpolate between TV distance and EMD, and then move on to a series of
applications to the hypergrid $[n]^d$ under $\ell_p$ metrics.

\subsubsection{Testing under a Threshold Metric}

Fix domain $[n]$ and let $0 < R \le n-1$. Let $\dist(\cdot, \cdot)$ denote the \emph{normalized
$R$-threshold} metric, given by
\[
    \forall x,y \in [n] :\qquad
    \dist(x, y) \define \frac{1}{R} \cdot \min\left\{ |x-y|, R \right\} \,.
\]
Note that TV distance is obtained by setting $R = 1$.  Recall that $\BOX$ is the class of
axis-aligned box clusterings, which on $[n]$ is simply the class of clusterings whose cells are
intervals. We first observe that cell discovery for $\BOX$ is easy, because one may perform binary
search to find the endpoints of the interval (cell) containing the representative $h$.

\begin{fact}
    \label{lemma:1d-cell-discovery}
    Fix domain $[n]$, let $0 < R \le n-1$, and let $\dist(\cdot, \cdot)$ be the normalized
    $R$-threshold metric. Then for any $0 < t_1 \le t_2$, $\qreject(\BOX, \BOX, t_1, t_2) \le
    \qcell(\BOX, \BOX) = O(\log n)$.
\end{fact}

\noindent
Combining this lemma with the general upper bound yields identity and equivalence testers under
$\EMD_\dist$; we let $\GammaStar$ be a clustering of $[n]$ of size $\nstar = O\left( \min\{n,
\frac{n}{\epsilon R}\} \right)$ whose cells are consecutive intervals of size $\max\{1,\epsilon R /
4\}$ and hence diameter at most $\epsilon/4$. The next statement then follows from
\cref{lemma:1d-cell-discovery,lemma:general-emd-testers,lemma:metric-general}.

\begin{example}
    \label{thm:capped-bb}
    Fix domain $[n]$, let $\epsilon \in (0,1)$ and $0 < R \le n-1$, and let $\dist(\cdot, \cdot)$ be
    the normalized $R$-threshold metric. Then $(\BOX, \BOX, \epsilon/8)$-diameter-guarded
    $\epsilon$-testing identity of distributions under $\EMD_\dist$ requires at most $m(\epsilon) =
    O\left( \frac{\sqrt{\nstar}}{\epsilon^2} \right)$ samples and $q(\epsilon) = O\left( m(\epsilon) \cdot \log n \right)$ queries, and
$(\BOX, \BOX, \epsilon/8)$-diameter-guarded $\epsilon$-testing equivalence of distributions under
$\EMD_\dist$ requires at most $m(\epsilon) = O\left( \frac{\sqrt{\nstar}}{\epsilon^2} +
\frac{(\nstar)^{2/3}}{\epsilon^{4/3}}\right)$ samples and $q(\epsilon) = O\left( \frac{\log
n}{\epsilon} \right)$ queries, where $\nstar = O\left( \min\{n, \frac{n}{\epsilon R} \}\right)$.
\end{example}

\subsubsection{Lemma for EMD Testing under $\ell_p$ on Hypergrids}

Fix the hypergrid domain $[n]^d$. Recall that a set $S \subseteq [n]^d$ is called \emph{connected}
if it is connected in the standard hypergrid graph (where there is an edge between $x,y \in [n]^d$
if and only if $\|x-y\|_1 = 1$). On this domain, for any $p \geq 1$, the normalized $\ell_p$ metric
is
\[
  \forall x,y \in [n]^d :\qquad \dist(x,y) \define \frac{\|x-y\|_p}{\|n \cdot \vec 1\|_p} 
  = \frac{\|x-y\|_p}{d^{1/p} \cdot n} \,.
\]

We start by giving concrete bounds on the sample complexity of testing identity and equivalence
under EMD on the hypergrid $[n]^d$, which is a subtask for our remaining applications. Following the
argument of \cite{DNNR11}, we employ a hierarchical clustering of the hypergrid $[n]^d$ at doubling
granularities to obtain the following bounds. Compared to \cite{DNNR11}, the only differences in our
presentation are that 1) we are also interested in identity, not only equivalence testing; 2) we
incorporate the more recent optimal bounds on equivalence testing from \cite{CDVV14}; and 3) we
state our result in terms of general normalized $\ell_p$ metrics. For the sake of completeness, we
include a proof in \cref{appendix:emd-testing-hypergrid}.

\begin{restatable}{lemma}{lemmaemdtestinghypergrid}
    \label{lemma:emd-testing-hypergrid}
    Let $[n]^d$ be any hypergrid equipped with any normalized $\ell_p$ metric $\dist(\cdot,\cdot)$
    for $p \ge 1$, and let $\epsilon > 0$. There exists an $\EMD_\dist$ identity tester with sample
    complexity ${ 2^{O(d)} \cdot \widetilde O\left( \epsilon^{-\max\left\{ 2, \frac{d}{2} \right\}}
    \right) }$ and an $\EMD_\dist$ equivalence tester with sample complexity ${ 2^{O(d)} \cdot
    \widetilde O\left( \epsilon^{-\max\left\{ 2, \frac{2d}{3} \right\}} \right) }$, where the
    $\widetilde O(\cdot)$ notation hides polylogarithmic factors in $1/\epsilon$.
\end{restatable}

\subsubsection{Connected Cells}

Our first result is a very general result which shows that sublinear sample and query complexity
suffice when the only promise on the clusterings is that they are connected, and the algorithm is
required to succeed on any connected clustering with high probability of low diameter.

\begin{restatable}{theorem}{thmintrocc}
\label{thm:intro-c-c}
Fix domain $[n]^d$, let $\epsilon > 0$, and let $\dist(\cdot,\cdot)$ be the normalized $\ell_p$
metric on $[n]^d$ for any $p \geq 1$. Let $\Delta \le \frac{1}{8d^{1/p}} \epsilon$. Then
$(\CONN,\CONN,\Delta)$-diameter-guarded $\epsilon$-testing equivalence under $\EMD_\dist$ requires
at most $m(\epsilon) = 2^{O(d)} \cdot \widetilde O\left( \epsilon^{-\max\{2, \frac{2d}{3}\}}
\right)$ samples and $q(\epsilon) = O\left(d \epsilon^{d-2} n^{d-1} \right)$ queries, and
$\epsilon$-testing identity requires at most $O(m(\epsilon))$ samples and at most $O(m(\epsilon) +
q(\epsilon))$ queries.
\end{restatable}

Note that the domain size is $N = n^d$, so $n^{d-1} = N^{1-1/d}$,
which is sublinear in the domain size.
The result on identity testing follows from reducing identity testing to equivalence testing, using
label queries to simulate samples from $\SAMP(\Gamma,\rep,\nu)$ where $\nu$ is the known
distribution. We use a reduction to equivalence instead of a direct result for identity testing,
because it may require as many as $O(N)$ queries to perform cell discovery when the only promise on
the cells is that they are connected.   To prove the theorem it is sufficient to prove bounds on the
cell rejection task; then the result follows from \cref{lemma:metric-general} and
\cref{lemma:emd-testing-hypergrid}.

\begin{lemma}
\label{lemma:cell-reject-connected}
Fix domain $[n]^d$ and let $\dist(\cdot,\cdot)$ be the normalized $\ell_p$ metric for any $p \geq
1$. Then for any $0 < \epsilon_1 < \epsilon_2 < 1$ that satisfy $\epsilon_2 > 2 d^{1/p} \cdot
\epsilon_1$,
\[
  \qreject(\CONN,\CONN, \epsilon_1, \epsilon_2) = O( d (\epsilon_2 \cdot n)^{d-1} )
\,.
\]
\end{lemma}
\begin{proof}
Let $H \subseteq [n]^d$ be any connected set, and let $h \in H$. Let $k = \lceil \epsilon_2
\cdot n / 2 \rceil$ and let $B \subseteq [n]^d$ be the set of points $B = \{ x \in [n]^d \;|\;
\|x-h\|_\infty = k \}$. The cell rejection algorithm will simply query every element of $B$,
using $|B| \leq d (1 + 2k)^{d-1} = O(d (\epsilon_2 \cdot n)^{d-1})$ queries, and accept if and only 
if $B \cap H = \emptyset$.

Suppose that $H$ has diameter at most $\epsilon_1$, so that $\|x - y\|_p \leq \epsilon_1 \cdot
d^{1/p} \cdot n < \epsilon_2 n / 2$ for all $x,y \in H$. Specifically, for all $x \in H$,
\[
\|x-h\|_\infty \leq \|x - h\|_p < k \,,
\]
so $H \cap B = \emptyset$, and the algorithm will correctly accept.

Now suppose that $H$ has diameter at least $\epsilon_2$, so there exist two points $x,y \in H$
with $\dist(x,y) \geq \epsilon_2$. By the triangle inequality, one of $\dist(x,h) \geq \epsilon_2/2$
or $\dist(y,h) \geq \epsilon_2/2$ must hold; assume $\dist(x,h) \geq \epsilon_2/2$. Then
\[
  \frac{\epsilon_2}{2} \leq \dist(x,h) = \frac{\|x-h\|_p}{d^{1/p} \cdot n} \leq \frac{d^{1/p} \cdot
\|x-h\|_\infty}{d^{1/p} \cdot n} \,,
\]
so $\|x-h\|_\infty \geq k$. Therefore, since $H$ is connected, there exists $x'$ on any path from
$x$ to $h$ such that $x' \in H \cap B$, so the algorithm will correctly reject.
\end{proof}

\subsubsection{Convex Cells on the Grid}

We now strengthen the promise on the clusterings to ensure that they are connected and convex. By
giving bounds on the query complexity of cell discovery in this case, we obtain a result for
identity testing. Recall that $\CC$ is the class of connected convex clusterings.

\begin{restatable}{theorem}{thmintrocccc}
\label{thm:intro-cc-cc}
Fix domain $[n]^d$, let $\epsilon \geq 8\Delta > 0$, and let $\dist(\cdot,\cdot)$ be the normalized
$\ell_p$ metric on $[n]^d$ for any $p \ge 1$. Then $(\CC,\CC,\Delta)$-diameter-guarded
$\epsilon$-testing of identity under $\EMD_\dist$ requires at most $m(\epsilon) = 2^{O(d)} \cdot
\widetilde O\left(\epsilon^{-\max\{2,d/2\}}\right)$ samples and $q(\epsilon) = O(m(\epsilon))
\cdot \left( 2^{O(d)} O(n^{d-1}) + O(\log n) \right)$ queries.
\end{restatable}

We note that, for equivalence testing, we fall back to the result from \cref{thm:intro-c-c} for
connected clusterings, which is better than what the cell discovery procedure given below would
imply.

To prove the theorem, fix domain $[n]^d$ and for each $v \in [n-1]^d$ define the \emph{subcube}
$S_v\subseteq \bR^d$ as 
\[
S_v \define \{ v + x \;|\; x \in \bR^d, \forall i \in [d], v_i \leq x_i \leq v_i+1 \} \,.
\]
We require the following lemma, which can be found in \cite{HY22}.
\begin{lemma}
Fix domain $[n]^d$ and let $A \subseteq \bR^d$ be any convex set. Then the number of subcubes $S_v$
where $A \cap S_v \notin \{ \emptyset, S_v \}$ is at most $O(d n^{d-1})$.
\end{lemma}
For any set $A \subseteq [n]^d$, we define the \emph{boundary} as $\mathsf{bd}(A)$ as the set of
points $x \in A$ with a neighbor in the complement of $A$.
We may now prove the required bounds on cell discovery with $\alpha=1$:
\begin{lemma}
Fix domain $[n]^d$. Then, for any $t_1 \leq t_2$, it holds that $\qreject(\CC,\CC,t_1,t_2)
\leq \qcell(\CC,\CC) = O(d 4^d n^{d-1} + \log n)$.
\end{lemma}
\begin{proof}[Proof sketch.]
It suffices to consider cell discovery with $\alpha=1$, due to \cref{fact:rejection-to-discovery}.
Given membership access to a connected convex cell $H \subseteq [n]^d$ and a representative $h \in
H$, the algorithm performs the following. Let $H'$ be the convex hull of $H$. Observe that it
suffices to discover the boundary $\mathsf{bd}(H)$, since $H'$ is the convex hull of
$\mathsf{bd}(H)$.  The algorithm uses $O(\log n)$ membership queries to do binary search on any
axis-aligned line through $h$ to find a point $x \in \mathsf{bd}(H)$. If $d = 1$ then it does binary
search again to find the only other boundary point and then terminates. Otherwise, for $d \geq 2$,
consider the set $\cS$ of subcubes $S_v$ such that $H' \cap S_v \notin \{ \emptyset, S_v \}$ and
$S_v \cap \mathsf{bd}(H) \neq \emptyset$.  We know $|\cS| = O(d n^{d-1})$. Consider a graph $G$ on
vertex set $\cS$ and put an edge between two subcubes $S_v, S_u \in \cS$ if they share a boundary
point in $\mathsf{bd}(H)$. Since $H$ is connected, $G$ is also connected. Therefore the algorithm
may perform breadth-first search on $G$ to discover the set $\mathsf{bd}(H)$. For each vertex $p$ of
$G$, at most $4^d$ membership queries to $H$ are required to discover the set of neighbors of
$p$, since there are at most $3^d-1$ possible neighbors containing at most $4^d$ points of $[n]^d$.
\end{proof}

\subsubsection{Decision Trees and Axis-Aligned Box Cells on the Grid}

Recall that clusterings computed by decision trees are a subclass of the axis-aligned box
clusterings, denoted $\BOX$. So the next theorem applies to decision trees as a special case.

\begin{restatable}{theorem}{thmintrobb}
\label{thm:intro-bb}
Fix domain $[n]^d$, let $\dist(\cdot,\cdot)$ be the normalized $\ell_p$ metric on $[n]^d$ for any $p
\ge 1$, and let $0 < 8\Delta \le \epsilon$. Then $(\BOX,\BOX,\Delta)$-diameter-guarded
$\epsilon$-testing of identity, under $\EMD_\dist$, requires at most $m(\epsilon) = 2^{O(d)}
\cdot \widetilde O( \epsilon^{-\max\{2,\frac{d}{2}\}} )$ samples and $q(\epsilon) =
O(m(\epsilon) \cdot d \log n)$ queries; and testing equivalence requires at most $m(\epsilon) =
2^{O(d)} \cdot \widetilde O( \epsilon^{-\max\{2,\frac{2d}{3}\}} )$ samples and $q(\epsilon) =
O(\frac{1}{\epsilon} \cdot d \log n)$ queries.
\end{restatable}

This follows easily from the general bound, by using basic binary search to exactly learn each cell.

\begin{proposition}
Let $\dist(\cdot,\cdot)$ be any metric on $[n]^d$ and let $0 < t_1 \le t_2$. Then
$\qreject(\BOX,\BOX, t_1, t_2) \leq \qcell(\BOX,\BOX) = O(d \log n)$.
\end{proposition}
\begin{proof}
On inputs $(H,h)$, the algorithm simply performs binary search in each dimension (along the
axis-aligned lines going through $h$) to find the boundary of the box $H$; since $H$ is defined by
its boundary, the algorithm may then output $H$. This uses at most $O(d \log n)$ queries.
\end{proof}

Our next theorem shows that in $d=2$ (and for $\ell_p$ metrics), the above result can be improved by
weakening the promise on the clusterings so that they are promised only to be connected and convex,
rather than axis-aligned boxes, without increasing the complexity of the tester.

\begin{restatable}{theorem}{thmintroccb}
\label{thm:intro-cc-b}
Fix domain $[n]^2$, let $\dist(\cdot,\cdot)$ be the normalized $\ell_p$ metric on $[n]^2$ with $p
\geq 1$, and let $0 < \Delta < \epsilon - \frac{8}{n}$ satisfy $\Delta \le \epsilon/8$.
Then $(\CC,\BOX,\Delta)$-diameter-guarded
$\epsilon$-testing of equivalence, under $\EMD_\dist$, requires at most $m(\epsilon) = \widetilde
O(\epsilon^{-2})$ samples and $q(\epsilon) = O(\frac{1}{\epsilon}\log n)$ queries.
\end{restatable}

\begin{proposition}
Fix domain $[n]^2$ and let $\dist(\cdot,\cdot)$ be the normalized $\ell_p$ metric on $[n]^2$ with $p
\geq 1$. Let $0 < t_1 < t_2 - \frac{8}{n}$. Then $\qreject(\CC,\BOX,t_1,t_2) = O(\log n)$.
\end{proposition}
\begin{proof}
Given $(H,h=\rep(H))$ belonging to a clustering $(\Gamma, \rep) \in \CC$, the algorithm performs the
following.
\begin{enumerate}
\item Use $O(\log n)$ membership queries to $H$ to perform binary
search on the horizontal and vertical lines through $h$ and find the largest $b_1 \geq h_1$, $b_2
\geq h_2$ and smallest $a_1 \leq h_1$, $a_2 \leq h_2$ such that the segments $\{ (x, h_2) : x \in
[a_1, b_1] \}$ and $\{ (h_1, y) : y \in [a_2, b_2] \}$ are contained in $H$ (these values must exist
since $H$ is convex).
\item Let $R \define [a_1,b_1] \times [a_2,b_2]$ with corners $(a_1,a_2)$, $(a_1, b_2)$,
$(b_1,a_2)$, $(b_1,b_2)$. Query the corner points and output $\bot$ if any of them is not in $H$.
Also query the points $(a_1-1,a_2), (a_1,a_2-1), (a_1-1,b_2), (a_1,b_2+1), (b_1+1,a_2),
(b_1,a_2-1), (b_1+1,b_2), (b_1,b_2+1)$ and output $\bot$ if any of them is in $H$.
\item Query the neighbors outside $R$ of the 1 or 2 midpoints of each of the 4 bounding lines of
$R$, and output $\bot$ if any of these points are in $H$. These are the points $\left( a_1 - 1,
\frac{a_2+b_2}{2} \right)$, $\left( a_2 - 1, \frac{a_1+b_1}{2} \right)$, $\left( b_1 + 1,
\frac{a_2+b_2}{2} \right)$, $\left( b_2 + 1, \frac{a_1+b_1}{2} \right)$ if the side-lengths
$(b_i-a_i)$ of $R$ are even, otherwise we take both the ceiling and floor of $(a_i + b_i)/2$.
\end{enumerate}
If the algorithm does not output $\bot$ above, then it accepts if and only if the rectangle $R$ has
diameter at most $t_1$, and otherwise it rejects.

If $H$ is an axis-aligned box, then the algorithm is correct, since we must have $H = R$. Now
suppose $H$ is connected and convex. If $\diam_\dist(H) \leq t_1$ then $\diam_\dist(R) \leq t_1$,
because the convexity of $H$ guarantees $R \subseteq H$, so the algorithm outputs either $\bot$ or
\accept, as desired.

If $\diam_\dist(H) > t_2$ then we observe that $H \subseteq R'$ where $R' \define [a_1-2, b_1+2]
\times [a_2-2, b_2+2]$ as follows. Suppose for the sake of contradiction that $x \in H \setminus
R'$, and assume without loss of generality that $x_1 \leq a_1$; similar arguments will hold when
$x_1 \geq b_1$, $x_2 \leq a_2$, and $x_2 \geq b_2$.

Consider two cases: either $x_2 \notin [a_2,b_2]$ or $x_2 \in [a_2,b_2]$.  If $x_2 \notin
[a_2,b_2]$, assume $x_2 < a_2$ without loss of generality,
then by connectivity of $H$ it must be the case that there exists either $y = (a_1, a_2
- z)$ for $z \geq 1$ such that $y \in H$, or $y = (a_1 - z, a_2)$ for $z \geq 1$ such that $y \in
H$. In each case, $y$ violates convexity, since $(a_1, a_2 - 1)$ and $(a_1-1, a_2)$ do not belong to
$H$ (due to step 2 in the algorithm), and one of them is the line between the corner $(a_1,a_2) \in
H$ and $y$. Therefore we cannot have $x_2 \notin [a_2,b_2]$.

If $x_2 \in [a_2, b_2]$ (in which case $x_1 < a_1-2$),
then we consider the line $L$ through $(x_1,x_2)$ and the queried point
$\left(a_1-1, \frac{a_2+b_2}{2}\right) \notin H$. Assume without loss of generality that $x_2 \leq
\frac{a_2+b_2}{2}$. Since $x_1 < a_1 - 2$, it is easy to verify that $L$ must intersect the convex
hull of $R$. This is a contradiction because $x \in H, R \subseteq H$ and $\left(a_1-1,
\frac{a_2-b_2}{2}\right) \notin H$, so convexity is violated.

We therefore conclude that $H \subseteq R'$, so $t_2 < \diam(H) \leq \diam(R')$. By the triangle
inequality, $\diam(R') \leq \diam(R) + \frac{8}{\diam_{\ell_p}([n]^2)} \leq \diam(R) + \frac{8}{n}$.
Therefore $\diam(R) \geq \diam(R') - \frac{8}{n} > t_2 - \frac{8}{n} \geq t_1$, so the algorithm
will $\reject$.
\end{proof}

\subsubsection{Axis-Aligned Boxes in $[0,1]^d$, with a Convexity Promise}

\newcommand{\BPrimeBoundary}{\partial B'}
\newcommand{\Ball}{\bB}
\newcommand{\side}{\mathsf{side}}

As stated in the introduction, we would like to test uniformity and equivalence of distributions on
domain $[0,1]^d$ given that 1) the clustering is promised to be convex and 2) we are required to
accept any clustering whose cells are axis-aligned boxes. Especially for uniformity testing, the
ability to reject non-box clusterings will be crucial. On the other hand, requiring that the
algorithm reject \emph{every} non-box clustering would not only be unnecessary, but in fact
introduce a more challenging algorithmic task. Thus this application precisely exploits the
definition of the model.

Let us define the required concepts. For $p \ge 1$, the normalized $\ell_p$ distance $\dist(\cdot,
\cdot)$ on $[0,1]^d$ is
\[
    \forall x,y \in [n]^d :\qquad
    \dist(x, y) \define \frac{\|x-y\|_p}{\|\vec 1\|_p} = \frac{\|x-y\|_p}{d^{1/p}} \,.
\]

For point $x \in \bR^n$ and $\delta > 0$, let $\Ball(x, \delta)$ denote the $\ell_2$-ball of radius
$\delta$ centered at $x$. For set $S \subset \bR^n$, let $\Ball(S, \delta) \define \bigcup_{x \in S}
\Ball(x, \delta)$. We write $\interior S$ and $\closure S$ for the interior and the closure of $S$,
respectively. Then, we define $\CONV_\delta$ as the set of clusterings $(\Gamma, \rep)$ of $[0,1]^d$
whose every cell $\Gamma_i \in \Gamma$ with representative $r_i \in \Gamma_i$ satisfies the
following:\footnote{The reason for the slightly subtle second requirement, rather than simply
requiring that $\Gamma_i$ be convex, is that we wish $\Gamma$ to be a partition of $[0,1]^d$, so
some flexibility is required at the boundary points. This technical point will be irrelevant for
absolutely continuous distributions with respect to the Lebesgue measure.} 1) $\Gamma_i$ is Lebesgue
measurable; 2) there is an open convex set $S$ such that $S \subseteq \Gamma_i \subseteq \closure
S$; and 3) $\Ball(r_i, \delta) \subseteq \Gamma_i$.

A set $B \subset \bR^d$ is a \emph{closed axis-aligned box} if it can be written as $B = \{ x \in
\bR^d : \forall i \in [d] ,\, a_i \le x_i \le b_i \}$ for some setting of $a_i$'s and $b_i$'s, and
we write $\side_i(B) \define b_i - a_i$ for the side length of $B$ along the $i$-th coordinate. We
say Lebesgue measurable set $B$ is an \emph{axis-aligned box} if there exists a closed axis-aligned
box $B'$ such that $\interior B' \subseteq B \subseteq B'$, and we write $\side_i(B) \define
\side_i(B')$. We then define $\BOX$ as the class of clusterings of $[0,1]^d$ whose cells are
axis-aligned boxes, and $\BOX_\delta \define \BOX \cap \CONV_\delta$.

For any bounded convex set $K \subset \bR^d$, let $B^*(K)$ denote the its minimum closed
axis-aligned bounding box, namely
\[
    B^*(K) \define \left\{ x \in \bR^d : \forall i \in [d] ,\,
        \inf_{y \in K} y_i \le x_i \le \sup_{y \in K} y_i \right\} \,.
\]

For this application, we will assume for simplicity that the input distributions $\mu, \nu$ over
$[0,1]^d$ are absolutely continuous with respect to the Lebesgue measure. Then the Radon-Nikodym
theorem implies that any Lebesgue measurable set in $[0,1]^d$ is $\mu$- and $\nu$-measurable; in
particular, this applies to the cells in clusterings from $\CONV_\delta$.

The main ingredient behind cell discovery and rejection in this setting is a ``bounding box''
procedure, which relies on the following result of \cite{LSV20} on convex optimization using
membership and evaluation oracles.

\begin{theorem}[\cite{LSV20}]
    \label{thm:convex-opt}
    Let $x_0 \in \bR^d$, $0 < r < R$, and $\epsilon > 0$. There exists a randomized algorithm
    which, given membership oracle access to a convex set $K$ satisfying $\Ball(x_0, r) \subseteq K
    \subseteq \Ball(x_0, R)$ and evaluation oracle access to a convex function $f$, uses $O\left( d^2
    \log^2\left( \frac{dR}{\epsilon r} \right) \right)$ queries and outputs, with probability at
    least $2/3$, a point $z \in \Ball(K, \epsilon)$ satisfying
    \[
        f(z) \le \min_{x \in K} f(x) + \epsilon\left( \max_{x \in K} f(x) - \min_{x \in K} f(x)
        \right) \,.
    \]
\end{theorem}
We remark that, since we are only concerned with approximations, $\min$ and $\max$ are exchangeable
with $\inf$ and $\sup$ as necessary.

\begin{lemma}[Bounding box by convex optimization]
    \label{lemma:bounding-box}
    Fix domain $[0,1]^d$, and let $\delta, \rho \in (0,1)$. There exists a randomized algorithm
    \textsc{BoundingBox} which, on input $(H,h)$ from some unknown clustering $(\Gamma, \rep) \in
    \CONV_\delta$, where $h$ is the representative of $H$, and access to $H$ is given implicitly via
    label queries while $h$ is given explicitly, makes $O(d^3 \log^2(d/\delta) \log(d/\rho))$ label
    queries and outputs a pair of closed axis-aligned boxes $B \subseteq B' \subset \bR^d$
    satisfying, with probability at least $1-\rho$,
    \begin{enumerate}
        \item $B \subseteq \interior B^*(H)$;
        \item $B^*(H) \subseteq \interior B'$; and
        \item $\side_i(B') \le 2 \cdot \side_i(B)$ for each $i \in [d]$.
    \end{enumerate}
\end{lemma}
\begin{proof}
    Let $\delta' \define \delta/8$. The algorithm proceeds as follows:
    \begin{enumerate}
        \item For each $i \in [d]$,
            \begin{enumerate}
                \item Using standard boosting arguments, the algorithm uses label queries to
simulate membership queries to $H$, and takes the median of
$O(\log(d/\rho))$ executions of the convex optimization algorithm of \cref{thm:convex-opt} with
convex set $H$, objective $f(x) = x_i - h_i$, and error parameter $\epsilon = \delta'$, to obtain an
estimate $\bm{\ell}_{i,-}$ that satisfies $\bm{\ell}_{i,-} = \ell^*_{i,-} \pm \delta$ with
probability at least $1 - \frac{\rho}{2d}$ where $\ell^*_{i,-} \define |\inf_{x \in H} x_i - h_i |$.

                \item Similarly, using another $O(\log(d/\rho))$ executions of the convex
optimization algorithm with objective $f(x) = h_i - x_i$, obtain an estimate $\bm{\ell}_{i,+}$ that
satisfies $\bm{\ell}_{i,+} = \ell^*_{i,+} \pm \delta$, where $\ell^*_{i,+} \define |\inf_{x \in H}
h_i - x_i |$.
            \end{enumerate}
        \item Output the boxes $B$ and $B'$ given by
            \[
                B = \left\{ x \in \bR^d : \forall i \in [d] ,\,
                    h_i - \bm{\ell}_{i,-} + 2\delta' \le x_i \le h_i + \bm{\ell}_{i,+} - 2\delta'
                \right\}
            \]
            and
            \[
                B' = \left\{ x \in \bR^d : \forall i \in [d] ,\,
                    h_i - \bm{\ell}_{i,-} - 2\delta' \le x_i \le h_i + \bm{\ell}_{i,+} + 2\delta'
                \right\}
            \]
    \end{enumerate}
    Each execution of the convex optimization algorithm uses $O\left(d^2 \log^2\left(\frac{d \cdot
    \sqrt{d}}{\delta' \cdot \delta} \right) \right)$ label queries by \cref{thm:convex-opt}, where
    we have plugged in $R = O(\sqrt{d})$ for the radius of the ball containing $[0,1]^d$. Each of
    the $d$ iterations uses $O(\log(d/\rho))$ executions, so the query complexity claim follows. We
    now show correctness.

    We first show that $B \subseteq \interior B^*(H)$. For any $y \in B$ and for each $i \in [d]$,
    we have
    \[
        y_i \le h_i + \bm{\ell}_{i,+} - 2\delta'
        \le h_i + \ell^*_{i,+} + \delta' - 2\delta'
        < h_i + \ell^*_{i,+}
        = h_i + \sup_{x \in H} x_i - h_i
        = \sup_{x \in H} x_i \,,
    \]
    and similarly $y_i > \inf_{x \in H} x_i$. Since this holds for every $i \in [d]$, we conclude
    that $y \in \interior B^*(H)$, as desired.

    Next, we show that $B^*(H) \subseteq \interior B'$. For any $y \in B^*(H)$ and for each $i \in
    [d]$, we have
    \[
        y_i
        \ge \inf_{x \in H} x_i
        = h_i - \abs*{\inf_{x \in H} x_i - h_i}
        = h_i - \ell^*_{i,-}
        \ge h_i - \bm{\ell}_{i,-} - \delta'
        > h_i - \bm{\ell}_{i,-} - 2\delta' \,,
    \]
    and similarly $y_i < h_i + \bm{\ell}_{i,+} + 2\delta'$. Thus $y \in \interior B'$, as desired.

    To verify the third requirement, let $e_i$ denote the $i$-th basis vector and note that, for
    each $i \in [d]$, we have that $h - \delta e_i, h + \delta e_i \in H$ by definition of
    $\CONV_\delta$, and hence $\ell^*_{i,+}, \ell^*_{i,-} \ge \delta = 8\delta'$. Therefore
    \[
        \frac{\side_i(B')}{\side_i(B)}
        = \frac{\bm{\ell}_{i,+} + \bm{\ell}_{i,-} + 4\delta'}
               {\bm{\ell}_{i,+} + \bm{\ell}_{i,-} - 4\delta'}
        \le \frac{\ell^*_{i,+} + \ell^*_{i,-} - 2\delta' + 4\delta'}
                 {\ell^*_{i,+} + \ell^*_{i,-} - 2\delta' - 4\delta'}
        = \frac{\ell^*_{i,+} + \ell^*_{i,-} + 2\delta'}{\ell^*_{i,+} + \ell^*_{i,-} - 6\delta'}
        \le \frac{18\delta'}{10\delta'}
        < 2 \,.
    \]
    where in the first inequality we used the lower bounds on the random variables along with the
    fact that $\frac{a+b}{a-b} \le \frac{a'+b}{a'-b}$ when $b < a' \le a$, and similarly for the
    second inequality.
\end{proof}

\noindent
We now use this subroutine to obtain cell rejection and discovery algorithms.

\begin{lemma}
    \label{lemma:qreject-conv-conv}
    Let $\dist(\cdot,\cdot)$ be the normalized $\ell_p$ metric on $[0,1]^d$ with $p \ge 1$, and let
    $\epsilon, \delta, \rho \in (0, 1)$. Then
    \[
        \qreject(\CONV_\delta, \CONV_\delta, \epsilon/d^{1/p}, 2\epsilon, \rho)
        = O(d^3 \log^2(d/\delta) \log(d/\rho)) \,.
    \]
\end{lemma}
\begin{proof}
    The algorithm is the following: simulate \textsc{BoundingBox} on input $(H,h)$ to obtain output
    $B \subseteq B'$; accept if $\diam_\dist(B) \le \epsilon$ and reject otherwise.

    Suppose \textsc{BoundingBox} succeeds as per \cref{lemma:bounding-box}, which occurs with
    probability at least $1-\rho$. We verify the requirements of \cref{def:cell-rejection}.

    The first requirement is trivially satisfied, since we never output $\bot$. Now, suppose
    $\diam_\dist(H) \le \epsilon/d^{1/p}$, so that for any $x, y \in H$ and $i \in [d]$, we have
    $|x_i - y_i| \le \|x-y\|_p = d^{1/p} \dist(x,y) \le d^{1/p} \cdot \epsilon / d^{1/p} =
    \epsilon$, and thus $\side_i(B^*(H)) \le \epsilon$. Then $\diam_\dist(B) \le \diam_\dist(B^*(H))
    \le \epsilon$, the first inequality since $B \subseteq \interior B^*(H)$ and the second since
    \[
        \diam_\dist(B^*(H))
        = \frac{1}{d^{1/p}} \cdot \left( \sum_{i=1}^d \side_i(B^*(H))^p \right)^{1/p}
        \le \frac{1}{d^{1/p}} \cdot \left( d \cdot \epsilon^p \right)^{1/p}
        = \epsilon \,.
    \]
    Thus the algorithm accepts.

    \sloppy
    On the other hand, suppose $\diam_\dist(H) > 2\epsilon$. Then $\diam_\dist(B) \ge \frac{1}{2}
    \diam_\dist(B') \ge \frac{1}{2} \diam_\dist(B^*(H)) \ge \frac{1}{2} \diam_\dist(H) > \epsilon$,
    the first inequality since $\side_i(B') \le 2 \cdot \side_i(B)$ and the second since $B^*(H)
    \subseteq \interior B'$. Thus the algorithm rejects.
\end{proof}

If we are only required to accept clusterings made of axis-aligned boxes, then we may obtain a
tighter gap between the diameters we guarantee to accept and reject, as well as a cell discovery
procedure (which enables the uniformity testing application).

\begin{lemma}
    \label{lemma:qreject-conv-box}
    Let $\dist(\cdot,\cdot)$ be the normalized $\ell_p$ metric on $[0,1]^d$ with $p \ge 1$, and let
    $\epsilon, \delta, \rho \in (0, 1)$. Then
    \[
        \qreject(\CONV_\delta, \BOX_\delta, \epsilon, 2\epsilon, \rho)
        = O(d^3 \log^2(d/\delta) \log(d/\rho)) \,.
    \]
\end{lemma}
\begin{proof}
    The algorithm is the following: simulate \textsc{BoundingBox} on input $(H,h)$ to obtain output
    $B \subseteq B'$. Query $\LABEL(\Gamma, \rep)$ on two opposite corner points of $B$ (namely
    corner points that differ on every coordinate, chosen arbitrarily), say $x, y$, and output
    $\bot$ if $\gamma(x) \ne \gamma(h)$ or $\gamma(y) \ne \gamma(h)$. Otherwise, accept if
    $\diam_\dist(B) \le \epsilon$ and reject if $\diam_\dist(B) > \epsilon$.

    Suppose \textsc{BoundingBox} succeeds as per \cref{lemma:bounding-box}, which occurs with
    probability at least $1-\rho$. We verify the requirements of \cref{def:cell-rejection}.

    First, suppose $(\Gamma, \rep) \in \BOX_\delta$. Then $\interior H = \interior B^*(H)$, and
    since $B \subseteq \interior B^*(H)$ we conclude that the corner points of $B$ belong to $H$.
    Thus the algorithm does not output $\bot$.

    Now, suppose $\diam_\dist(H) \le \epsilon$ and suppose the algorithm does not output $\bot$.
    Then $H$ contains two opposite corner points of $B$, so $\diam_\dist(B) \le \diam_\dist(H) \le
    \epsilon$, and the algorithm accepts.

    On the other hand, suppose $\diam_\dist(H) > 2\epsilon$ and suppose the algorithm does not
    output $\bot$. Then as in \cref{lemma:qreject-conv-conv}, we have $\diam_\dist(B) \ge
    \frac{1}{2} \cdot \diam_\dist(B') \ge \frac{1}{2} \cdot \diam_\dist(B^*(H)) \ge \frac{1}{2}
    \cdot \diam_\dist(H) > \epsilon$, so the algorithm rejects.
\end{proof}

\begin{lemma}
    \label{lemma:qcell-conv-box}
    Let $\dist(\cdot,\cdot)$ be the normalized $\ell_p$ metric on $[0,1]^d$ with $p \ge 1$, and let
    $\epsilon, \delta, \rho \in (0, 1)$. Let $\nu$ be the uniform distribution over $[0,1]^d$. Then
    \[
        \qcell(\CONV_\delta, \BOX_\delta, 2^{-d}, \nu, \rho)
        = O(d^3 \log^2(d/\delta)\log(d/\rho)) \,.
    \]
\end{lemma}
\begin{proof}
    The algorithm is the following: simulate \textsc{BoundingBox} on input $(H,h)$ to obtain output
    $B \subseteq B'$. Output $B'$.

    Suppose \textsc{BoundingBox} succeeds as per \cref{lemma:bounding-box}, which occurs with
    probability at least $1-\rho$. We verify the requirements of \cref{def:cell-discovery}.

    First, the output $B'$ indeed satisfies $H \subseteq B'$, since $B^*(H) \subseteq \interior B'$
    by \cref{lemma:bounding-box}.

    In addition, suppose $(\Gamma, \rep) \in \BOX_\delta$. Then $\interior H = \interior B^*(H)$,
    and since $B \subseteq \interior B^*(H)$ and $\side_i(B') \le 2 \cdot \side_i(B)$ for each $i
    \in [d]$, we conclude that $\nu(B') \le 2^d \nu(B) \le 2^d \nu(B^*(H)) = 2^d \nu(H)$, and thus
    $\nu(H) \ge 2^{-d} \nu(B')$ as desired.
\end{proof}

We conclude results for testing uniformity and equivalence in this setting. Although our general
\cref{lemma:metric-general} is stated for finite metric spaces, the same strategy works more
generally as long as the EMD inequality corresponding to \cref{lemma:emd-tv-diameter} holds. We
remark in \cref{appendix:emd} that this is indeed the case for $[0,1]^d$.

\begin{restatable}{theorem}{thmintrocvbunif}
    \label{thm:intro-cv-b-unif}
    Fix domain $[0,1]^d$, let $\dist(\cdot,\cdot)$ be the normalized $\ell_p$ metric with $p \ge 1$,
    and let $\epsilon, \delta \in (0,1)$. Then $(\CONV_\delta, \BOX_\delta,
    \epsilon/16)$-diameter-guarded $\epsilon$-testing uniformity of distributions on $[0,1]^d$ under
    $\EMD_\dist$ requires at most
    \[
        m(\epsilon)
        = 2^{O(d)} \cdot \widetilde O\left( \epsilon^{-\max\left\{ 2, \frac{d}{2} \right\}} \right)
    \]
    samples and
    \[
        q(\epsilon) = 2^{O(d)} \cdot \widetilde O\left( e^{-\max\left\{ 2, \frac{d}{2} \right\}}
                                        \log^2\left( \frac{1}{\delta} \right) \right)
    \]
    queries; and $(\CONV_\delta, \BOX_\delta, \epsilon/16)$-diameter-guarded $\epsilon$-testing
    equivalence of distributions requires at most
    \[
        m(\epsilon)
        = 2^{O(d)} \cdot \widetilde O\left( \epsilon^{-\max\left\{ 2, \frac{2d}{3} \right\}} \right)
    \]
    samples and
    \[
        q(\epsilon)
        = O\left( \frac{1}{\epsilon}
            \cdot d^3 \log^2\left( \frac{d}{\delta} \right)
            \log\left( \frac{d}{\epsilon} \right) \right)
    \]
    queries, where the $\widetilde O(\cdot)$ notation hides polylogarithmic factors in $1/\epsilon$.
\end{restatable}
\begin{proof}
    Combine \cref{lemma:metric-general,lemma:emd-testing-hypergrid} (via
    \cref{remark:continuous-cube}) and \cref{lemma:qreject-conv-box}; and for uniformity testing,
    also \cref{lemma:qcell-conv-box}.
\end{proof}

\subsubsection{Convex Cells in $[0,1]^d$}

The tools and results above also yield an equivalence tester for convex cell clusterings, with a
larger gap between the diameters we require to accept and reject compared to the previous
application.

\begin{restatable}{theorem}{thmintrocvcv}
    \label{thm:intro-cv-cv}
    Fix domain $[0,1]^d$, let $\dist(\cdot,\cdot)$ be the normalized $\ell_p$ metric with $p \ge 1$,
    and let $\epsilon, \delta \in (0,1)$. Then $\left(\CONV_\delta, \CONV_\delta,
    \frac{\epsilon}{16d^{1/p}}\right)$-diameter-guarded $\epsilon$-testing equivalence of
    distributions requires at most
    \[
        m(\epsilon)
        = 2^{O(d)} \cdot \widetilde O\left( \epsilon^{-\max\left\{ 2, \frac{2d}{3} \right\}} \right)
    \]
    samples and
    \[
        q(\epsilon)
        = O\left( \frac{1}{\epsilon}
            \cdot d^3 \log^2\left( \frac{d}{\delta} \right)
            \log\left( \frac{d}{\epsilon} \right) \right)
    \]
    queries, where the $\widetilde O(\cdot)$ notation hides polylogarithmic factors in $1/\epsilon$.
\end{restatable}
\begin{proof}
    Combine \cref{lemma:metric-general,lemma:emd-testing-hypergrid} (via
    \cref{remark:continuous-cube}) and \cref{lemma:qreject-conv-conv}.
\end{proof}

\section{Part II: Random Clustering}

Having established bounds for worst-case adversarial clusterings in Part I, we now consider random
clusterings. We will prove two results:
\begin{enumerate}
\item An upper bound on testing uniformity of distributions over the vertices of paths and cycles
with \emph{zero} label queries (\cref{thm:confused-collector-main}
in \cref{sec:upper-bound-confused-collector}). The proof of this theorem spans 
\cref{section:random-expectation,section:random-concentration,section:random-zero-query}.
\item An improved bound in the same setting when queries are allowed
(\cref{thm:technical-random-with-queries}), proved in \cref{section:random-with-queries}.
\end{enumerate}

\subsection{Preliminaries}

In Part II, we will consider the domain to be $\bZ_n$, the integers mod $n$, which we interpret as
the vertices of either a path or a cycle. We will always write $\mu$ for an arbitrary probability
distribution (which will be the input distribution to the tester), and write $\nu$ for the uniform
distribution.

Let us recall the random clustering model. Let $G = (V,E)$ be a path or cycle on $n$ vertices
$\bZ_n$. Let $\rho \in (0,1]$ be the ``resolution'' parameter. It will be convenient to also define
$\eta \define 1 - \rho$. Then we will define a class $\cU_\rho$ of random clusterings. A random
clustering $(\bm \Gamma, \bm \rep) \sim \cU_\rho$ is constructed as follows:
\begin{enumerate}
\item Each edge $e \in E$ is deleted independently with probability $\rho$ to obtain a random
subgraph $\bm{H}$ of $G$.
\item The clustering $(\bm{\Gamma}, \bm{\rep})$ is defined as $\bm{\Gamma} = \{ \bm{\Gamma_1},
\dotsc, \bm{\Gamma}_{\bm{k}} \}$ where $\bm{k}$ is the number of connected components of $\bm{H}$
and each $\bm{\Gamma_i}$ is a (maximal) connected component. Then $\bm{\rep}(\bm{\Gamma}_i)$ is
selected uniformly at random from $\bm{\Gamma}_i$ (the specific choice of $\bm{\rep}$ is not
important).
\end{enumerate}
Observe that a resolution $\rho = 1$ means that $\bm{H}$ has no edges (with probability 1) and
therefore each $\bm{\Gamma}_i$ is a singleton -- every element of the domain is distinguished
perfectly. We define the conditions that must hold for a uniformity tester in the confused collector
model with these random clusterings:

\begin{definition}[Distribution Testing with Random Clustering]
\label{def:testing-random-clusterings}
Let $G = (V,E)$ be either a cycle or a path and let $\epsilon, \rho, \delta \in (0,1)$. We say an
algorithm $A$ is an $(\epsilon, \rho)$-uniformity tester for $G$, with error probability $\delta$,
if it is given clustered-sample and label oracle access to inputs $(\bm{\Gamma}, \bm{\rep})$ and
$\mu$, with $\mu$ being an arbitrary distribution over $V$ and $(\bm{\Gamma}, \bm{\rep}) \sim
\cU_\rho$ defined above, and satisfies
\begin{enumerate}
\item If $\mu=\nu$ then $\Pr{ \bm{A} \text{ accepts}} \geq 1 - \delta$; and
\item If $\dist_\TV(\mu,\nu) > \epsilon$ then $\Pr{ \bm{A} \text{ rejects}} \geq 1-\delta$,
\end{enumerate}
where the probabilities are over the randomness of $A$, the responses to the oracle calls, and the
random clustering $(\bm{\Gamma}, \bm{\rep}) \sim \cU_\rho$.
\end{definition}

\subsubsection{Notation for Random Subgraphs}
\label{subsection:notation-random-subgraphs}

It is convenient to introduce a shared vocabulary for analyzing the cycle and the path. We will
label the $n$ vertices of the cycle with the set $\bZ_n$ of integers mod $n$, and we will also label
the \emph{edges} of the cycle with the set $\bZ_n$, so that edge $i$ connects vertices $i$ and $i+1$
(with arithmetic mod $n$).  We will treat the path on $n$ vertices as the subgraph of the cycle that
excludes edge $n-1$ connecting vertices labeled $0$ and $n-1$.  When the subgraph $H$ contains edge
$e$, we will sometimes abuse notation and write $e \in H$.

\newcommand{\llangle}{\langle\!\langle}
\newcommand{\rrangle}{\rangle\!\rangle}

A \textbf{circular interval} is a tuple $\llangle i, d \rrangle$ where $i \in \bZ_n$ and $d \in
\bZ$. If $d \geq 0$, we define the elements $\cE\llangle i, d \rrangle$ as the multiset of elements
starting at vertex $i \in \bZ_n$ and containing the $d-1$ elements ``clockwise'' from $i$, \ie the
multiset $\{ i, i+1, i+2, \dotsc, i+d-1 \}$, where addition is mod $n$. Note that for $d=1$ this
contains only $i$, while for $d > n$ this contains some elements with multiplicity greater than
1. If $d < 0$, we define the elements $\cE\llangle i, d \rrangle$ as the multiset of elements
starting at vertex $i \in \bZ_n$ and containing the $|d|-1$ elements ``counter-clockwise'', \ie the
multiset $\{i, i-1, i-2, \dotsc, i-|d|+1\}$.

The \textbf{endpoints} of $\llangle i, d \rrangle$ are the integers $i$ and $i + d - 1$ if $d \geq
0$, or $i+d+1$ and $i$ if $d < 0$.  We will often drop the $\cE$ from the notation, and equivocate
between the tuple $\llangle i, d \rrangle$ and its multiset of elements, so that we write $x \in
\llangle i, d \rrangle$ instead of $\cE \llangle i, d \rrangle$. However, a circular interval is
\emph{not} identified with its multiset of elements; for example, the circular intervals $\llangle
i, n \rrangle$ and $\llangle i+1, n \rrangle$ both contain the same elements $\bZ_n$, but they have
different endpoints.

For a circular interval $I$ and a vector $u : \bZ \to \bR$, we define
\[
  u[I] \define \sum_{s \in I} u_s \,,
\]
where we note that $s$ may occur multiple times in $I$ and $u_s$ is counted each time.

For a circular interval $\llangle i, d \rrangle$, we will define the circular interval $\llangle i,
d \rrangle^*$ to be the integers corresponding to the \emph{edges} induced by the vertices
$\llangle i, d \rrangle$; specifically
\[
\llangle i, d \rrangle^* = \begin{cases}
  \emptyset &\text{ if } d = 0 \\
  \llangle i, d-1 \rrangle &\text{ if } d \geq 1 \\
  \llangle i-1, 1-|d| \rrangle &\text{ if } d \leq -1 \,.
\end{cases}
\]
For any $s \in \bZ_n$, we say that a circular interval $I$ \textbf{crosses} $s$ if $s \in I^*$; \ie
$s$ is an edge between two vertices in $I$.

Fix any subgraph $H$ of the cycle (or path), and suppose that $H$ has $b$ connected components; note
that each connected component is a circular interval. We define the \textbf{cells} induced by $H$
as $\Gamma_1, \dotsc, \Gamma_n$ such that $\Gamma_1, \dotsc, \Gamma_b$ are the connected components
of $H$, while $\Gamma_{b+1}, \dotsc, \Gamma_n = \emptyset$. For each vertex $i \in \bZ_n$, we define
\[
  \gamma(i) \define t \text{ such that } i \in \Gamma_t \,.
\]
We say that two vertices $i,j$ are \textbf{joined} if $\gamma(i) = \gamma(j)$, and we define the
\textbf{join matrix} $\Phi = \Phi(H)$ as
\[
  \Phi_{i,j} \define \begin{cases}
    1 &\text{ if } \gamma(i) = \gamma(j) \\
    0 &\text{ otherwise.}
  \end{cases}
\]
We define a \textbf{join function} $J$ such that for any circular interval
$I = \llangle i,d \rrangle$,
\[
    J(I) \define \ind{ \forall e \in \llangle i,d \rrangle^* : e \in H } \,.
\]
Thus if $J(I) = 1$, then for every $i, j \in \cE(I)$ we have $\Phi_{i,j} = 1$.

For the path, the circular intervals that cross the edge between vertices $0$ and $n-1$ 
are irrelevant, so it is convenient to define $\cI^\cyc$ as the set of all circular intervals, and
$\cI^\pth$ as the set of all circular intervals that do not cross edge $n-1$.

We will use $\cI \in \{ \cI^\cyc, \cI^\pth \}$ to denote the set of circular intervals relevant to
the analysis. In the case $\cI = \cI^\cyc$, each pair of vertices $i \le j$ has two disjoint paths
connecting them and therefore may be joined together in two ways. We define $\smallinterval(i,j)$
and $\largeinterval(i,j)$ as the two circular intervals defined as follows. Let
\[
  I_1 \define \llangle i, j - i + 1 \rrangle
  \qquad\text{ and }\qquad
  I_2 \define \llangle j, n - (j-i) + 1 \rrangle
\]
as the circular intervals corresponding to the two separate paths between $i$ and $j$. Then we
define
\[
  \smallinterval(i,j) \define \arg \max_{I \in \{I_1, I_2\}} \Ex{ \bm{J}[ I ] }
  \qquad\text{ and }\qquad
  \largeinterval(i,j) \define \arg \min_{I \in \{I_1, I_2\}} \Ex{ \bm{J}[ I ] } \,,
\]
breaking ties arbitrarily. Note that, in the case of the path, we will have only one way of joining
$i$ and $j$, so that $\Ex{ \bm{J}[ \largeinterval(i,j) ] } = 0$ in this case. Symmetrically, when
$i > j$ we define $\smallinterval(i,j) \define \smallinterval(j,i)$ and
$\largeinterval(i,j) \define \largeinterval(j,i)$.

For $\cI \in \{\cI^\cyc, \cI^\pth\}$, we define
\[
  \zeta(\cI) \define \max_{i,j} \Ex{ \bm{J}[ \largeinterval(i,j) ] } \,.
\]

\begin{proposition}
\label{prop:zeta-cycle-bound}
$\zeta(\cI^\pth) = 0$ and $\zeta(\cI^\cyc) \leq (1-\rho)^{n/2}$.
\end{proposition}
\begin{proof}
For distinct $i, j \in \bZ_n$, define $I_1$ and $I_2$ as above, and note that for $a \in \{1,2\}$,
\[
  \Ex{ \bm{J}[I_a] } = (1-\rho)^{|I_a^*|} \,.
\]
Since $I_1, I_2$ partition $\bZ_n$, so we have either $|I_1| \geq n/2$ or $|I_2| \geq n/2$, so the
minimum is at most $(1-\rho)^{n/2}$.
\end{proof}
\subsubsection{Notation for Clustered Samples}
\label{subsection:notation-clustered-samples}

For a fixed sample (\ie multiset) $S \subset \bZ_n$ and for $i \in \bZ_n$, write $T_i$ for the
multiplicity of element $i$ in $S$. We then define for each $i \in [n]$ the variable
\[
  X_i \define \sum_{j \in \bZ_n : \gamma(j) = i} T_j \,,
\]
which is the total multiplicity of elements from cell $\Gamma_i$ that occur in $S$.

Observe that the above variables depend on the subgraph $H$ and the sample $S$. For a random
subgraph $\bm H$ chosen according to the confused collector sampling procedure, and a random sample
$\bm S$ of vertices, we write the above variables in bold to denote the random variables depending
on $\bm H$ and $\bm S$. We will then write
\[
  \phi \define \Ex{\bm \Phi} \,,
\]
and observe that
\[
  \phi_{i,j} = \Pr{ \bm{\gamma}(i) = \bm{\gamma}(j) } \,.
\]
We will write $\bPhi$ for the random join matrix and $\phi$ for its expectation when
statements hold for both the path and the cycle. Otherwise, we will specify $\bPhi^\pth$ or
$\bPhi^\cyc$.

We will use the standard Poissonization technique to simplify the analysis. We write $m$ for the
sample-size parameter. For each $j \in \bZ_n$, let $\bm{T}_j$ be the random variable counting the
number of sample points obtained from element $j$. Using the Poissonization technique, we have
\[
  \bm{T}_j \sim \Poi(m \mu_j)
\]
for each $j \in \bZ_n$, and therefore
\[
  \bm{X}_i \sim \Poi( m \cdot \mu[\bm \Gamma_i ] ) 
\]
for each $i \in [n]$.

\subsection{Expectation of the Test Statistic}
\label{section:random-expectation}

The main step in our algorithm for testing uniformity will be to accept or reject depending on
whether the \emph{test statistic} surpasses some threshold. The test statistic is:

\begin{definition}[Test Statistic]
For a fixed parameter $m$ and random variables defined as above, we define
\[
  \bm Y \define \frac{1}{m^2} \sum_{i=1}^n \bm{X}_i (\bm{X}_i - 1) \,.
\]
By expanding the variables $\bm{X}_i$, the test statistic may be written as the quadratic form
\[
  \bm Y = \frac{1}{m^2} \left( \bm{T}^\top \bm{\Phi} \bm{T} - \| \bm T \|_1 \right) \,.
\]
\end{definition}

We start by giving an expression for the expectation of the statistic $\bm{Y}$. Recall that we write
$\mu : \bZ_n \to [0,1]$ for the distribution over the vertices (of either the path or the
cycle), $m$ is the sample-size parameter, and $\phi = \Ex{\bm \Phi}$.

\begin{proposition}
    \label{prop:expectation-y-quadratic-form}
    \label{eq:expectation-y-quadratic-form-confused-collector}
    The statistic $\bm{Y}$ satisfies $\Ex{\bm{\bm{Y}}} = \mu^\top \phi \mu$.
\end{proposition}
\begin{proof}
    We use the facts that $\bm{T}$ and $\bPhi$ are independent and that, for $i \neq j$,
    $\bm{T}_i$ and $\bm{T}_j$ are independent. We will also use the property that, for
    $\bm{Z} \sim \Poi(\lambda)$, we have $\Ex{\bm{Z}} = \Var{\bm{Z}} = \lambda$ and, therefore,
    $\Ex{\bm{Z}^2} = \Ex{\bm{Z}} + \Ex{\bm{Z}}^2$. We obtain:
    \begin{align*}
        \Ex{\bm{Y}} &= \frac{1}{m^2} \left( \Ex{\bm{T}^\top \bPhi \bm{T}} - \Ex{\|\bm{T}\|_1} \right)
        = \frac{1}{m^2} \sum_{i=0}^{n-1} \sum_{j=0}^{n-1} \Ex{\bm{T}_i \bm{T}_j} \Ex{\bPhi_{i,j}}
            - \frac{1}{m^2} \sum_{i=0}^{n-1} \Ex{\bm{T}_i} \\
        &= \frac{1}{m^2} \sum_{i=0}^{n-1} \sum_{j=0}^{n-1} \Ex{\bm{T}_i} \Ex{\bm{T}_j} \phi_{i,j}
            + \frac{1}{m^2} \sum_{i=0}^{n-1} \Ex{\bm{T}_i} - \frac{1}{m^2} \sum_{i=0}^{n-1} \Ex{\bm{T}_i}
        = \frac{1}{m^2} \sum_{i=0}^{n-1} \sum_{j=0}^{n-1} (m \mu_i) (m \mu_j) \phi_{i,j} \\
        &= \mu^\top \phi \mu \,. \qedhere
    \end{align*}
\end{proof}

We would like to show that $\Ex{\bm{Y}}$ is small when $\mu = \nu$ is uniform, and large when $\mu$
is far from uniform.  Letting $\bm{Y}^{(\nu)}$ denote the test statistic when $\mu$ is the uniform
distribution, we have:
\begin{proposition}[Expectation of $\bm{Y}$ in the uniform case]
    \label{prop:expectation-y-uniform-confused-collector}
    When $\mu = \nu$, $\bm{Y} = \bm{Y}^{(\nu)}$ satisfies
    \[
        \Ex{\bm{Y}^{(\nu)}} = \nu^\top \phi \nu = \frac{1}{n^2} \sum_{i,j} \phi_{i,j} \,.
    \]
\end{proposition}
\begin{proof}
    The claim follows from \cref{eq:expectation-y-quadratic-form-confused-collector} and the
    assumption that $\mu = \nu = \vec 1 / n$.
\end{proof}

\noindent
We must now show that $\Ex{\bm{Y}}$ is larger than this threshold when $\mu$ is far from uniform.
\begin{proposition}
    \label{prop:quadratic-form-decomposition}
    Write $\mu = \nu + z$. Then $\bm{Y}$ satisfies
    \[
        \Ex{\bm{Y}} =
        \nu^\top \phi \nu + 2 \nu^\top \phi z + z^\top \phi z \,.
    \]
\end{proposition}
\begin{proof}
    This follows immediately from \cref{prop:expectation-y-quadratic-form} by expanding the
    quadratic form and recalling that $\bPhi$ is always a symmetric matrix, and hence so is
    $\phi = \Ex{\bPhi}$.
\end{proof}
\noindent
Therefore, we must find a lower bound on $2 \nu^\top \phi z + z^\top \phi z$.  Our strategy will
be:
\begin{enumerate}
\item Show that the minimum eigenvalue of $\phi$ is large, and hence so is $z^\top \phi z$ when
$\|z\|_2^2$ is large; we will do this separately for the path and the cycle.
\item Show that the term $\nu^\top \phi z$ is small in absolute value, so it does not affect the sum
too much.
\end{enumerate}
Both $\phi^\pth$ and $\phi^\cyc$ enjoy nice properties (they are a Toeplitz and a circulant matrix,
respectively), and we bound the minimum eigenvalue of each in turn. Let ${\lambda_{\min}}(\cdot),
{\lambda_{\max}}(\cdot)$ denote the minimum and maximum eigenvalues of a (real symmetric) matrix,
respectively.

\subsubsection{Expectation for the Path}
When $G$ is the path, the expected join matrix $\phi^\pth \define \Ex{\bPhi^\pth}$ has a simple
formulation in terms of $\rho$. Recall the notation $\eta \define 1 - \rho$.

\begin{proposition}
    \label{prop:phi-pth}
    The matrix $\phi^\pth$ is given by
    \[
        \phi^\pth_{i,j} = \eta^{\abs{i-j}}
    \]
    for each $i, j \in \bZ_n$.
\end{proposition}
\begin{proof}
    Here, the relevant intervals are $\cI = \cI^\pth$. Hence,
    for any $i < j$, we have that $i$ and $j$ are in the same cell if and only if every edge
    between them is in $\bm{H}$:
    \[
        \phi^\pth_{i,j} = \Pr{\bPhi^\pth_{i,j}=1}
        = \Pr{\forall e \in \llangle i, j-i+1 \rrangle^* : e \in \bm{H}}
        = \eta^{j-i} \,. \qedhere
    \]
\end{proof}

\begin{lemma}[Minimum eigenvalue of $\phi^\pth$]
    \label{lemma:min-eigenvalue-pth}
    Let $\rho \in (0,1]$. Then $\lambda_{\min}(\phi^\pth) > \rho/2$.
\end{lemma}
\begin{proof}
The matrix $\phi^\pth$ is a symmetric Toeplitz matrix,
and its inverse may be found as in \cite{Sra}.
Recall that $\eta = 1 - \rho$, so that $0 \le \eta < 1$ and $\phi^\pth_{i,j} =
\eta^{\abs{i-j}}$ by \cref{prop:phi-pth}.
Then the inverse of $\phi^\pth$ (written $\phi^{-1}$ for short)
is the following tridiagonal matrix:
    \[
        \phi^{-1} = \frac{1}{1-\eta^2} \cdot
            \begin{bmatrix}
                1       & -\eta         & 0            & 0            & \dotsm    & 0       \\
                -\eta    & 1 + \eta^2    & -\eta         & 0            & \dotsm    & 0       \\
                0       & -\eta         & 1 + \eta^2    & -\eta         & \dotsm    & 0       \\
                        & \ddots       & \ddots       & \ddots       & \ddots    &         \\
                0       & \dotsc       & \dotsc       & -\eta         & 1 + \eta^2 & -\eta    \\
                0       & \dotsc       & \dotsc       & 0            & -\eta      & 1       \\
            \end{bmatrix}
        \,.
    \]
Now, we may upper bound the maximum eigenvalue of $\phi^{-1}$
using the Gershgorin circle theorem:
    \[
        \lambda_{\max}\left(\phi^{-1}\right)
        \le \max_{i \in \bZ_n} \left\{ \phi^{-1}_{i,i}
            + \sum_{j \neq i} \abs*{\phi^{-1}_{i,j}} \right\}
        \le \left(\frac{1}{1-\eta^2}\right) \cdot \left(1 + \eta^2 + 2\eta\right)
        = \frac{1+\eta}{1-\eta} \,.
    \]
    Hence we obtain
    \[
        \lambda_{\min}(\phi^\pth)
        = \frac{1}{\lambda_{\max}\left(\phi^{-1}\right)}
        \ge \frac{1-\eta}{1+\eta}
        > \frac{1-\eta}{2}
        = \rho/2 \,. \qedhere
    \]
\end{proof}

\subsubsection{Expectation for the Cycle}
When $G$ is the cycle, so that the expected join matrix is $\phi^\cyc \define \Ex{\bPhi^\cyc}$,
we need to account for the small and large intervals (in the notation of
\cref{subsection:notation-random-subgraphs}) connecting $i$ and $j$, as follows.

\begin{proposition}
    \label{prop:phi-cyc}
    The matrix $\phi^\cyc$ is given by
    \[
        \phi^\cyc_{i,j} = \eta^{\abs{i-j}} + \eta^{n-\abs{i-j}} - \eta^n \,.
    \]
\end{proposition}
\begin{proof}
    The sets $\smallinterval(i,j)^*$ and $\largeinterval(i,j)^*$ have sizes
    $\min\{\abs{i-j}, n-\abs{i-j}\}$ and $\max\{\abs{i-j}, n-\abs{i-j}\}$, respectively
    (recall they partition the edges of the cycle). By the principle of inclusion-exclusion,
    \begin{align*}
        \phi^\cyc_{i,j}
        &= \Pr{\bPhi^\cyc_{i,j} = 1} \\
        &= \Pr{\bm{J}(\smallinterval(i,j))=1} + \Pr{\bm{J}(\largeinterval(i,j))=1}
            - \Pr{\forall e \in E : e \in \bm{H}} \\
        &= \eta^{\abs{i-j}} + \eta^{n-\abs{i-j}} - \eta^n \,. \qedhere
    \end{align*}
\end{proof}

It is convenient to work with a simplified close approximation for $\phi^\cyc$. Essentially, we wish
to ignore the large intervals and instead work with the matrix $\phi^{\smallinterval}$ given by
\[
    \phi^{\smallinterval}_{i,j} \define \eta^{\abs*{\smallinterval(i,j)}-1} \,.
\]
We will need the observation that $\zeta(\cI)$ is negligibly small in our range of parameter $\rho$.

\begin{proposition}
    \label{prop:zeta-is-small-confused-collector}
    Let $\delta \in (0, 1)$ be a constant. Suppose $\rho \ge \Omega(n^{-\delta})$, and let $K > 0$
    be any constant. Then for all sufficiently large $n$,
    \[
        \zeta(\cI) \le \eta^{n/2} = o(n^{-K}) \,.
    \]
\end{proposition}
\begin{proof}
    The first inequality is \cref{prop:zeta-cycle-bound}. The second one is easy to check:
    \[
        \eta^{n/2} = (1-\rho)^{n/2} \le e^{-\rho \cdot n/2}
        \le e^{-\Omega(n^{-\delta} \cdot n)}
        = e^{-\Omega(n^{1-\delta})}
        = o(n^{-K}) \,. \qedhere
    \]
\end{proof}

We are now ready to lower bound the eigenvalues of $\phi^{\smallinterval}$ and
$\phi^\cyc$. We first lower bound $\lambda_{\min}(\phi^{\smallinterval})$, and then show that the
approximation error is negligible.

\begin{fact}[Eigenvalues of circulant matrices; see \cite{Gra06}]
    \label{fact:eigenvalues-circulant}
    Let $c_0, c_1, \dotsc, c_{n-1} \in \bR$. Then the matrix
    \[
        M = \begin{bmatrix}
            c_0      & c_1     & c_2     & \cdots & c_{n-1} \\
            c_{n-1}  & c_0     & c_1     & \cdots & c_{n-2} \\
            \vdots   & \vdots  & \vdots  & \ddots & \vdots \\
            c_2      & c_3     & c_4     & \cdots & c_1 \\
            c_1      & c_2     & c_3     & \cdots & c_0
        \end{bmatrix}
    \]
    given by $M_{j,k} = c_{(k-j) \mod n}$ has eigenvalues
    \[
        \lambda_\ell = \sum_{k=0}^{n-1} c_k \omega^{\ell k} \qquad \ell = 0, 1, \dotsc, n-1 \,,
    \]
    where $\omega = e^{-\frac{2\pi i}{n}}$ is a primitive $n$-th root of unity.
\end{fact}

\begin{lemma}[Minimum eigenvalue of $\phi^{\smallinterval}$]
    \label{lemma:min-eigenvalue-smallinterval}
    Let $\delta \in (0,1)$ be a constant and suppose $\rho \ge \Omega(n^{-\delta})$. Then for all
    sufficiently large $n$, $\lambda_{\min}(\phi^{\smallinterval}) > \rho/3$.
\end{lemma}
\begin{proof}
    First assume $n$ is odd. For each $k = 0, 1, \dotsc, n-1$, let
    $c_k \define \phi^{\smallinterval}_{0,k}$, so that $\phi^{\smallinterval}$ is a symmetric
    circulant matrix of the form stated in \cref{fact:eigenvalues-circulant}. In particular,
    letting $h \define \lfloor n/2 \rfloor$ for convenience, we have
    \[
        c_k = \begin{cases}
            \eta^k & \text{if } k \le h \\
            \eta^{n-k} & \text{if } k > h \,.
        \end{cases}
    \]
    Therefore for each $\ell = 0, 1, \dotsc, n-1$, the eigenvalue $\lambda_\ell$ is
    \begin{align*}
        \lambda_\ell
        &= \sum_{k=0}^{n-1} c_k \omega^{\ell k}
        = c_0 + \sum_{k=1}^h \eta^k \omega^{\ell k} + \sum_{k=h+1}^{n-1} \eta^{n-k} \omega^{\ell k}
        = 1 + \sum_{k=1}^h \eta^k \omega^{\ell k} + \sum_{k=1}^h \eta^{n - (n-k)} \omega^{\ell (n-k)} \\
        &= 1 + \sum_{k=1}^h \eta^k \omega^{\ell k} + \sum_{k=1}^h \eta^k \omega^{-\ell k}
        = 1 + \frac{\eta \omega^\ell - \eta^{h+1} \omega^{\ell(h+1)}}{1 - \eta \omega^\ell}
            + \frac{\eta \omega^{-\ell} - \eta^{h+1} \omega^{-\ell (h+1)}}{1 - \eta \omega^{-\ell}} \\
        &= \frac{1 - \eta(\omega^\ell + \omega^{-\ell}) + \eta^2 
                +  \eta \omega^\ell - \eta^2
                        - \eta^{h+1} \omega^{\ell(h+1)} + \eta^{h+2} \omega^{\ell h} 
                +  \eta \omega^{-\ell} - \eta^2
                        - \eta^{h+1} \omega^{-\ell(h+1)} + \eta^{h+2} \omega^{-\ell h}}
                {1 - \eta(\omega^\ell + \omega^{-\ell}) + \eta^2} \\
        &= \frac{1  - \eta^2 -\eta^{h+1}(\omega^{\ell(h+1)} + \omega^{-\ell(h+1)})
                    + \eta^{h+2}(\omega^{\ell h} + \omega^{-\ell h})}
                {1 - \eta(\omega^\ell + \omega^{-\ell}) + \eta^2}
        = \frac{(1 - \eta)(1 + \eta) \pm O(\eta^{n/2})}{1 - 2 \eta \cos(2\pi\ell/n) + \eta^2} \,,
    \end{align*}
    where we used the identity $e^{i \theta} + e^{-i \theta} = 2\cos(\theta)$ in the last step.
    Thus, recalling that $\eta = 1-\rho \in [0,1)$, we conclude that $\lambda_\ell$ is lower bounded
    by
    \[
        \frac{(1 - \eta)(1 + \eta) \pm O(\eta^{n/2})}{1 - 2 \eta \cos(2\pi\ell/n) + \eta^2}
        \ge \frac{(1-\eta)(1+\eta)}{(1+\eta)^2} - \frac{O(\eta^{n/2})}{(1-\eta)^2}
        \ge \frac{\rho}{2} - \frac{O(\eta^{n/2})}{(1-\eta)^2} \,.
    \]
    Then, using \cref{prop:zeta-is-small-confused-collector} with $K = 3\delta$,
    \[
        \frac{O(\eta^{n/2})}{(1-\eta)^2}
        = \frac{o(n^{-3\delta})}{\rho^2}
        = o(n^{-3\delta} \cdot n^{2\delta})
        = o(n^{-\delta})
        = o(\rho) \,,
    \]
    and thus $\lambda_\ell > \rho/3$. When $n$ is even, the same argument applies with an extra term
    of order $O(\eta^{n/2})$, which leaves the asymptotic analysis unaffected.
\end{proof}

\newcommand{\err}{\mathsf{err}}

\begin{lemma}
    \label{lemma:min-eigenvalue-cyc}
    Let $\delta \in (0,1)$ be a constant and suppose $\rho \ge \Omega(n^{-\delta})$. Then for all
    sufficiently large $n$, $\lambda_{\min}(\phi^{\cyc}) > \rho/4$.
\end{lemma}
\begin{proof}
    Let $\phi^\err \define \phi^\cyc - \phi^{\smallinterval}$. It is standard to check that
    $\lambda_{\min}(\phi^\cyc) \ge
    \lambda_{\min}(\phi^{\smallinterval}) + \lambda_{\min}(\phi^\err)$.
    Since $\lambda_{\min}(\phi^{\smallinterval}) > \rho / 3$ by
    \cref{lemma:min-eigenvalue-smallinterval}, it suffices to show that
    $\lambda_{\min}(\phi^\err) > -o(\rho)$.
    Since $\phi^\cyc_{i,j} =
    \eta^{\abs{\smallinterval(i,j)}-1} + \eta^{\abs{\largeinterval(i,j)}-1} - \eta^n$ by
    \cref{prop:phi-cyc}, we obtain
    \[
        \phi^\err_{i,j} = \eta^{\abs{\largeinterval(i,j)}-1} - \eta^n
    \]
    for all $i,j \in \bZ_n$. By definition of $\zeta(\cI)$ and recalling
    \cref{prop:zeta-cycle-bound}, we conclude that
    \[
        \|\phi^\err\|_\infty \le \zeta(\cI) \le \eta^{n/2} \,.
    \]
    By the Gershgorin circle theorem and \cref{prop:zeta-is-small-confused-collector},
    \begin{align*}
        \lambda_{\min}(\phi^\err)
        &\ge \min_{i \in \bZ_n} \Big\{ \phi^\err_{i,i} - \sum_{j \ne i} \abs*{\phi^\err_{i,j}} \Big\}
        > -n \cdot \eta^{n/2}
        \ge -n \cdot o(n^{-1 - \delta})
        = -o(n^{-\delta})
        \ge -o(\rho) \,. \qedhere
    \end{align*}
\end{proof}

\subsubsection{Bounds on the Cross Term}
Now we show that the cross term $\nu^\top \phi z$ is small. When $G$ is the cycle, this term will
in fact be zero; when $G$ is the path, the cross term is relevant due to the asymmetry between the
vertices closer to the endpoints or to the middle. However, this will not be a problem as long as
$\|z\|_\infty$ is not too large (which indeed the tester will be able to check).

\begin{restatable}{proposition}{propzdeltacrossterm}
    \label{prop:z-delta-cross-term}
    Let $\rho \in (0,1]$ and let $\phi = \Ex{\bPhi}$ be the corresponding expected join matrix.
    Let $\delta > 0$ be any positive real number. Then for any $z \in \bR^n$ satisfying
    \begin{enumerate}
        \item $\sum_i z_i = 0$; and
        \item $\|z\|_\infty \le \delta$,
    \end{enumerate}
    it is the case that
    \[
        \abs*{\nu^\top \phi z} \le \frac{2\delta}{n \rho^2} \,.
    \]
\end{restatable}

\cref{prop:z-delta-cross-term} is proved by constructing a worst-case vector $z$ whose negative
(respectively positive) entries are as concentrated as possible toward the ends (respectively the
middle) of the vector, subject to the $\ell_\infty$ constraint, and evaluating geometric series to
upper bound the cross term. The proof is elementary, and we defer it to
\cref{appendix:missing-proofs-from-random-expectation}.

\begin{lemma}
    \label{lemma:cross-term-bound}
    Let $C > 0$, and let $\mu$ be a distribution over $\bZ_n$ satisfying $\|\mu\|_\infty \le \frac{C
    \log n}{m}$. Then as long as $\frac{8C \log n}{\rho^3 \epsilon^2} \le m \le C n \log n$, the
    following holds:
    \[
        \abs*{2 \nu^\top \phi z} \le \frac{1}{2n} \epsilon^2 \rho \,.
    \]
\end{lemma}
\begin{proof}
    Write $\mu = \nu + z$. Since $\|\mu\|_1 = \|\nu\|_1 = 1$, it follows that $\sum_i z_i = 0$,
    satisfying the first condition of \cref{prop:z-delta-cross-term}. We will show that $z$ also
    satisfies the second condition with $\delta \define \frac{C \log n}{m}$. Indeed, that $z_i \le
    \frac{C \log n}{m}$ for every $i$ follows immediately from the assumption on $\mu$. On the other
    hand, using the assumption that $m \le C n \log n$ we conclude that $z_i = \mu(i) - \frac{1}{n}
    \ge -\frac{1}{n} \ge -\frac{C \log n}{m}$ and thus $\|z\|_\infty \le \frac{C \log n}{m}$ as
    desired. \cref{prop:z-delta-cross-term} implies that $\abs*{\nu^\top \phi z} \le \frac{2C \log
    n}{m n \rho^2}$, and combining this inequality with our assumed lower bound on $m$ we obtain
    \[
        \abs*{2 \nu^\top \phi z}
        \le 2 \cdot \frac{2C \log n}{m n \rho^2}
        = \frac{8C \log n}{\rho^3 \epsilon^2} \cdot \frac{\epsilon^2 \rho}{2n} \cdot \frac{1}{m}
        \le \frac{1}{2n} \epsilon^2 \rho \,. \qedhere
    \]
\end{proof}

\subsubsection{Gap in Expected Values}
We combine the previous results to show the desired separation in the expected value of $\bm{Y}$.
The first step is to incorporate the eigenvalue bounds: 

\begin{proposition}
    \label{cor:separation-general-form}
    For all sufficiently large $n$, the following holds.
    Let $\mu$ be a distribution over $\bZ_n$ such that $\dist_\TV(\mu, \nu) > \epsilon$, and write
    $\mu = \nu + z$. Then
    \[
        \Ex{\bm{Y}} > \Ex{\bm{Y}^{(\nu)}} + \frac{\rho}{8} \|z\|_2^2
                + \frac{1}{2n} \epsilon^2 \rho + 2 \nu^\top \phi z \,.
    \]
\end{proposition}
\begin{proof}
    Combine \cref{prop:quadratic-form-decomposition,prop:expectation-y-uniform-confused-collector},
    along with the fact that $x^\top M x \ge \lambda_{\min}(M) \|x\|_2^2$ for any symmetric matrix
    $M$ and vector $x$ and the eigenvalue bounds
    \cref{lemma:min-eigenvalue-pth,lemma:min-eigenvalue-cyc} to obtain
    \begin{align*}
        \Ex{\bm{Y}}
        &= \nu^\top \phi \nu + 2 \nu^\top \phi z + z^\top \phi z
        \ge \Ex{\bm{Y}^{(\nu)}} + \lambda_{\min}(\phi) \|z\|_2^2 + 2 \nu^\top \phi z \\
        &\ge \Ex{\bm{Y}^{(\nu)}} + \frac{\rho}{4} \|z\|_2^2 + 2 \nu^\top \phi z
        = \Ex{\bm{Y}^{(\nu)}}
            + \frac{\rho}{8} \|z\|_2^2
            + \frac{\rho}{8} \|z\|_2^2
            + 2 \nu^\top \phi z \,.
    \end{align*}
    Then, since $\|z\|_1 = 2\dist_\TV(\mu,\nu) > 2\epsilon$, we have
    $\|z\|_2^2 > \left(\frac{2\epsilon}{n}\right)^2 \cdot n = 4\epsilon^2 / n$, concluding the
    proof.
\end{proof}

\begin{lemma}[Separation in the expected value of $\bm{Y}$]
    \label{lemma:separation-expected-value-confused-collector}
    Let $C, c > 0$ be constants, and let $m = c \cdot \frac{\sqrt{n}}{\epsilon^2 \rho^{3/2}} \log^2
    n$. Then for all sufficiently large $n$ and all $\epsilon, \rho \in (0, 1]$ satisfying $\rho \ge
    (c/C)^{2/3} \frac{\log^{2/3} n}{n^{1/3} \epsilon^{4/3}}$, the following holds. Suppose $\mu$ is
    a distribution over $\bZ_n$ satisfying $\|\mu\|_\infty \le \frac{C\log n}{m}$ and
    $\dist_\TV(\mu, \nu) > \epsilon$. Write $\mu = \nu + z$. Then the test statistic $\bm{Y}$
    satisfies
    \[
        \Ex{\bm{Y}} > \Ex{\bm{Y}^{(\nu)}} + \frac{\rho}{8} \|z\|_2^2 \,.
    \]
\end{lemma}
\begin{proof}
    By \cref{cor:separation-general-form}, we have
    \[
        \Ex{\bm{Y}} > \Ex{\bm{Y}^{(\nu)}} + \frac{\rho}{8} \|z\|_2^2
                + \frac{1}{2n} \epsilon^2 \rho + 2 \nu^\top \phi z \,.
    \]
    Hence, we will be done if we can show that
    $2 \nu^\top \phi z \ge -\frac{1}{2n} \epsilon^2 \rho$. This will follow immediately from
    \cref{lemma:cross-term-bound} as long as we can verify the preconditions on $m$. We first
    check the lower bound:
    \[
        m \ge \frac{8C \log n}{\rho^3 \epsilon^2}
        \iff c \frac{\sqrt{n}}{\epsilon^2} \cdot \frac{\log^2 n}{\rho^{3/2}}
            \ge \frac{8C \log n}{\rho^3 \epsilon^2}
        \iff \rho \ge (8C/c)^{2/3} \frac{1}{n^{1/3} \log^{2/3} n} \,,
    \]
    which holds for all sufficiently large $n$ by our assumption on $\rho$. As for the upper bound,
    \[
        m \le C n \log n
        \iff c \frac{\sqrt{n}}{\epsilon^2} \cdot \frac{\log^2 n}{\rho^{3/2}} \le C n \log n
        \iff \rho \ge (c/C)^{2/3} \frac{\log^{2/3} n}{n^{1/3} \epsilon^{4/3}} \,,
        \]
    which holds by assumption. Hence \cref{lemma:cross-term-bound} applies and we are done.
\end{proof}

\subsection{Concentration of the Test Statistic}
\label{section:random-concentration}


Our goal in this section is to establish the concentration of the test statistic $\bm{Y}$, by giving
the formal version of \cref{lemma:intro-informal-concentration}. Recall that the (random) partition
of vertices into cells $\bm \Gamma_1, \dotsc, \bm \Gamma_n$ depends on the random subgraph $\bm H$.
We start by noting that we can break down the variance of $\bm{Y}$ into two components by the law of
total variance:
\[
    \Var{\bm{Y}} = \Varu{\bm{H}}{\Exuc{\bm{T}}{\bm{Y}}{\bm{H}}}
                 + \Exu{\bm{H}}{\Varuc{\bm{T}}{\bm{Y}}{\bm{H}}} \,.
\]

\subsubsection{Relative Concentration}

One of the main tools in our analysis will be ``relative concentration''\!\!, which compares the
probability mass of $\mu$ inside the circular intervals $I$ to the size of that interval, weighted
by the resolution $\rho$ -- with the crucial addition that intervals with weight less than some
threshold $t$ essentially have no effect, so this notion differs from the standard $\ell_\infty$
norm in that it ``smooths out'' small perturbations.

\begin{definition}[Relative Concentration]
\label{def:relative-concentration}
Let $\cI \in \{\cI^\cyc, \cI^\pth\}$ and let $\mu : \bZ_n \to \bR_{\geq 0}$ be a probability
distribution. Let $t > 0$. Then we define
\[
    \RelativeConcentrationT{\mu}{t} \define
        \max_{I \in \cI : |I| \le n}
            \frac{\mu[I]}{ \max\left\{ \rho |I^*| , t \right\} } \,.
\]
\end{definition}

We will instantiate this measure with $t = \frac{1}{\log n}$ and use it in the analysis of the
tester in \cref{section:random-zero-query}. We will require the following lemma, which allows us to
find an interval $I$ exhibiting a large ratio between $\mu[I]$ and $\rho |I^*|$ if we assume
high relative concentration $\RelativeConcentrationT{\mu}{t}$. 

\begin{proposition}
\label{lemma:relative-concentration-structural}
Let $\cI \in \{ \cI^\cyc, \cI^\pth \}$, and let $\mu : \bZ_n \to \bR_{\geq 0}$ be a probability
distribution. Then there exists $I \in \cI$ satisfying:
    \begin{enumerate}
        \item $\rho |I^*| \le t$; and
        \item $\mu[I] \ge \frac{1}{2} \cdot t \cdot \RelativeConcentrationT{\mu}{t}$.
    \end{enumerate}
\end{proposition}
\begin{proof}
    By definition of relative concentration, there exists an interval $I \in \cI$ of size at most
    $n$ such that either
    \begin{enumerate}
    \item $\rho |I^*| \le t$ and $\mu[I] = t \cdot \RelativeConcentrationT{\mu}{t}$; or
    \item $\rho |I^*| \ge t$ and $\mu[I] = \rho |I^*| \cdot \RelativeConcentrationT{\mu}{t}$.
    \end{enumerate}
    In the former case, $I$ satisfies the required conditions and we are done.
    Therefore, we may assume that the second condition holds.

    Let $I$ be an minimal interval satisfying
    $\rho |I^*| \ge t$ and $\mu[I] \ge \rho |I^*| \RelativeConcentrationT{\mu}{t}$ (in particular,
    equality will hold). Note that we must have $|I| \ge 2$, since otherwise $I^*$ would be empty,
    contradicting the assumption that $\rho |I^*| \ge t > 0$. 

    Consider any partition
    partition $I = L \cup R$ into two intervals $L, R$ such that $L$ is on the left and $R$ is on
    the right. Then the following inequality holds:
    \[
      \RelativeConcentrationT{\mu}{t} = \frac{\mu[I]}{\rho |I^*|} = \frac{\mu[L] + \mu[R]}{\rho |L^*| + \rho
|I^* \setminus L^*|} \leq \max\left\{ \frac{\mu[L]}{\rho |L^*|}, \frac{\mu[R]}{\rho |I^* \setminus L^*|}
\right\} \leq \max \left\{ \frac{\mu[L]}{\rho |L^*|}, \frac{\mu[R]}{\rho |R^*|} \right\} \,.
    \]

    Consider separately the cases where $|I|$ is even or odd. First suppose $|I|$ is even. Then
    partition $I = L \cup R$ such that $|L| = |R| = |I|/2$ and observe that $|L^*| = |R^*|$,
    so either $\rho |L^*| , \rho |R^*| < t$ or $\rho |L^*| , \rho |R^*| \geq t$. If $\rho |L^*| < t$
    then either $\mu[L] \geq \mu[I]/2 \geq \frac{t}{2} \RelativeConcentrationT{\mu}{t}$ and $L$ satisfies
    the desired conditions (due to the inequality above), or $\mu[R] \geq \mu[I]/2$ and $R$ satisfies
    the desired conditions. On the other hand, if $\mu[L] \geq t$ then either $L$ or $R$ contradicts
    the minimality of $I$.

    Now suppose $|I|$ is odd. Consider two ways to partition $I$ into intervals. Let $I = L_- \cup
R_-$ be the partition such that $|L_-| = \floor{|I|/2}$ and $|R_-| = \ceil{|I|/2}$, and let $I = L_+
\cup R_+$ be the partition such that $|L_+| = \ceil{|I|/2}$ and $|R_+| = \floor{|I|/2}$.
    Note that $|L_-^*| = |R_+^*|$ and $|L_+^*| = |R_-^*| = |I^*|/2$.

    Suppose that $\mu[L_-] \geq \frac{1}{2} \mu[I]$. If $\rho |L_-^*| \leq t$ then $L_-^*$ satisfies the
    required conditions and we are done, so assume $\rho |L_-^*| > t$. Since $|I^*| > 2
|L_-^*|$, we have
    \[
      \mu[L_-] \geq \frac{1}{2} \mu[I] = \frac{1}{2} \rho |I^*| \RelativeConcentrationT{\mu}{t}
      > \rho |L_-^*| \RelativeConcentrationT{\mu}{t} \,,
    \]
    which contradicts the minimality of $I$. A similar argument shows that $R_+$ satisfies the
required conditions when $\mu[R_+] \geq \frac{1}{2} \mu[I]$. So we may assume $\mu[L_-], \mu[R_+] < \frac12
\mu[I]$ in the sequel, which implies $\mu[L_+], \mu[R_-] \geq \frac12 \mu[I]$.

Now if $\rho |L_+^*| \leq t$, then $L_+^*$ satisfies the required conditions and we are done.
Suppose $\rho |L_+^*| > t$ and recall $|L_+^*| = |R_-^*| = |I^*|/2$. Then
$\rho |I^*| = 2 \rho |L_+^*|$, so
\[
  \mu[L_+] \geq \frac12 \mu[I] \geq \frac12 \rho |I^*| \RelativeConcentrationT{\mu}{t}
         = \rho |L_+^*| \RelativeConcentrationT{\mu}{t} \,,
\]
which contradicts the minimality of $I$ and concludes the proof.
\end{proof}

\subsubsection{Variance of the Test Statistic: First Component}

We bound the first component of the variance, $\Varu{\bm{H}}{\Exuc{\bm{T}}{\bm{Y}}{\bm{H}}}$.

\begin{proposition}
\label{prop:negligible-interval-bounds-new}
\label{prop:negligible-interval-bounds}
Let $\cI \in \{\cI^\cyc, \cI^\pth\}$.
For every $i, j, k, \ell \in \bZ_n$, the following hold:
\begin{enumerate}
\item $\Ex{\Phi(\bm J)_{i,j}} \ge \Ex{\bm J [\smallinterval(i,j)]}$;
\item $\Ex{\Phi(\bm{J})_{i,j} \cdot \Phi(\bm{J})_{k,\ell}} \le
        \Ex{\bm J [\smallinterval(i,j)] \cdot \bm J [\smallinterval(k,\ell)]} + 4 \cdot \zeta(\cI)$.
\end{enumerate}
\end{proposition}
\begin{proof}
Recall that $\Phi(\bm J)_{i,j} = 1$ if and only if $\bm\gamma(i) = \bm\gamma(j)$. This will occur
if $\smallinterval(i,j) \subseteq \bm \Gamma_{\gamma(i)}$, which happens when
$\bm{J}[\smallinterval(i,j)]=1$, yielding the first conclusion. Next, observe
\begin{align*}
    \bm\Phi_{i,j}
    = \max\{ \bm{J}[\smallinterval(i,j)], \bm{J}[\largeinterval(i,j)] \}
    \le \bm{J}[\smallinterval(i,j)] + \bm{J}[\largeinterval(i,j)] \,.
\end{align*}
    To prove the second statement, expand the product and use the fact that
    $\bm{J}$ is Boolean-valued:
    \begin{align*}
        \bm\Phi_{i,j} \bm \Phi_{k,\ell}
        &\le \left( \bm{J}[\smallinterval(i,j)] + \bm{J}[\largeinterval(i,j)] \right)
            \left( \bm{J}[\smallinterval(k,\ell)] + \bm{J}[\largeinterval(k,\ell)] \right) \\
        &\le \bm{J}[\smallinterval(i,j)]\bm{J}[\smallinterval(k,\ell)]
            + 2\left( \bm{J}[\largeinterval(i,j)] + \bm{J}[\largeinterval(k,\ell)] \right) \,.
    \end{align*}
    We have $\Ex{\bm{J}[\largeinterval(i,j)]}, \Ex{\bm{J}[\largeinterval(k,\ell)]} \leq \zeta(\cI)$
by definition, so the conclusion follows from taking the expectation.
\end{proof}

\begin{lemma}
\label{prop:general-variance-first-component}
Let $\cI \in \{\cI^\cyc, \cI^\pth\}$.  There exists an absolute constant $c > 0$ such that the first
component of the variance of $\bm{Y}$ satisfies
\[
\Varu{\bm{H}}{\Exuc{\bm{T}}{\bm{Y}}{\bm{H}}}
  \le 5 \zeta(\cI) 
    + c \cdot \sum_{\substack{I=\llangle i,d \rrangle \in \cI \\ 1 \le d \le n}}
        \mu_{i} \mu_{i+d-1} \mu[I]^2 \Ex{\bm{J}(I)}  \,.
\]
\end{lemma}
\begin{proof}
    Fix any subgraph $H$. Conditional on $\bm H = H$,
    \begin{align*}
        \Exuc{\bm{T}}{\bm{Y}}{\bm{H} = H}
        &= \frac{1}{m^2} \Ex{\bm{T}^\top \Phi \bm{T}} - \frac{1}{m^2} \Ex{\|\bm{T}\|_1}
        = \frac{1}{m^2} \sum_{i,j \in \bZ_n} \Ex{\bm{T}_i \bm{T}_j} \Phi_{i,j} - \frac{1}{m^2} \Ex{\Poi(m)} \\
        &= \frac{1}{m^2} \left( \sum_{i \in \bZ_n} \Ex{\bm{T}_i}
            + \sum_{i,j \in \bZ_n} \Ex{\bm{T}_i} \Ex{\bm{T}_j} \Phi_{i,j} \right)
            - \frac{1}{m}
        = \frac{1}{m^2} \left( m + m^2 \mu^\top \Phi \mu \right) - \frac{1}{m} \\
        &= \mu^\top \Phi \mu \,,
    \end{align*}
    and therefore the desired variance is
    \[
        \Varu{\bm{H}}{\Exuc{\bm{T}}{\bm{Y}}{\bm{H}}}
        = \Varu{\bm H}{\mu^\top \bm\Phi \mu} \,.
    \]
    Then, recalling that $\phi = \Ex{\bPhi}$, we expand
    $\Varu{\bm H}{\mu^\top \bm\Phi \mu}$ as follows:
    \begin{align*}
        \Var{\mu^\top \bm\Phi \mu}
        &= \Ex{\left( \mu^\top \bm\Phi \mu \right)^2}
            - \left( \Ex{\mu^\top \bm\Phi \mu} \right)^2
        = \Ex{\left( \mu^\top \bm\Phi \mu \right)^2} - \left( \mu^\top \phi \mu \right)^2 \\
        &= \Ex{\left( \sum_{i,j \in \bZ_n} \mu_i \mu_j \bm\Phi_{i,j} \right)^2}
            - \left( \sum_{i,j \in \bZ_n} \mu_i \mu_j \phi_{i,j} \right)^2 \\
        &= \Ex{\sum_{i,j,k,\ell \in \bZ_n}
                    \mu_i \mu_j \mu_k \mu_\ell \bm\Phi_{i,j} \bm\Phi_{k,\ell}}
            - \sum_{i,j,k,\ell \in \bZ_n} \mu_i \mu_j \mu_k \mu_\ell \phi_{i,j} \phi_{k,\ell} \\
        &= \sum_{i,j,k,\ell \in \bZ_n} \mu_i \mu_j \mu_k \mu_\ell \left(
            \Ex{\bm\Phi_{i,j} \bm\Phi_{k,\ell}} - \phi_{i,j} \phi_{k,\ell} \right) \,.
    \end{align*}
    We now use \cref{prop:negligible-interval-bounds} to simplify the quantity
    $\Ex{\bm\Phi_{i,j} \bm\Phi_{k,\ell}} - \phi_{i,j} \phi_{k,\ell}$:
    \begin{align*}
        &\Ex{\bm\Phi_{i,j} \bm\Phi_{k,\ell}} - \phi_{i,j} \phi_{k,\ell} \\
        &\qquad\le \Ex{\bm{J}[\smallinterval(i,j)] \bm{J}[\smallinterval(k,\ell)]}
            - \Ex{\bm{J}[\smallinterval(i,j)]} \Ex{\bm{J}[\smallinterval(k,\ell)]}
            + 4 \cdot \zeta(\cI) \,.
    \end{align*}
    If the intervals $\smallinterval(i,j)$ and $\smallinterval(k,\ell)$ are disjoint, then
    \[
      \Ex{\bm{J}[\smallinterval(i,j)] \bm{J}[\smallinterval(k,\ell)]}
    - \Ex{\bm{J}[\smallinterval(i,j)]} \Ex{\bm{J}[\smallinterval(k,\ell)]} = 0 \,.
    \]
    On the other hand, if these intervals are not disjoint, we will employ the simple upper bound
    \[
      \Ex{\bm{J}[\smallinterval(i,j)] \bm{J}[\smallinterval(k,\ell)]}
    - \Ex{\bm{J}[\smallinterval(i,j)]} \Ex{\bm{J}[\smallinterval(k,\ell)]} \le
    \Ex{\bm{J}[\smallinterval(i,j)] \bm{J}[\smallinterval(k,\ell)]} \,.
    \]
    We then consider two cases.

First, suppose that for every edge $s \in \bZ_n$, $\smallinterval(i,j)$ crosses $s$ or
$\smallinterval(k,\ell)$ crosses $s$. Then $\bZ_n \subseteq \smallinterval(i,j)^* \cup
\smallinterval(k,\ell)^*$, so $\bm J(\smallinterval(i,j)) \cdot \bm J(\smallinterval(k,\ell)) = 1$
only when every edge appears in $\bm{H}$, which happens with probability at most $\zeta(\cI)$
(since this event implies that every large interval is joined).
In this case, $\Ex{\bm{J}[\smallinterval(i,j)] \bm{J}[\smallinterval(k,\ell)]} \le \zeta(\cI)$.

As for the second case, let $s \in \bZ_n$ be such that neither $\smallinterval(i,j)$ nor
$\smallinterval(k,\ell)$ crosses $s$. Since $\smallinterval(i,j)$ and $\smallinterval(k,\ell)$ are not
disjoint,  it follows that there exists an interval $I = I_{i,j,k,\ell} \in \cI$ satisfying the
following:
    \begin{enumerate}
        \item The set $\smallinterval(i,j) \cup \smallinterval(k,\ell)$ is equal to
            the set of elements of $I_{i,j,k,\ell}$, where we are here taking the union as sets (not as
multisets); 
        \item The endpoints of $I_{i,j,k,\ell}$ are two of the indices $i,j,k,\ell$; and,
        \item $|I_{i,j,k,\ell}| \le n$ (because, in particular, $I_{i,j,k,\ell}$ does not cross $s$).
    \end{enumerate}
    It follows that $\bm{J}[\smallinterval(i,j)] \bm{J}[\smallinterval(k,\ell)] = 1$ if and only if
    $\bm{J}[I] = 1$, and hence we have the upper bound
    $\Ex{\bm{J}[\smallinterval(i,j)] \bm{J}[\smallinterval(k,\ell)]}
    - \Ex{\bm{J}[\smallinterval(i,j)]} \Ex{\bm{J}[\smallinterval(k,\ell)]} \le \Ex{\bm{J}[I]}$.
    Therefore, 
\begin{align*}
  &\sum_{i,j,k,\ell \in \bZ_n} \mu_i \mu_j \mu_k \mu_\ell \left(
    \Ex{\bm\Phi_{i,j} \bm\Phi_{k,\ell}} - \phi_{i,j} \phi_{k,\ell} \right) \\
  &\leq \sum_{i,j,k,\ell \in \bZ_n} \mu_i \mu_j \mu_k \mu_\ell \left(
    \Ex{\bm{J}[\smallinterval(i,j)] \bm{J}[\smallinterval(k,\ell)]}
      - \Ex{\bm{J}[\smallinterval(i,j)]} \Ex{\bm{J}[\smallinterval(k,\ell)]}
      + 4\cdot\zeta \right) \\
  &\leq 4 \cdot \zeta
    + \sum_{i,j,k,\ell \in \bZ_n} \mu_i\mu_j\mu_k\mu_\ell \left(
        \ind{ \bZ_n \subseteq \smallinterval(i,j)^* \cup \smallinterval(k,\ell)^* } \cdot \zeta(\cI)
      + \ind{I_{i,j,k,\ell} \text{ exists}} \cdot \Ex{ \bm{J}[ I_{i,j,k,\ell} ] } \right) \\
  &\leq 5 \zeta(\cI) \| \mu \|_1^4
    + \sum_{i,j,k,\ell \in \bZ_n} \mu_i\mu_j\mu_k\mu_\ell 
      \ind{I_{i,j,k,\ell} \text{ exists}} \cdot \Ex{ \bm{J}[ I_{i,j,k,\ell} ] } \,.
\end{align*}
The latter is bounded by summing over all intervals $I = \llangle i, d \rrangle \in \cI$
with $d \leq n$ and for each one taking the expression $c \cdot \mu_i \mu_{i+d-1} \sum_{j,k \in I}
\mu_j \mu_k \cdot \Ex{\bm{J}[I]}$, where $i$ and $i+d-1$ are the endpoints of $I$, and $c$ is a constant
counting the number of ways to get intersecting intervals with endpoints in $i,(i+d-1),k,\ell$.
Now, using $\sum_{j,k \in I} \mu_j \mu_k = \mu[I]^2$, we obtain
\[
    \Var{\mu^\top \Phi(\bm{J}) \mu}
    \leq 5 \zeta(\cI) 
    + c \cdot \sum_{\substack{I = \llangle i, d \rrangle \in \cI \\ 1 \le d \le n}}
      \Ex{\bm{J}[I]} \mu_i \mu_{i+d-1} \mu[I]^2 \,. \qedhere
\]
\end{proof}

\begin{lemma}[First component of the variance]
    \label{lemma:var-first-component-confused-collector}
    Let $\delta \in (0, 1)$ be a constant, let $n \in \bN$ be sufficiently large and let $\mu$ be a
    probability distribution over $\bZ_n$. Suppose that $\rho \ge \Omega(n^{-\delta})$ and $0 < t <
    1$. Then
    \[
        \Varu{\bm{H}}{\Exuc{\bm{T}}{\bm{Y}}{\bm{H}}}
        \le O\left( \frac{\RelativeConcentrationT{\mu}{t}^2}{\rho} \right) \cdot \|\mu\|_2^2 \,.
    \]
\end{lemma}
\begin{proof}
    By \cref{prop:general-variance-first-component}, for some constant $c > 0$ we have
    \[
        \Varu{\bm{H}}{\Exuc{\bm{T}}{\bm{Y}}{\bm{H}}}
            \le 5 \zeta(\cI)
            + c \cdot \sum_{\substack{I=\llangle i,d \rrangle \in \cI \\ 1 \le d \le n}}
            \mu_{i} \mu_{i+d-1} \mu[I]^2 \Ex{\bm{J}(I)}  \,.
    \]
    We start with the second component of the RHS.
    Recall that for any $I = \llangle i,d \rrangle$, $\Ex{\bm{J}(I)} \le \eta^{d-1}$ (this value
    may be zero if $\cI = \cI^\pth$ and $I$ crosses the edge between vertices $0$ and $n-1$).
    We have
    \begin{align*}
        &\sum_{\substack{I=\llangle i,d \rrangle \in \cI \\ 1 \le d \le n}} \mu_{i} \mu_{i+d-1} \mu[I]^2 \Ex{\bm{J}(I)}
        \le \sum_{i=0}^{n-1} \sum_{d=1}^{n} \mu_i \mu_{i+d-1} \mu[\llangle i,d \rrangle]^2 \eta^{d-1} \\
        &\quad \le \sum_{i=0}^{n-1} \sum_{d=1}^{n} \mu_i \mu_{i+d-1} \eta^{d-1}
            \left[ \RelativeConcentrationT{\mu}{t} \cdot \max\left\{ \rho(d-1), t \right\} \right]^2
            & \text{(Relative concentration)} \\
        &\quad \le \sum_{i=0}^{n-1} \sum_{d=1}^n \mu_i \mu_{i+d-1} \eta^{d-1}
            \left[ \RelativeConcentrationT{\mu}{t}^2
                \cdot \left( \rho^2 d^2 + 1 \right) \right]
            & \text{(Since $t < 1$)} \\
        &\quad = \RelativeConcentrationT{\mu}{t}^2 \left(
            \rho^2 \sum_{d=1}^n d^2 \eta^{d-1} \sum_{i=0}^{n-1} \mu_i \mu_{i+d-1} +
            \sum_{d=1}^n \eta^{d-1} \sum_{i=0}^{n-1} \mu_i \mu_{i+d-1} \right) \\
        &\quad \le \RelativeConcentrationT{\mu}{t}^2 \|\mu\|_2^2 \left(
            \rho^2 \sum_{d=1}^n d^2 \eta^{d-1} + \sum_{d=1}^n \eta^{d-1} \right)
            & \text{(By Cauchy-Schwarz)} \\
        &\quad \le \RelativeConcentrationT{\mu}{t}^2 \|\mu\|_2^2 \left(
            \rho^2 \cdot \frac{1+\eta}{\rho^3} + \frac{1}{\rho} \right)
            & \text{(Since $\rho = 1-\eta$)} \\
        &\quad \le 3 \cdot \frac{\RelativeConcentrationT{\mu}{t}^2}{\rho} \|\mu\|_2^2
            & \text{(Since $\eta < 1$)} \,,
    \end{align*}
    as desired. Now, we show that the term $5\zeta(\cI)$ satisfies the claimed bound.
    First, \cref{prop:zeta-is-small-confused-collector} implies that $\zeta(\cI) \le 1/n^3$.
    Note that, by taking intervals $I = \llangle i, 1 \rrangle$ in
    \cref{def:relative-concentration}, we have that
    $\RelativeConcentrationT{\mu}{t} \ge \mu(i)/t$ for every $i$, and hence
    $\RelativeConcentrationT{\mu}{t} \ge 1/nt$. Along with the inequalities $\|\mu\|_2^2 \ge 1/n$,
    $\rho \le 1$ and $t < 1$, we obtain
    \[
        \zeta(\cI) \le \frac{1}{n^3}
        \le \frac{1}{n^3 t^2}
        \le \frac{\RelativeConcentrationT{\mu}{t}^2}{n}
        \le \frac{\RelativeConcentrationT{\mu}{t}^2}{\rho} \|\mu\|_2^2 \,. \qedhere
    \]
\end{proof}

\subsubsection{Variance of the Test Statistic: Second Component}
\label{section:shared-analysis-variance-2}

We now bound the second component of the variance, $\Exu{\bm{H}}{\Varuc{\bm{T}}{\bm{Y}}{\bm{H}}}$.
The first step is to relate it to the 2- and 3-norms of the clustered
vector $\mu$. Recall that $\bm T_i \sim \Poi(m\mu_i)$ is the number of occurrences of vertex $i \in
\bZ_n$ in the sample. We first compute the variance of the terms $\bm{X}_i (\bm{X}_i - 1)$
that make up the test statistic:

\begin{proposition}
    \label{prop:poisson-expression-variance}
    If $\bm X \sim \Poi(\lambda)$, then $\Var{\bm X(\bm X-1)} = 4\lambda^3 + 2\lambda^2$.
\end{proposition}
\begin{proof}
    The Poisson random variable $\bm{X}$ has the following raw moments (see \eg \cite{Rio37}):
    \begin{align*}
        \Ex{X}   &= \lambda \,, \\
        \Ex{X^2} &= \lambda + \lambda^2 \,, \\
        \Ex{X^3} &= \lambda + 3\lambda^2 + \lambda^3 \,, \\
        \Ex{X^4} &= \lambda + 7\lambda^2 + 6\lambda^3 + \lambda^4 \,.
    \end{align*}
    Therefore we have
    \begin{align*}
        \Var{X(X-1)}
        &= \Ex{(X(X-1))^2} - \Ex{X(X-1)}^2
        = \Ex{(X^2 - X)^2} - (\Ex{X^2} - \Ex{X})^2 \\
        &= \Ex{X^4} - 2\Ex{X^3} + \Ex{X^2} - \Ex{X^2}^2 + 2\Ex{X^2}\Ex{X} - \Ex{X}^2 \\
        &= (\lambda + 7\lambda^2 + 6\lambda^3 + \lambda^4)
            - 2(\lambda + 3\lambda^2 + \lambda^3)
            + (\lambda + \lambda^2)
            - (\lambda + \lambda^2)^2
            + 2(\lambda + \lambda^2)\lambda
            - \lambda^2 \\
        &= 4\lambda^3 + 2\lambda^2 \,. \qedhere
    \end{align*}
\end{proof}

Recall the notation $\induced{\mu}{\Gamma}$ for the distribution on the cluster representatives
induced by $\mu$ (\cref{def:induced-distributions}).
\begin{lemma}
\label{prop:general-variance-second-component}
Let $H$ be a subgraph with induced cells $\Gamma = (\Gamma_1, \dotsc, \Gamma_b)$, and let $\mu$ be a
distribution on $\bZ_n$ Then the conditional variance of $\bm{Y}$ given $\bm{H} = H$ satisfies
    \begin{align*}
        \Varuc{\bm{T}}{\bm{Y}}{\bm{H} = H}
            = \frac{2}{m^2} \|\induced{\mu}{\Gamma}\|_2^2
                + \frac{4}{m} \|\induced{\mu}{\Gamma}\|_3^3 \,.
    \end{align*}
\end{lemma}
\begin{proof}
Using \cref{prop:poisson-expression-variance}, the desired variance is
\begin{align*}
      \Varuc{\bm{T}}{\bm{Y}}{\bm{H} = H}
      &= \Var{\frac{1}{m^2} \sum_{i=1}^b \bm{X}_i (\bm{X}_i - 1)}
      = \frac{1}{m^4} \sum_{i=1}^b \left[
          4\left(m \mu\left[\Gamma_i\right]\right)^3
          + 2\left(m \mu\left[\Gamma_i\right]\right)^2
          \right] \\
      &= \frac{4}{m} \|\induced{\mu}{\Gamma}\|_3^3
            + \frac{2}{m^2} \|\induced{\mu}{\Gamma}\|_2^2 \,. \qedhere
\end{align*}
\end{proof}

\noindent
We will need the following auxiliary result.

\begin{proposition}[Cells are almost always small]
    \label{prop:bucket-sizes-small} Let $K \ge 2$ be a constant, and
    suppose $n \ge 3$. Then the cells $\bm{\Gamma} = (\bm{\Gamma}_i, \dotsc, \bm{\Gamma}_{\bm{b}})$
    induced by $\bm{H}$ satisfy
    \[
        \abs*{\bm{\Gamma}_i} \le \frac{2K\log n}{\rho} \qquad \forall i \in [\bm{b}]
    \]
    except with probability at most $1/n^K$.
\end{proposition}
\begin{proof}
    Let $d \define \left\lceil \frac{2K \log n}{\rho} \right\rceil$, and fix any interval
    $I = \llangle i, d \rrangle$. The probability that $I$ is joined is
    \[
        \Pr{\bm{J}(I)=1} = (1-\rho)^{d-1}
        \le (1-\rho)^{\frac{2K\log n}{\rho} - 1}
        \le (1-\rho)^{\frac{(2K-1)\log n}{\rho}}
        \le e^{-(2K-1)\log n}
        = 1/n^{2K-1}
        \le 1/n^{K+1} \,,
    \]
    where we used the facts that $n \ge 3 \implies \frac{\log n}{\rho} \ge 1$ and that
    $K \ge 2 \implies 2K-1 \ge K+1$. Now, if any cell has size at least $d$, then some interval
    $I = \llangle i,d \rrangle$ satisfies $J(I)=1$. Since there are at most
    $n$ such intervals, the probability of this event is at most $1/n^K$ by the union bound.
\end{proof}

\begin{lemma}[Second component of the variance]
    \label{lemma:var-second-component-confused-collector}
    Let $n \in \bN$ be sufficiently large and let $\mu$ be a probability distribution over $\bZ_n$.
    Suppose $m \le \poly(n)$ and $0 < t < 1$. Then
    \[
        \Exu{\bm{H}}{\Varuc{\bm{T}}{\bm{Y}}{\bm{H}}}
        \le O\left(
                \frac{\log n}{m^2 \rho} \cdot \left(
                    1 + m \RelativeConcentrationT{\mu}{t} \log n \right)
            \right) \cdot \|\mu\|_2^2 \,.
    \]
\end{lemma}
\begin{proof}
    Let $K \ge 2$ be some constant such that $m \le O(n^{K-1})$.
    By \cref{prop:bucket-sizes-small}, the cells
    $\bm{\Gamma} = (\bm{\Gamma}_1, \dotsc, \bm{\Gamma}_n)$ induced by $\bm{H}$ are such that
    $\abs*{\bm{\Gamma}_i} \le \frac{2K\log n}{\rho}$ for all $i$,
    except with probability at most $1/n^K$. We show that the variance is small when this condition
    holds, and that the low-probability case where the condition fails does not contribute too much
    to the expectation.

    \textbf{Case 1.} Suppose $H$ is such that its induced cells
    $\Gamma = (\Gamma_1, \dotsc, \Gamma_b)$ satisfy
    $\abs*{\Gamma_i} \le \frac{2K\log n}{\rho}$
    for every $i \in [b]$. We wish to show that $\Varuc{\bm{T}}{\bm{Y}}{\bm{H} = H}$
    satisfies the upper bound from the statement.
    We start with the result from
    \cref{prop:general-variance-second-component}:
    \[
        \Varuc{\bm{T}}{\bm{Y}}{\bm{H} = H} = \frac{2}{m^2} \|\induced{\mu}{\Gamma}\|_2^2
                + \frac{4}{m} \|\induced{\mu}{\Gamma}\|_3^3 \,.
    \]
    We start by bounding the first term in the RHS. For each cell $\Gamma_i$, we have
    \[
        \left(\induced{\mu}{\Gamma}\right)_i^2
        = \left( \sum_{j \in \Gamma_i} \mu_j \right)^2
        = \sum_{j,k \in \Gamma_i} \mu_j \mu_k
        \le \sum_{j,k \in \Gamma_i} \frac{\mu_j^2 + \mu_k^2}{2}
        = \abs*{\Gamma_i} \sum_{j \in \Gamma_i} \mu_j^2
        \le O\left(\frac{\log n}{\rho}\right) \sum_{j \in \Gamma_i} \mu_j^2 \,.
    \]
    Hence, we obtain
    \[
        \frac{1}{m^2} \|\induced{\mu}{\Gamma}\|_2^2
        = \frac{1}{m^2} \sum_{i=1}^b \induced{\mu}{\Gamma}(i)^2
        \le O\left(\frac{\log n}{m^2 \rho}\right) \|\mu\|_2^2 \,,
    \]
    as desired. Moving on to the second term, first note that
    $\|\induced{\mu}{\Gamma}\|_3^3 = \sum_i \mu[\Gamma_i]^3
    \le \sum_i (\max_j \mu[\Gamma_j]) \mu[\Gamma_i]^2
    = \|\induced{\mu}{\Gamma}\|_\infty \|\induced{\mu}{\Gamma}\|_2^2$. We claim that
    $\|\induced{\mu}{\Gamma}\|_\infty \le O\left( \RelativeConcentrationT{\mu}{t} \log n \right)$.
    Fix any $i \in [b]$ and consider entry
    $\left(\induced{\mu}{\Gamma}\right)_i = \mu\left[\Gamma_i\right]$. The anticoncentration of
    $\mu$ yields
    \[
        \mu\left[\Gamma_i\right] \le \RelativeConcentrationT{\mu}{t} \cdot \max\left\{
            \rho\left( \abs*{\Gamma_i}-1 \right), t \right\} \,.
    \]
    Combining with the assumption that $\abs*{\Gamma_i} \le O\left(\frac{\log n}{\rho}\right)$
    and recalling that $t < 1$, we obtain that
    $\mu\left[\Gamma_i\right] \le O\left( \RelativeConcentrationT{\mu}{t} \log n \right)$, as
    claimed. We have already shown that
    $\|\induced{\mu}{\Gamma}\|_2^2 \le \frac{\|\mu\|_2^2}{\rho} \cdot O(\log n)$, and thus
    \[
        \frac{1}{m} \|\induced{\mu}{\Gamma}\|_3^3
        \le \frac{1}{m} \cdot O\left( \RelativeConcentrationT{\mu}{t} \log n \right)
            \cdot \frac{\|\mu\|_2^2}{\rho} \cdot O(\log n)
        = O\left( \frac{\log n}{m^2 \rho} \cdot
            m \RelativeConcentrationT{\mu}{t} \log n \right) \cdot \|\mu\|_2^2 \,,
    \]
    which concludes Case 1.

    \textbf{Case 2.} In the rare event that $\bm{H} = H$ is a subgraph that
    fails the small-cells condition of Case 1, we will fall back to a looser upper bound for
    the conditional variance that holds for every $H$. We once again start with the result
    from \cref{prop:general-variance-second-component}:
    \[
        \Varuc{\bm{T}}{\bm{Y}}{\bm{H} = H}
        = \frac{2}{m^2} \|\induced{\mu}{\Gamma}\|_2^2
            + \frac{4}{m} \|\induced{\mu}{\Gamma}\|_3^3 \,.
    \]
    Using $\|\induced{\mu}{\Gamma}\|_1 = 1$ along with the monotonicity of $\ell_p$ norms gives
    \[
        \Varuc{\bm{T}}{\bm{Y}}{\bm{H} = H}
        \le \frac{2}{m^2} \|\induced{\mu}{\Gamma}\|_1^2
                + \frac{4}{m} \|\induced{\mu}{\Gamma}\|_1^3
        \le \frac{6}{m^2} \cdot m
        = \frac{1}{m^2} \cdot O(n^{K-1}) \,.
    \]

    \textbf{Concluding the argument.} We now combine both cases to upper bound the expected
    variance. Using \cref{prop:bucket-sizes-small}, we have
    \begin{align*}
        \Exu{\bm{H}}{\Varuc{\bm{T}}{\bm{Y}}{\bm{H}}}
        &\le \left[ \begin{array}{l}
            \Pr{\bm{H} \text{ satisfies } \|\bm{\Gamma}\|_\infty \le \frac{2K\log n}{\rho}}
                \Exuc{\bm{H}}
                     {\Varuc{\bm{T}}{\bm{Y}}{\bm{H}}}
                     {\bm{H} \text{ satisfies } \|\bm{\Gamma}\|_\infty \le \frac{2K\log n}{\rho}} \\
            + \Pr{\bm{H} \text{ does not satisfy } \|\bm{\Gamma}\|_\infty \le \frac{2K\log n}{\rho}}
                \max_{H} \left\{ \Varuc{\bm{T}}{\bm{Y}}{\bm{H} = H} \right\}
        \end{array} \right] \\
        &\le O\left(
                    \frac{\log n}{m^2 \rho} \cdot \left(
                        1 + m \RelativeConcentrationT{\mu}{t} \log n \right)
                \right) \cdot \|\mu\|_2^2
            + \frac{1}{n^K} \cdot \frac{1}{m^2} \cdot O(n^{K-1}) \,.
    \end{align*}
    Since $\|\mu\|_2^2 \ge \|\nu\|_2^2 = 1/n$, the first term dominates the second, concluding the
    proof.
\end{proof}

We obtain the general form of our concentration bound; the result below gives
\cref{lemma:intro-informal-concentration} with $\RelativeConcentration{\mu} =
\RelativeConcentrationT{\mu}{t}$ for any $0 < t < 1$. While we use a specialized version in the
analysis of our tester (\cref{lemma:concentration-specialized}) that is better by logarithmic
factors, we prove the following for the sake of generality.

\begin{lemma}[Formal version of \cref{lemma:intro-informal-concentration}]
    \label{lemma:concentration-inequality-t}
    Let $\delta \in (0,1)$ be a constant, let $n \in \bN$ be sufficiently large and let $\mu$ be a
    probability distribution over $\bZ_n$. Suppose $\rho \ge \Omega(n^{-\delta})$, $m \le \poly(n)$
    and $0 < t < 1$. Then for all $\tau > 0$,
    \[
        \Pr{\abs*{\bm{Y} - \Ex{\bm{Y}}} \ge \tau}
        \le \frac{\|\mu\|_2^2}{\rho \tau^2}
                \cdot \max\left\{ \RelativeConcentrationT{\mu}{t}, \frac{1}{m} \right\}^2
                \cdot O(\log^2 n) \,.
    \]
\end{lemma}
\begin{proof}
    By
    \cref{lemma:var-first-component-confused-collector,lemma:var-second-component-confused-collector},
    along with the law of total variance, we obtain
    \begin{align*}
        \Var{\bm{Y}}
        &\le O\left( \frac{\RelativeConcentrationT{\mu}{t}^2}{\rho} \right) \cdot \|\mu\|_2^2
            + O\left(
                \frac{\log n}{m^2 \rho} \cdot \left(
                    1 + m \RelativeConcentrationT{\mu}{t} \log n \right)
            \right) \cdot \|\mu\|_2^2 \\
        &\le \frac{\|\mu\|_2^2}{\rho} \cdot O\left(
                \RelativeConcentrationT{\mu}{t}^2
                    + 2 \cdot \RelativeConcentrationT{\mu}{t} \cdot \frac{1}{m}
                    + \frac{1}{m^2}
                \right) \cdot O(\log^2 n) \\
        &= \frac{\|\mu\|_2^2}{\rho}
                \cdot \max\left\{ \RelativeConcentrationT{\mu}{t}, \frac{1}{m} \right\}^2
                \cdot O(\log^2 n) \,.
    \end{align*}
    The claim follows by Chebyshev's inequality.
\end{proof}

\subsection{Testing Uniformity Without Queries}
\label{section:random-zero-query}
\label{sec:upper-bound-confused-collector}

We now give the upper bound for testing uniformity on either the path or the cycle with zero
queries.  The tester is \cref{alg:tester-confused-collector}, and consists of two steps:
\begin{enumerate}
    \item Concentration test: checks whether any count in the sample is too large;
        this case corresponds to highly concentrated distributions, which can be rejected.
    \item Collision-based test: accept or reject depending on whether the test statistic $\bm Y$ is
        below a certain threshold.
\end{enumerate}

\begin{algorithm}[H]
    \caption{Uniformity tester in the confused collector model.}
    \label{alg:tester-confused-collector}
      
    \hspace*{\algorithmicindent}
        Set $m \gets c \cdot \frac{\sqrt{n}}{\epsilon^2} \cdot \frac{\log^2 n}{\rho^{3/2}}$. \\
    \hspace*{\algorithmicindent} \textbf{Constants:}
        $\alpha, \beta, L, c > 0$ to be defined later. \\
    \hspace*{\algorithmicindent} \textbf{Requires:}
        $\rho \ge \frac{L \log^{4/5} n}{n^{1/5} \epsilon^{4/5}}$. \\
    \begin{algorithmic}[1]
        \Procedure{UniformityTester-ConfusedCollector}{$\mu, n, \epsilon, \rho$}
            \State Let $X = (X_1, \dotsc, X_n)$ be the variables defined in
                                            \cref{subsection:notation-clustered-samples}
                for a sample of size $\Poi(m)$ from $\mu$.
            \State \textbf{If} $\max_i X_i \ge \alpha \log n$ then \textbf{reject}.
            \State $Y \gets \frac{1}{m^2} \sum_i X_i(X_i-1)$.
            \State \textbf{If}
                $Y \ge \frac{1}{n^2}\sum_{i,j} \phi_{i,j} + \beta \frac{1}{n} \epsilon^2 \rho$
                then \textbf{reject}.
            \State \textbf{Accept}.
        \EndProcedure
    \end{algorithmic}
\end{algorithm}

\paragraph{Remark on the optimality of the collision-based tester.} Considering that we give a
Poissonized tester whose main statistic $\bm{Y}$ is equivalent to the collision-based statistic of
\cite{GR00} when $\rho=1$, it may seem surprising that we obtain a sample complexity of $\widetilde
O\left( \sqrt{n}/\epsilon^2 \right)$ -- as opposed to $O\left(\sqrt{n}/\epsilon^4\right)$ -- when it
is known that, for an analysis based on bounding the variance of $\bm{Y}$ and applying Chebyshev's
inequality, establishing the optimal sample complexity is only possible with a different test
statistic (\eg the modified chi-squared statistic~\cite{CDVV14,DKN15b,VV17}) or a careful analysis
of the non-Poissonized tester~\cite{DGPP19} (see also the Remark in Section~2 therein). Our analysis
implicitly avoids this issue via the relative concentration test, which upper bounds
$\|\mu\|_\infty$, and indeed specializing our proof to the case $\rho=1$ would yield a sample
complexity dependence on $1/\epsilon^4$ rather than $\sqrt{n}/\epsilon^4$; and since our analysis
only handles the case $\epsilon \ge \widetilde \Omega\left( n^{-1/4} \right)$ (see
\cref{remark:confused-collector-parameter-ranges}), the term $\sqrt{n}/\epsilon^2$ dominates
$1/\epsilon^4$.

\subsubsection{Easy Case: Highly Concentrated Distributions}
\label{subsec:confused-collector-highly-concentrated}

\begin{definition}[Highly concentrated distributions]
Given a constant $C > 0$, positive integer $m$, resolution parameter $\rho$, and probability
distribution $\mu$ over $\bZ_n$, we say that $\mu$ is \emph{$C$-highly concentrated (under resolution
$\rho$ with respect to $m$)} if $\RelativeConcentrationT{\mu}{t} \ge \frac{C \log^2 n}{m}$, where
$t = \frac{1}{\log n}$.
\end{definition}

To understand this definition (in particular, the choice of $t$), observe that
\cref{alg:tester-confused-collector} tries to reject highly concentrated distributions by checking
whether any entry $\bm{X}_i$ in its clustered sample is large, namely $\bm{X}_i \ge \alpha \log n$
(which is unlikely to occur under the uniform distribution). When $\mu$ is highly concentrated,
\cref{lemma:relative-concentration-structural} ensures an interval $I$ exists, satisfying $\rho
|I^*| \le 1 / \log n$ (so $I$ is likely to be joined in the clustering) and $\mu[I] \ge
\frac{C \log n}{2m}$; so, for sufficiently large $C$, the $\Poi(m \mu[I])$ will be larger than
$\alpha \log n$, causing the algorithm to reject.

Furthermore, when $\mu$ is not highly concentrated, the condition
$\RelativeConcentrationT{\mu}{t} < \frac{C \log^2 n}{m}$ means that the term $\max\left\{
    \RelativeConcentrationT{\mu}{t}, \frac{1}{m} \right\}$ in the general concentration inequality
(\cref{lemma:concentration-inequality-t}) cannot be improved by more than logarithmic factors.
Therefore this choice of definition is both sensible for the first step of the tester, and nearly
optimal for our analysis of the second step of the tester.


\begin{observation}
    Let $t \define \frac{1}{\log n}$. Then by inspecting the definition of relative concentration
    and breaking down into the cases where $|I| = 1$ or $|I| \ge 2$, we have
    $\RelativeConcentrationT{\mu}{t}
    \le \max\left\{ \frac{2}{\rho}, \log n \right\} \cdot \|\mu\|_\infty$.
\end{observation}

\begin{remark}
    \label{remark:non-concentrated-l-infty-confused-collector}
    If $\mu$ is not $C$-highly concentrated, then in particular $\|\mu\|_\infty < \frac{C\log n}{m}$,
    as can be seen by taking intervals $I = \llangle i,1 \rrangle$ for each $i \in \bZ_n$.
\end{remark}

We now show that the first step of the tester correctly accepts the uniform distribution and
rejects highly concentrated distributions with good probability. Therefore, we will be able
to assume that $\mu$ is not highly concentrated when analyzing the second step of the tester.

We will need the following tail bounds for the Poisson distribution, as stated in \cite{Can17}.

\begin{fact}
    \label{fact:poisson-concentration}
    Let $\bm{X} \sim \Poi(\lambda)$ for some $\lambda > 0$. Then for any $t > 0$,
    \[
        \Pr{\bm{X} \le \lambda - t}, \Pr{\bm{X} \ge \lambda + t}
        \le e^{-\frac{t^2}{2(\lambda+t)}} \,.
    \]
\end{fact}

The result below uses the hypothesis $m \le n \rho$, which intuitively corresponds to the sublinear
sample complexity regime in the standard uniformity testing model. This hypothesis holds in the
range of parameters we consider, but not necessarily in more extreme regimes (see
\cref{remark:confused-collector-parameter-ranges}).

\begin{lemma}
    \label{lemma:concentration-test-eta}
    For sufficiently large constant $\alpha > 0$ and all sufficiently large $n$, the following holds.
    Suppose $m \le n \rho$. Then for any distribution $\mu$ over $\bZ_n$, we have:
    \begin{enumerate}
        \item If $\mu$ is uniform, the first step of the tester only rejects with probability
            at most $1/100$; and
        \item If $\mu$ is $4\alpha$-highly concentrated, the first step of the tester rejects
            with probability at least $99/100$.
    \end{enumerate}
\end{lemma}
\begin{proof}
    \textbf{Completeness.} Suppose $\mu$ is the uniform distribution over $\bZ_n$.
    From \cref{prop:bucket-sizes-small}, we obtain that every cell has size at most
    $\frac{4\log n}{\rho}$ except with probability $o(1)$. Assume that this is the case, and fix some
    particular cell $\Gamma_i$. The number of elements sampled from this cell is distributed as
    $\bm{X}_i \sim \Poi\left(m \mu\left[\Gamma_i\right]\right) = \Poi\left( m|\Gamma_i|/n \right)$.
    Then, using \cref{fact:poisson-concentration} and for $\alpha > 16$, the probability that
    $\bm{X}_i$ is so large that the tester rejects is
    \begin{align*}
        \Pr{\bm{X}_i \ge \alpha \log n}
        &\le \Pr{\Poi\left( m \cdot \frac{4\log n}{\rho} \cdot \frac{1}{n} \right) \ge \alpha \log n} \\
        &\le \Pr{\Poi (4\log n) \ge \alpha \log n} \qquad & \text{(Since $m \le n \rho$)} \\
        &= \Pr{\Poi(4\log n) - 4\log n \ge (\alpha-4)\log n} \\
        &\le e^{-\frac{(\alpha-4)^2 \log^2 n}{2(4\log n + (\alpha-4)\log n)}}
        \le e^{-\frac{(\alpha/2)^2 \log n}{2\alpha}}
        \le e^{-2\log n}
        = 1/n^2 \,.
    \end{align*}
    Hence, the probability that this happens for any $\bm{X}_i$ is at most $1/n = o(1)$.

    \textbf{Soundness.} Suppose $\mu$ is $4\alpha$-highly concentrated.
    Using \cref{lemma:relative-concentration-structural} and the definition of high concentration,
    we may choose some interval $I=\llangle i,d \rrangle \in \cI$ satisfying
    \begin{enumerate}
        \item $(d-1)\rho \le \frac{1}{\log n}$,
            and thus $|I| = d \le 1 + \frac{1}{\rho \log n}$; and
        \item $\mu[I] \ge \frac{1}{2} \cdot \frac{1}{\log n} \cdot \frac{4\alpha \log^2 n}{m}
            = \frac{2\alpha \log n}{m}$.
    \end{enumerate}
    We first claim that all the elements in $I$ will be joined with high probability, \ie every
    edge in $I^*$ will be sampled into $\bm{H}$. Indeed, by the union bound, we have
    \[
        \Pr{J(I) = 0}
        = \Pr{\exists e \in I^* : e \not\in \bm{H}} \le
        (|I|-1) \cdot \rho \le \frac{1}{\rho \log n} \cdot \rho = \frac{1}{\log n} = o(1) \,.
    \]
    Now, suppose every element in $I$ belongs to the same cell, say $\Gamma_i$. Recall that the
    random variable $\bm{X}_i \sim \Poi\left(m \mu\left[\Gamma_i\right]\right)$ counts the number
    of elements drawn from this cell, and by our assumption on $I$, we have
    \[
        m \mu\left[\Gamma_i\right]
        \ge m \cdot \frac{2\alpha \log n}{m}
        = 2\alpha \log n \,.
    \]
    We now claim that, with high probability, $\bm{X}_i > \alpha \log n$ and hence the tester will
    reject. Indeed, using \cref{fact:poisson-concentration}, the probability that this does not
    occur is
    \begin{align*}
        \Pr{\bm{X}_i \le \alpha \log n}
        &\le \Pr{\Poi\left( 2\alpha \log n \right) \le \alpha \log n}
        = \Pr{\Poi\left(2\alpha\log n\right) \le 2\alpha\log n - \alpha \log n} \\
        &\le e^{-\frac{\alpha^2 \log^2 n}
                       {2(2\alpha\log n + \alpha\log n)}}
        = e^{-\frac{\alpha \log n}{6}}
        = o(1) \,. \qedhere
    \end{align*}
\end{proof}

\subsubsection{Correctness of the Tester}

Combining our separation and concentration results, we can show that $\bm{Y}$ is concentrated on the
correct side of the tester's threshold. First, we specialize our general concentration result to the
case of non-highly-concentrated distributions.

\begin{lemma}
    \label{lemma:concentration-specialized}
    Let $C > 0$ and $\delta \in (0, 1)$ be constants, let $n \in \bN$ be sufficiently large and
    suppose $\mu$ is a probability distribution over $\bZ_n$ that is not $C$-highly concentrated
    with respect to $m$. Suppose $\rho \ge \Omega(n^{-\delta})$ and $m \le \poly(n)$. Then for all
    $\tau > 0$,
    \[
        \Pr{\abs*{\bm{Y} - \Ex{\bm{Y}}} \ge \tau}
        \le \frac{\|\mu\|_2^2}{\rho m^2 \tau^2} \cdot O\left( \log^4 n \right) \,.
    \]
\end{lemma}
\begin{proof}
    We again combine
    \cref{lemma:var-first-component-confused-collector,lemma:var-second-component-confused-collector},
    which apply since here $0 < t = \frac{1}{\log n} < 1$, via the law of total variance, and
    conclude with Chebyshev's inequality (directly applying the general
    \cref{lemma:concentration-inequality-t} would also work, but lead to a small loss in logarithmic
    factors). We have
    \begin{align*}
        \Var{\bm{Y}}
        &\le O\left( \frac{\RelativeConcentrationT{\mu}{t}^2}{\rho} \right) \cdot \|\mu\|_2^2
            + O\left(
                \frac{\log n}{m^2 \rho} \cdot \left(
                    1 + m \RelativeConcentrationT{\mu}{t} \log n \right)
            \right) \cdot \|\mu\|_2^2 \\
        &\le O\left( \frac{\log^4 n}{\rho m^2} \right) \cdot \|\mu\|_2^2
            + O\left(
                \frac{\log n}{m^2 \rho} \cdot \left(
                    1 + m \cdot \frac{\log^2 n}{m} \cdot \log n \right)
            \right) \cdot \|\mu\|_2^2
        = \frac{\|\mu\|_2^2}{\rho m^2} \cdot O\left( \log^4 n \right) \,. \qedhere
    \end{align*}
\end{proof}

\begin{lemma}
    \label{lemma:tester-second-step-confused-collector}
    Let $\alpha, L > 0$ be constants. Then there exist constants $\beta > 0$, and
    $c = c_{\alpha,\beta} > 0$ such that 
    the following holds for all sufficiently large $n$.
    Let $\epsilon, \rho \in (0, 1]$ satisfy $\rho \ge \frac{L \log^{4/5} n}{n^{1/5} \epsilon^{4/5}}$.
    Suppose $\mu$ is a probability distribution over $\bZ_n$ that is not $4\alpha$-highly concentrated
    with respect to $m = c \cdot \frac{\sqrt{n}}{\epsilon^2 \rho^{3/2}} \log^2 n$.
    Let $T \define \frac{1}{n^2} \sum_{i,j} \phi_{i,j} + \beta \frac{1}{n} \epsilon^2 \rho$ be the
    threshold used by the second step of \cref{alg:tester-confused-collector}.
    Then the test statistic $\bm{Y}$ satisfies the following:
    \begin{enumerate}
        \item (Completeness) If $\mu = \nu$, then $\bm{Y} < T$ with probability at least $99/100$;
        \item (Soundness) If $\dist_\TV(\mu, \nu) > \epsilon$, then $\bm{Y} > T$ with probability at
            least $99/100$.
    \end{enumerate}
\end{lemma}
\begin{proof}
    \textbf{Completeness.} Suppose $\mu = \nu$.
    By \cref{prop:expectation-y-uniform-confused-collector}, $\bm{Y}$ satisfies
    $\Ex{\bm{Y}} = \frac{1}{n^2} \sum_{i,j} \phi_{i,j}$.
    Hence for any fixed $\beta$ (to be chosen below),
    it suffices to show that $\bm{Y} < \Ex{\bm{Y}} + \beta \frac{1}{n} \epsilon^2 \rho$ with good
    probability. By \cref{lemma:concentration-specialized} and using the fact that
    $\|\nu\|_2^2 = 1/n$,
    \[
        \Pr{\bm{Y} \ge \Ex{\bm{Y}} + \beta \frac{1}{n} \epsilon^2 \rho}
        \le \Pr{\abs*{\bm{Y} - \Ex{\bm{Y}}} \ge \beta \frac{1}{n} \epsilon^2 \rho}
        \le \frac{\|\mu\|_2^2}{\rho m^2 \left(\beta \frac{1}{n} \epsilon^2 \rho\right)^2}
            \cdot O(\log^4 n)
        = \frac{O(n \log^4 n)}{\beta^2 m^2 \epsilon^4 \rho^3} \,,
    \]
    and we have
    \begin{equation}
        \label{eq:m-requirement-1}
        \frac{O(n \log^4 n)}{\beta^2 m^2 \epsilon^4 \rho^3} \le 1/100
        \iff m \ge \frac{1}{\beta} \cdot
                O\left( \frac{\sqrt{n}}{\epsilon^2 \rho^{3/2}} \log^2 n\right) \,,
    \end{equation}
    as desired. Thus, there exists constant
    $c^{(1)} = c^{(1)}_{\alpha,\beta}$ such that if
    $m \ge c^{(1)} \frac{\sqrt{n}}{\epsilon^2 \rho^{3/2}} \log^2 n$
    then, for all sufficiently large $n$, $\bm{Y} < T$ with probability at least $99/100$.

    \textbf{Soundness.} We proceed similarly.
    We first note that the conditions of \cref{lemma:separation-expected-value-confused-collector}
    are met: the condition $\|\mu\|_\infty \le \frac{C \log n}{m}$ holds with $C = 4\alpha$ by
    \cref{remark:non-concentrated-l-infty-confused-collector}, and the lower bound on $\rho$ holds
    for the following reason. First, since $\rho \le 1$ always, its lower bound implies that
    $\epsilon \ge \Omega(n^{-1/4})$:
    \[
        \frac{L \log^{4/5} n}{n^{1/5} \epsilon^{4/5}} \le \rho \le 1
        \implies \epsilon \ge \frac{L^{5/4} \log n}{n^{1/4}} \,.
    \]
    Then, we conclude that $\rho = \Omega\left(\frac{\log^{2/3} n}{n^{1/3} \epsilon^{4/3}}\right)$,
    since $\rho \ge \frac{L \log^{4/5} n}{n^{1/5} \epsilon^{4/5}}$ and
    \[
        \frac{\left(\frac{\log^{4/5} n}{n^{1/5} \epsilon^{4/5}}\right)}
             {\left(\frac{\log^{2/3} n}{n^{1/3} \epsilon^{4/3}}\right)}
        = (\log n)^{2/15} n^{2/15} \epsilon^{8/15}
        \ge (\log n)^{2/15} n^{2/15} \Omega(n^{-2/15})
        = \omega(1) \,.
    \]
    Thus, writing $\mu = \nu + z$, by
    \cref{lemma:separation-expected-value-confused-collector,prop:expectation-y-uniform-confused-collector}
    we have
    \[
        \Ex{\bm{Y}} > \Ex{\bm{Y}^{(\nu)}} + \frac{\rho}{8} \|z\|_2^2
        = \frac{1}{n^2}\sum_{i,j} \phi_{i,j} + \frac{\rho}{8} \|z\|_2^2 \,.
    \]
    Therefore it suffices to show that, for appropriately chosen $\beta$, we have
    \[
        \bm{Y} \gtquestion \Ex{\bm{Y}} - \frac{\rho}{8} \|z\|_2^2 + \beta \frac{1}{n} \epsilon^2 \rho \,.
    \]
    Recall that, when $\dist_\TV(\mu, \nu) > \epsilon$, we have $\|z\|_1 > 2\epsilon$, which implies
    $\|z\|_2^2 > 4\epsilon^2 / n$ and hence
    $\frac{\rho}{8} \|z\|_2^2 > \frac{1}{2n} \epsilon^2 \rho$, so that
    \[
        \Ex{\bm{Y}} - \frac{\rho}{8} \|z\|_2^2 + \beta \frac{1}{n} \epsilon^2 \rho
        < \Ex{\bm{Y}} - \frac{\rho}{8} \|z\|_2^2 + 2\beta \frac{\rho}{8} \|z\|_2^2
        = \Ex{\bm{Y}} - (1-2\beta) \frac{\rho}{8} \|z\|_2^2 \,.
    \]
    Thus, for $\beta \le 1/3$, we have $1-2\beta \ge \beta$ and it suffices to show that the
    following holds with probability at least $99/100$:
    \[
        \bm{Y} \gtquestion \Ex{\bm{Y}} - \beta \frac{\rho}{8} \|z\|_2^2 \,.
    \]
    We apply \cref{lemma:concentration-specialized} again, along with
    $\|\mu\|_2^2 = \|\nu\|_2^2 + \|z\|_2^2 = \frac{1}{n} + \|z\|_2^2$ and
    $\|z\|_2^2 \ge \frac{4\epsilon^2}{n}$.
    \begin{align*}
        \Pr{\bm{Y} \le \Ex{\bm{Y}} - \beta \frac{\rho}{8} \|z\|_2^2}
        &\le \Pr{\abs*{\bm{Y} - \Ex{\bm{Y}}} \ge \beta \frac{\rho}{8} \|z\|_2^2}
        \le \frac{\|\mu\|_2^2}{\rho m^2 \left( \beta \frac{\rho}{8} \|z\|_2^2 \right)^2}
            \cdot O(\log^4 n) \\
        &= \frac{(\frac{1}{n} + \|z\|_2^2) O(\log^4 n)}{\beta^2 m^2 \rho^3 \|z\|_2^4}
        \le \frac{\frac{1}{n} O(\log^4 n)}{\beta^2 m^2 \rho^3 (\epsilon^2 / n)^2}
            + \frac{O(\log^4 n)}{\beta^2 m^2 \rho^3 (\epsilon^2 / n)}
        = \frac{O(n \log^4 n)}{\beta^2 m^2 \rho^3 \epsilon^4} \,.
    \end{align*}
    This failure probability is asymptotically the same as that obtained in the completeness case.
    Thus there exists a constant $c^{(2)} = c^{(2)}_{\alpha,\beta} > 0$ such that, for
    $m \ge c^{(2)} \frac{\sqrt{n}}{\epsilon^2 \rho^{3/2}} \log^2 n$
    and all sufficiently large $n$, $\bm{Y} > T$ with probability at least $99/100$.
    Setting $c = \max\left\{ c^{(1)}_{\alpha,\beta}, c^{(2)}_{\alpha,\beta} \right\}$
    concludes the proof.
\end{proof}

\noindent
Finally, we establish correctness by combining our results for the two steps of the tester:

\begin{theorem}[Formal version of \cref{thm:intro-confused-collector-main}]
\label{thm:confused-collector-main}
    There exist constants $\alpha > 0$, $\beta > 0$, $c = c_{\alpha,\beta} > 0$,
    and $L = L_c > 0$ such that
    the following holds for all sufficiently large $n$. Suppose $\epsilon, \rho \in (0, 1]$ satisfy
    $\rho \ge \frac{L \log^{4/5} n}{n^{1/5} \epsilon^{4/5}}$ and let $G$ be either the cycle or the
    path on $n$ vertices.
    Then \cref{alg:tester-confused-collector} (instantiated with constants
    $\alpha, \beta, L \text{ and } c$) is an $(\epsilon,\rho)$-tester for uniformity with sample
    complexity $\Theta\left(\frac{\sqrt{n}}{\epsilon^2 \rho^{3/2}} \log^2 n\right)$ and query
    complexity 0.
\end{theorem}
\begin{proof}
    We start by instantiating $\alpha > 0$ large enough as per
    \cref{lemma:concentration-test-eta}. That lemma also requires that $m \le n \rho$, which we
    now verify. Fix any constant $c > 0$ and suppose
    $m = c \cdot \frac{\sqrt{n}}{\epsilon^2 \rho^{3/2}} \log^2 n$. Then
    \[
        m \le n \rho
        \iff c \cdot \frac{\sqrt{n}}{\epsilon^2} \cdot \frac{\log^2 n}{\rho^{3/2}} \le n \rho
        \iff \rho^{5/2} \ge c \cdot \frac{\log^2 n}{n^{1/2} \epsilon^2}
        \iff \rho \ge \frac{c^{2/5} \log^{4/5} n}{n^{1/5} \epsilon^{4/5}} \,.
    \]
    Therefore, for any choice of $c$, setting $L = L_c \ge c^{2/5}$ ensures that $m \le n \rho$.

    Thus, instantiate $\beta, c > 0$ as provided by
    \cref{lemma:tester-second-step-confused-collector}, and the corresponding $L_c$ as above.
    Now, we can use \cref{lemma:concentration-test-eta,lemma:tester-second-step-confused-collector}
    to establish overall correctness of the tester. (Note that the sample complexity claim follows
    from the specification of the algorithm.)

    \textbf{Completeness.} By \cref{lemma:concentration-test-eta}, the first step of the tester
    rejects only with probability at most $1/100$. Likewise, by
    \cref{lemma:tester-second-step-confused-collector}, the second step of the tester rejects
    only with probability at most $1/100.$ Hence the total rejection probability is at most
    $2/100 < 1/10$.

    \textbf{Soundness.} There are two cases depending on the concentration of $\mu$. First, suppose
    $\mu$ is $4\alpha$-highly concentrated. Then the first step of the tester rejects with
    probability at least $99/100$ by \cref{lemma:concentration-test-eta}. On the other hand, if
    $\mu$ is not $4\alpha$-highly concentrated, then the second step of the tester rejects with
    probability at least $99/100$ by \cref{lemma:tester-second-step-confused-collector}. Either
    way, the tester rejects with probability at least $99/100 > 9/10$.
\end{proof}

\begin{remark}
    \label{remark:confused-collector-parameter-ranges}
    Our tester and analysis cover the range of resolution parameter $\rho \ge \widetilde
    \Theta\left( n^{-1/5} \epsilon^{-4/5} \right)$, which, as shown in the proof of
    \cref{lemma:tester-second-step-confused-collector}, also implies an $\epsilon \ge \widetilde
    \Theta\left( n^{-1/4} \right)$ bound. As shown above, together with our sample complexity these
    bounds ensure that $m \le n \rho$, which was required by our
    \cref{lemma:concentration-test-eta}. An interesting question is whether there exists a
    zero-query tester for more extreme parameter ranges, in particular those requiring sample
    complexity $m \gg n$. For the situation where queries are allowed, see the next subsection.
\end{remark}

\subsection{Testing Uniformity With Queries}
\label{section:random-with-queries}

We show that the uniformity testing result from \cref{thm:confused-collector-main} can be improved
if queries are allowed. The tester will use the instance-optimal identity tester of \cite{VV17} to
test identity on the \emph{singletons} of the random clustering.

\begin{definition}
    Let $\cX$ be a set and let $\Gamma$ be a clustering of $\cX$. We say that $x \in \cX$ is a
    \emph{$\Gamma$-singleton} if $\Gamma_{\gamma(x)} = \{ x \}$.
\end{definition}

\begin{theorem}
\label{thm:technical-random-with-queries}
There exists a constant $L > 0$ such that the following holds for all sufficiently large $n \in
\bN$. Let $G$ be the path or cycle on $\bZ_n$. Let $\epsilon \in (0, 1)$ and let $\rho \in \left[ L
(n \epsilon)^{-1/4}, 1 \right]$. Then there is an $(\epsilon,\rho)$-tester for uniformity on $G$
using $O\left( \frac{\sqrt{n}}{\rho \epsilon^2} \right)$ clustered samples and $O(\rho n \log n)$
label queries.
\end{theorem}
\begin{proof}
    Let $\nu$ be the uniform distribution over $\bZ_n$. Let $c_1, c_2 > 0$ be universal constants to
    be specified later. The algorithm is the following:
    \begin{enumerate}
        \item Using binary search, attempt to learn every cell of the clustering using $O(\log n)$
          label queries per cell. Stop and reject if more than $c_1 \cdot \rho n$ cells are
          discovered, so that $O( \rho n \log n)$ queries are used at most.
        \item Let $S \subseteq \bZ_n$ be the set of $\Gamma$-singletons of $\bZ_n$, and let
            $\overline S \define \bZ_n \setminus S$. If $|S| > 3\rho^2 n$, halt and reject.
            Let $\Gamma'$ be the clustering
            $\{ \overline S \} \cup \{ \{ x \} \;|\; x \in S \}$,
            \ie $\Gamma'$ has one cluster for each $\Gamma$-singleton plus one
            cluster for any remaining elements. Run the instance-optimal identity tester of
            \cite{VV17} on input distribution $\induced{\mu}{\Gamma'}$ against target distribution
            $\induced{\nu}{\Gamma'}$ with proximity parameter $c_2 \rho^2 \epsilon$ and success
            probability $99/100$, and accept or reject according to the output of the tester.
    \end{enumerate}

    \textbf{First step.}
    We first claim that the first step rejects only with probability at most $1/100$. Note that the
    number of clusters satisfies $|\bm{\Gamma}| \le 1 + \sum_{e \in G} \ind{e \not\in \bm{H}}$, and
    this sum ranges over edges $e = (x,x+1)$ for $0 \le x \le n-1$ when $G$ is the cycle, and $0 \le
    x \le n-2$ when $G$ is the path. Since each edge of $G$ appears in $\bm{H}$ with probability
    $1-\rho$, the expected number of clusters is $\Ex{|\bm{\Gamma}|} \le 1 + \rho n \le 2 \rho n$,
    the last inequality since $\rho n \ge L n^{3/4} \epsilon^{-1/4} \ge 1$. Choosing $c_1 = 200$,
    Markov's inequality yields that $|\bm{\Gamma}| \ge c_1 \cdot \rho n = 200 \rho n$ only with
    probability at most $1/100$. 

    \textbf{Number of singletons.}
    Let us also upper bound the number of singletons. Let $A_1 \define \{2k : k \in \bZ_{\ge 0}\}
    \cap \{1, \dotsc, n-2\}$, and let $A_2 \define \{2k+1 : k \in \bZ_{\ge 0}\} \cap \{1, \dotsc,
    n-2\}$. Note that for each $i$, the vertices of $G$ indexed by elements of $A_i$ form an
    independent set, so the random variables $\{ \ind{x \in \bm{S}} \;|\; x \in A_i \}$ are mutually
    independent. Moreover, for each $x \in A_i$ we have that $(x-1, x)$ and $(x, x+1)$ are both
    edges of $G$, so $\Pr{x \in \bm{S}} = \Pr{(x-1, x) \not\in \bm{H} \text{ and } (x, x+1) \not\in
    \bm{H}} = \rho^2$. For each $i \in \{1,2\}$, define $\bm{X}_i \define |\bm{S} \cap A_i| =
    \sum_{x \in A_i} \ind{x \in \bm{S}}$. Let $a_i \define |A_i|$, and note that $n/3 \le a_i \le
    n/2$ for all sufficiently large $n$. Then $\Ex{\bm{X}_i} = \rho^2 a_i$ and a Chernoff bound
    yields, for any constant $0 < \delta < 1$,
    \[
        \Pr{\bm{X}_i \ge (1+\delta) \rho^2 a_i}
        \le e^{-\frac{\delta^2 \rho^2 a_i}{3}}
        \le e^{-\frac{\delta^2 \rho^2 n}{9}}
        \le 1/100 \,,
    \]
    the third inequality since $\rho^2 n \ge L^2 n^{1/2} \epsilon^{-1/2} = \omega_n(1)$. Then,
    observing that $|\bm{S}| \le \bm{X}_1 + \bm{X}_2 + 2$ since at most two elements of $\bm{S}$ do
    not appear in $A_1 \cup A_2$, the Chernoff and union bounds yield
    \[
        \Pr{|\bm{S}| \ge 3 \rho^2 n}
        \le \Pr{|\bm{S}| - 2 \ge 2 \rho^2 n}
        \le \Pr{\bm{X}_1 \ge \rho^2 n} + \Pr{\bm{X}_2 \ge \rho^2 n}
        \le 2/100 \,,
    \]
    the first inequality since $\rho^2 n \ge 2$. Thus $|\bm{S}| \le 3 \rho^2 n$ except with
    probability at most $2/100$.

    \textbf{Second step (correctness).}
    To show correctness of the second step, we need to show that the identity tester will accept
    when $\mu = \nu$ and reject when $\dist_\TV(\mu, \nu) > \epsilon$, both with probability at
    least (say) $98/100$. When $\mu = \nu$, we have $\induced{\mu}{\bm{\Gamma}'} =
    \induced{\nu}{\bm{\Gamma}'}$ for any value of $\bm{\Gamma}$, so the identity tester will indeed
    accept.

    Now, suppose $\dist_\TV(\mu, \nu) > \epsilon$. We need to show that
    $\dist_\TV(\induced{\mu}{\bm{\Gamma}'}, \induced{\nu}{\bm{\Gamma}'}) \ge c_2 \rho^2 \epsilon$
    with good probability. We again break the analysis into even- and odd-indexed elements to make
    use of independence. For each $x \in \bZ_n$, let $b_x \define \left( \frac{1}{n} - \mu(x)
    \right)^+$. For each $i \in \{1,2\}$, let
    \[
        w_i \define
        \sum_{x \in A_i} \left( \nu(x) - \mu(x) \right)^+
        = \sum_{x \in A_i} b_x \,.
    \]
    Let $\epsilon_0 \define \dist_\TV(\mu, \nu)$. Recalling that $\mu$ and $\nu$ have total mass $1$
    each, we have
    \[
        \epsilon_0
        = \dist_\TV(\mu, \nu)
        = \sum_{x \in \bZ_n} b_x
        \le \left( \sum_{x \in A_1} b_x \right)
            + \left( \sum_{x \in A_2} b_x \right)
            + b_0 + b_{n-1}
        \le w_1 + w_2 + \frac{2}{n} \,.
    \]
    Since $\epsilon_0 \ge \epsilon \ge L^4 \rho^{-4} n^{-1} \ge 4/n$, we conclude that either $w_1
    \ge \epsilon_0/4$ or $w_2 \ge \epsilon_0/4$. Suppose without loss of generality that $w_1 \ge
    \epsilon_0/4$.

    Define random variable $\bm{y}_x \define \ind{x \in \bm{S}} \cdot b_x$ for each $x \in \bZ_n$.
    Let $\bm{Y} \define \sum_{x \in A_1} \bm{y}_x$. Then
    \begin{align*}
        \dist_\TV(\induced{\mu}{\bm{\Gamma}'}, \induced{\nu}{\bm{\Gamma}'})
        &= \frac{1}{2} \sum_{z \in \bm{\Gamma}'}
            \abs*{\induced{\nu}{\bm{\Gamma}'}(z) - \induced{\mu}{\bm{\Gamma}'}(z)}
        \ge \frac{1}{2} \sum_{z \in \bm{\Gamma}' : z \ne \overline {\bm{S}}}
            \left( \induced{\nu}{\bm{\Gamma}'}(z) - \induced{\mu}{\bm{\Gamma}'}(z) \right)^+ \\
        &= \frac{1}{2} \sum_{x \in \bm{S}} \left( \frac{1}{n} - \mu(x) \right)^+
        \ge \frac{1}{2} \sum_{x \in A_1} \ind{x \in \bm{S}} \cdot b_x
        = \frac{1}{2} \bm{Y} \,.
    \end{align*}
    Thus our task is to show that $\bm{Y} \ge 2 c_2 \rho^2 \epsilon$ with good probability. We have
    \[
        \Ex{\bm{Y}} = \sum_{x \in A_1} \Pr{x \in \bm{S}} \cdot b_x
        = \rho^2 w_1
        \ge \rho^2 \epsilon_0 / 4 \,.
    \]
    The random variables $\bm{y}_x$ for $x \in A_1$ are mutually independent and bounded in $[0,
    b_x]$, so Hoeffding's inequality gives
    \[
        \Pr{\bm{Y} \le \rho^2 \epsilon_0 / 8}
        \le \Pr{\bm{Y} \le \Ex{\bm{Y}} - \rho^2 \epsilon_0 / 8}
        \le e^{-\frac{\rho^4 \epsilon_0^2 / 32}{\sum_{x \in A_1} b_x^2}} \,.
    \]
    Since $\sum_{x \in A_1} b_x = w_1$ and each $b_x$ lies within $[0,1/n]$, it holds that
    $\sum_{x \in A_1} b_x^2 \leq \frac{1}{n^2} \cdot \frac{w_1}{1/n} = w_1/n$ (which would be
    attained if there were $n\cdot w_1$ values $x$ with $b_x = 1/n$). Also noticing that $w_1
    \le \epsilon_0$ and recalling that $\epsilon_0 \ge \epsilon$, we obtain
    \[
        \Pr{\bm{Y} \le \rho^2 \epsilon_0 / 8}
        \le e^{-\frac{\rho^4 \epsilon_0^2 / 32}{w_1 / n}}
        \le e^{-\rho^4 \epsilon_0 n / 32}
        \le e^{-L^4 (n \epsilon)^{-1} \cdot \epsilon n / 32}
        = e^{-L^4 / 32}
        \le 1/100
    \]
    for sufficiently large $L$. Thus $\dist_\TV(\induced{\mu}{\bm{\Gamma}'},
    \induced{\nu}{\bm{\Gamma}'}) \ge \rho^2 \epsilon / 16$ except with probability at most $1/100$.
    Hence the second step, with constant $c_2 = 1/16$, rejects with probability at least $98/100$.

    \textbf{Second step (sample complexity).}
    \sloppy Finally, we show that the second step satisfies the claimed sample complexity. The
    instance-optimal identity tester of \cite{VV17} requires at most $O\left(\max\left\{
        \frac{1}{\delta}, \frac{\|p^{-\max}\|_{2/3}}{\delta^2} \right\}\right)$ samples, where $p$
    is the known target distribution, $p^{-\max}$ is the vector obtained by removing a largest entry
    from $p$, and $\delta$ is the distance parameter. Here we have $p = \induced{\nu}{\bm{\Gamma}'}$
    and $\delta = c_2 \rho^2 \epsilon$. Since removing a non-largest entry as opposed to a largest
    entry from $\induced{\nu}{\bm{\Gamma}'}$ can only overestimate the sample complexity, we may
    upper bound $\|\induced{\nu}{\bm{\Gamma}'}^{-\max}\|_{2/3}$ by removing the entry corresponding
    to $\overline {\bm{S}}$, so that only the singletons remain:
    \[
        \|\induced{\nu}{\bm{\Gamma}'}^{-\max}\|_{2/3}
        \le \left( |\bm{S}| \cdot (1/n)^{2/3} \right)^{3/2} \,.
    \]
    We have shown that $\bm{S} \le 3 \rho^2 n$ except with probability at most $2/100$, in which
    case the algorithm rejects, so
    all we need to show is that the sample complexity is as claimed when $\bm{S} \le 3 \rho^2 n$. In
    this case,
    \[
        \|\induced{\nu}{\bm{\Gamma}'}^{-\max}\|_{2/3}
        \le \left( 3 \rho^2 \cdot n^{1/3} \right)^{3/2}
        = O\left( \rho^3 \sqrt{n} \right)
    \]
    and thus
    \[
        \frac{\|\induced{\nu}{\bm{\Gamma}'}^{-\max}\|_{2/3}}{\delta^2}
        \le O\left( \frac{\sqrt{n}}{\rho \epsilon^2} \right) \,.
    \]
    This term dominates the expression $\max\left\{ \frac{1}{\delta},
    \frac{\|\induced{\nu}{\bm{\Gamma}'}^{-\max}\|_{2/3}}{\delta^2} \right\}$ as long as
    \[
        \frac{1}{\rho^2 \epsilon} \le \frac{\sqrt{n}}{\rho \epsilon^2}
        \iff \rho \ge \frac{\epsilon}{n^{1/2}}
        \impliedby \frac{L}{n^{1/4} \epsilon^{1/4}} \ge \frac{\epsilon}{n^{1/2}}
        \iff L n^{1/4} \ge \epsilon^{5/4} \,,
    \]
    which holds since the RHS is at most $1$, concluding the proof.
\end{proof}

\begin{samepage}
\newpage
\begin{center}
\Large \bf Acknowledgments
\end{center}

We thank Eric Blais for helpful discussions and comments on the presentation of this article, and
Maryam Aliakbarpour for references on testing with imperfect information.  We thank the anonymous
reviewers for their comments and references to related work.

\end{samepage}

\bibliographystyle{alpha}
\bibliography{references.bib}

\newcommand{\etalchar}[1]{$^{#1}$}
\begin{thebibliography}{CDGR18}

\bibitem[ACF{\etalchar{+}}21]{ACFST21}
Jayadev Acharya, Cl{\'e}ment~L Canonne, Cody Freitag, Ziteng Sun, and Himanshu
  Tyagi.
\newblock Inference under information constraints iii: Local privacy
  constraints.
\newblock {\em IEEE Journal on Selected Areas in Information Theory},
  2(1):253--267, 2021.

\bibitem[ACFT19]{ACFT19}
Jayadev Acharya, Cl{\'e}ment Canonne, Cody Freitag, and Himanshu Tyagi.
\newblock Test without trust: Optimal locally private distribution testing.
\newblock In {\em Proceedings, International Conference on Artificial
  Intelligence and Statistics (AISTATS)}, pages 2067--2076. PMLR, 2019.

\bibitem[ACH{\etalchar{+}}20]{ACH+20}
Jayadev Acharya, Cl{\'e}ment~L Canonne, Yanjun Han, Ziteng Sun, and Himanshu
  Tyagi.
\newblock Domain compression and its application to randomness-optimal
  distributed goodness-of-fit.
\newblock In {\em Proceedings, Conference on Learning Theory (COLT)}, pages
  3--40. PMLR, 2020.

\bibitem[ACT19]{ACT19}
Jayadev Acharya, Cl{\'e}ment~L Canonne, and Himanshu Tyagi.
\newblock Inference under information constraints: Lower bounds from chi-square
  contraction.
\newblock In {\em Proceedings, Conference on Learning Theory (COLT)}, pages
  3--17. PMLR, 2019.

\bibitem[ACT20]{ACT20}
Jayadev Acharya, Cl{\'e}ment~L Canonne, and Himanshu Tyagi.
\newblock Inference under information constraints ii: Communication constraints
  and shared randomness.
\newblock {\em IEEE Transactions on Information Theory}, 66(12):7856--7877,
  2020.

\bibitem[AGP{\etalchar{+}}19]{AGPRY19}
Maryam Aliakbarpour, Themis Gouleakis, John Peebles, Ronitt Rubinfeld, and Anak
  Yodpinyanee.
\newblock Towards testing monotonicity of distributions over general posets.
\newblock In {\em Proceedings, Conference on Learning Theory (COLT)}, pages
  34--82. PMLR, 2019.

\bibitem[AKR19]{AKR19}
Maryam Aliakbarpour, Ravi Kumar, and Ronitt Rubinfeld.
\newblock Testing mixtures of discrete distributions.
\newblock In {\em Proceedings, Conference on Learning Theory (COLT)}, pages
  83--114. PMLR, 2019.

\bibitem[AS20]{AS20}
Maryam Aliakbarpour and Sandeep Silwal.
\newblock Testing properties of multiple distributions with few samples.
\newblock In {\em Proceedings, Innovations in Theoretical Computer Science
  (ITCS)}, 2020.

\bibitem[BCG19]{BCG19}
Eric Blais, Cl{\'e}ment~L Canonne, and Tom Gur.
\newblock Distribution testing lower bounds via reductions from communication
  complexity.
\newblock {\em ACM Transactions on Computation Theory}, 11(2):1--37, 2019.

\bibitem[BKR04]{BKR04}
Tugkan Batu, Ravi Kumar, and Ronitt Rubinfeld.
\newblock Sublinear algorithms for testing monotone and unimodal distributions.
\newblock In {\em Proceedings of the thirty-sixth annual ACM symposium on
  Theory of computing}, pages 381--390, 2004.

\bibitem[Can16]{Can16}
Cl{\'e}ment~L Canonne.
\newblock Are few bins enough: Testing histogram distributions.
\newblock In {\em Proceedings, ACM Symposium on Principles of Database Systems
  (PODS)}, pages 455--463, 2016.

\bibitem[Can17]{Can17}
Cl{\'e}ment Canonne.
\newblock A short note on poisson tail bounds, 2017.
\newblock
  \url{http://www.cs.columbia.edu/~ccanonne/files/misc/2017-poissonconcentration.pdf}.

\bibitem[Can20]{Can20}
Cl{\'e}ment~L Canonne.
\newblock A survey on distribution testing: Your data is big. but is it blue?
\newblock {\em Theory of Computing}, pages 1--100, 2020.

\bibitem[Can22]{Can22}
Cl{\'e}ment Canonne.
\newblock Topics and techniques in distribution testing: A biased but
  representative sample.
\newblock {\em Foundations and Trends in Communications and Information
  Theory}, 19(6):1032--1198, 2022.

\bibitem[CDGR18]{CDGR18}
Cl{\'e}ment~L Canonne, Ilias Diakonikolas, Themis Gouleakis, and Ronitt
  Rubinfeld.
\newblock Testing shape restrictions of discrete distributions.
\newblock {\em Theory of Computing Systems}, 62(1):4--62, 2018.

\bibitem[CDKL22]{CDKL22}
Cl{\'e}ment~L Canonne, Ilias Diakonikolas, Daniel Kane, and Sihan Liu.
\newblock Nearly-tight bounds for testing histogram distributions.
\newblock {\em Proceedings, Advances in Neural Information Processing Systems
  (NeurIPS)}, 35:31599--31611, 2022.

\bibitem[CDVV14]{CDVV14}
Siu-On Chan, Ilias Diakonikolas, Paul Valiant, and Gregory Valiant.
\newblock Optimal algorithms for testing closeness of discrete distributions.
\newblock In {\em Proceedings of the twenty-fifth annual ACM-SIAM symposium on
  Discrete algorithms}, pages 1193--1203. SIAM, 2014.

\bibitem[CFG{\etalchar{+}}22]{CFGMS22}
Sourav Chakraborty, Eldar Fischer, Arijit Ghosh, Gopinath Mishra, and Sayantan
  Sen.
\newblock Testing of index-invariant properties in the huge object model, 2022.
\newblock arXiv:2207.12514.

\bibitem[CW20]{CW20}
Cl{\'e}ment~L Canonne and Karl Wimmer.
\newblock Testing data binnings.
\newblock In {\em Proceedings of APPROX/RANDOM}. Schloss
  Dagstuhl-Leibniz-Zentrum fur Informatik GmbH, Dagstuhl Publishing, 2020.

\bibitem[CW21]{CW21}
Cl{\'e}ment~L Canonne and Karl Wimmer.
\newblock Identity testing under label mismatch.
\newblock In {\em Proceedings, International Symposium on Algorithms and
  Computation (ISAAC)}. Schloss Dagstuhl-Leibniz-Zentrum f{\"u}r Informatik,
  2021.

\bibitem[DGPP19]{DGPP19}
Ilias Diakonikolas, Themis Gouleakis, John Peebles, and Eric Price.
\newblock Collision-based testers are optimal for uniformity and closeness.
\newblock {\em Chicago Journal of Theoretical Computer Science}, 1:1--21, 2019.

\bibitem[DGTZ18]{DGTZ18}
Constantinos Daskalakis, Themis Gouleakis, Chistos Tzamos, and Manolis
  Zampetakis.
\newblock Efficient statistics, in high dimensions, from truncated samples.
\newblock In {\em Proceedings, IEEE Symposium on Foundations of Computer
  Science (FOCS)}, pages 639--649. IEEE, 2018.

\bibitem[DK16]{DK16}
Ilias Diakonikolas and Daniel~M Kane.
\newblock A new approach for testing properties of discrete distributions.
\newblock In {\em Proceedings, IEEE Symposium on Foundations of Computer
  Science (FOCS)}, pages 685--694. IEEE, 2016.

\bibitem[DKN15]{DKN15b}
Ilias Diakonikolas, Daniel~M Kane, and Vladimir Nikishkin.
\newblock Testing identity of structured distributions.
\newblock In {\em Proceedings, ACM-SIAM Symposium on Discrete Algorithms
  (SODA)}, pages 1841--1854. SIAM, 2015.

\bibitem[DNNR11]{DNNR11}
Khanh {Do Ba}, Huy~L Nguyen, Huy~N Nguyen, and Ronitt Rubinfeld.
\newblock Sublinear time algorithms for earth mover’s distance.
\newblock {\em Theory of Computing Systems}, 48:428--442, 2011.

\bibitem[DNS23]{DNS23}
Anindya De, Shivam Nadimpalli, and Rocco~A Servedio.
\newblock Testing convex truncation.
\newblock In {\em Proceedings, ACM-SIAM Symposium on Discrete Algorithms
  (SODA)}, pages 4050--4082. SIAM, 2023.

\bibitem[FH23]{FHparity}
Renato {Ferreira Pinto Jr.} and Nathaniel Harms.
\newblock Distribution testing under the parity trace, 2023.
\newblock arXiv:2304.01374.

\bibitem[FKKT21]{FKKT21}
Dimitris Fotakis, Alkis Kalavasis, Vasilis Kontonis, and Christos Tzamos.
\newblock Efficient algorithms for learning from coarse labels.
\newblock In {\em Proceedings, Conference on Learning Theory (COLT)}, pages
  2060--2079. PMLR, 2021.

\bibitem[Gol20]{Gol16}
Oded Goldreich.
\newblock The uniform distribution is complete with respect to testing identity
  to a fixed distribution.
\newblock In {\em Computational Complexity and Property Testing: On the
  Interplay Between Randomness and Computation}. Springer, 2020.
\newblock ECCC TR16-015.

\bibitem[GR11]{GR00}
Oded Goldreich and Dana Ron.
\newblock On testing expansion in bounded-degree graphs.
\newblock In {\em Studies in Complexity and Cryptography. Miscellanea on the
  Interplay between Randomness and Computation}, pages 68--75. Springer, 2011.

\bibitem[GR18]{GR18}
Marco Gaboardi and Ryan Rogers.
\newblock Local private hypothesis testing: Chi-square tests.
\newblock In {\em Proceedings, International Conference on Machine Learning
  (ICML)}, pages 1626--1635. PMLR, 2018.

\bibitem[GR22]{GR22}
Oded Goldreich and Dana Ron.
\newblock Testing distributions of huge objects.
\newblock In {\em Proceedings, Innovations in Theoretical Computer Science
  (ITCS)}. Schloss Dagstuhl-Leibniz-Zentrum f{\"u}r Informatik, 2022.

\bibitem[Gra06]{Gra06}
Robert~M Gray.
\newblock Toeplitz and circulant matrices: A review.
\newblock {\em Foundations and Trends in Communications and Information
  Theory}, 2(3):155--239, 2006.

\bibitem[HY22]{HY22}
Nathaniel Harms and Yuichi Yoshida.
\newblock Downsampling for testing and learning in product distributions.
\newblock In {\em Proceedings, International Colloquium on Automata, Languages
  and Programming (ICALP)}. Schloss Dagstuhl-Leibniz-Zentrum f{\"u}r
  Informatik, 2022.

\bibitem[ILR12]{ILR12}
Piotr Indyk, Reut Levi, and Ronitt Rubinfeld.
\newblock Approximating and testing k-histogram distributions in sub-linear
  time.
\newblock In {\em Proceedings, ACM Symposium on Principles of Database Systems
  (PODS)}, pages 15--22, 2012.

\bibitem[IT03]{IT03}
Piotr Indyk and Nitin Thaper.
\newblock Fast image retrieval via embeddings.
\newblock In {\em International Workshop on Statistical and Computational
  Theories of Vision (ICCV)}, 2003.

\bibitem[LRR13]{LRR13}
Reut Levi, Dana Ron, and Ronitt Rubinfeld.
\newblock Testing properties of collections of distributions.
\newblock {\em Theory of Computing}, 9(1):295--347, 2013.

\bibitem[LRR14]{LRR14}
Reut Levi, Dana Ron, and Ronitt Rubinfeld.
\newblock Testing similar means.
\newblock {\em SIAM Journal on Discrete Mathematics}, 28(4):1699--1724, 2014.

\bibitem[LSV20]{LSV20}
Yin~Tat Lee, Aaron Sidford, and Santosh~S Vempala.
\newblock Efficient convex optimization with oracles.
\newblock In {\em Building Bridges II: Mathematics of L{\'a}szl{\'o}
  Lov{\'a}sz}, pages 317--335. Springer, 2020.

\bibitem[Pan08]{Pan08}
Liam Paninski.
\newblock A coincidence-based test for uniformity given very sparsely sampled
  discrete data.
\newblock {\em IEEE Transactions on Information Theory}, 54(10):4750--4755,
  2008.

\bibitem[Rio37]{Rio37}
John Riordan.
\newblock Moment recurrence relations for binomial, poisson and hypergeometric
  frequency distributions.
\newblock {\em The Annals of Mathematical Statistics}, 8(2):103--111, 1937.

\bibitem[RV23]{RV23}
Ronitt Rubinfeld and Arsen Vasilyan.
\newblock Testing distributional assumptions of learning algorithms.
\newblock In {\em Proceedings, ACM Symposium on Theory of Computing (STOC)}.
  {ACM}, 2023.

\bibitem[She18]{She18}
Or~Sheffet.
\newblock Locally private hypothesis testing.
\newblock In {\em Proceedings, International Conference on Machine Learning
  (ICML)}, pages 4605--4614. PMLR, 2018.

\bibitem[Sra]{Sra}
Suvrit Sra.
\newblock Diagonalizing a certain real and symmetric toeplitz matrix.
\newblock MathOverflow.
\newblock \url{https://mathoverflow.net/q/68471} (version: 2011-07-05).

\bibitem[VV17]{VV17}
Gregory Valiant and Paul Valiant.
\newblock An automatic inequality prover and instance optimal identity testing.
\newblock {\em SIAM Journal on Computing}, 46(1):429--455, 2017.

\end{thebibliography}

\appendix
\addtocontents{toc}{\protect\setcounter{tocdepth}{1}}

\section{EMD Inequality and Testing on the Hypergrid}
\label{appendix:emd}

In this appendix, we prove the inequality we require between the EMD and total variation distance of
distributions under a hierarchical clustering, and apply this result to obtain identity and
equivalence testers for the hypergrid under EMD. These results are relatively straightforward
adaptations of \cite{IT03,DNNR11}.

\subsection{EMD-Total Variation Inequality under Hierarchical Clusterings}
\label{appendix:proof-of-tv-lemma}

\begin{definition}
    Let $(\cX, \dist)$ be a metric space and let $\Gamma^{(1)}, \dotsc, \Gamma^{(t)}$ be clusterings
    of $\cX$. We say $(\Gamma^{(1)}, \dotsc, \Gamma^{(t)})$ is a \emph{hierarchical clustering} of
    $\cX$ if $\Gamma^{(i+1)}$ is a refinement of $\Gamma^{(i)}$ for each $i \in [t-1]$. We say that
    $(\Gamma^{(1)}, \dotsc, \Gamma^{(t)})$ has \emph{diameter bounds} $\delta_1, \dotsc, \delta_t$
    if $\diam_\dist\left( \Gamma^{(i)}_j \right) \le \delta_i$ for each $i \in [t]$ and
    $\Gamma^{(i)}_j \in \Gamma^{(i)}$. For convenience, we also define $\delta_0 \define
    \diam_\dist(\cX)$.
\end{definition}

The following lemma relates the EMD to the TV distance of clustered distributions, as shown by
\cite{IT03} and applied by \cite{DNNR11} to the problem of EMD testing. These works focused on the
case of a fixed explicit clustering or max-diameter error bounds, but the same proof (which we
sketch below for completeness) applies to the formulation we require.

\begin{lemma}
    \label{lemma:emd-inequality-hierarchical}
    Let $(\cX, \dist)$ be a finite metric space and let $(\Gamma^{(1)}, \dotsc, \Gamma^{(t)})$ be a
    hierarchical clustering of $\cX$ with diameter bounds $\delta_1, \dotsc, \delta_t$. Equip each
    $\Gamma^{(i)}$ with its cluster assignment function function $\gamma^{(i)}$. Then for every pair
    of probability distributions $\mu, \nu$ over $\cX$,
    \[
        \EMD_\dist(\mu, \nu)
        \le \sum_{i=1}^t
                \delta_{i-1} \dist_\TV(\induced{\mu}{\Gamma^{(i)}}, \induced{\nu}{\Gamma^{(i)}})
            + \Exu{\bm{x} \sim \mu}{\diam\left( \Gamma^{(t)}_{\gamma^{(t)}(\bm{x})} \right)} \,.
    \]
\end{lemma}
\begin{proof}[Proof sketch]

    For any measure $m$ on $\cX \times \cX$ and $z \in \cX$, define $m(z, *) \define \sum_{y \in
    \cX} m(z, y)$ and $m(*, z) \define \sum_{x \in \cX} m(x, z)$. Also, define two trivial
    clusterings: $\Gamma^{(0)}$ containing the single cluster $\cX$, and $\Gamma^{(t+1)}$ containing
    singleton clusters $\{x\}$ for each $x \in \cX$.

    We give a coupling $\pi \in \Pi(\mu, \nu)$ iteratively by a sequence of measures $\pi^{(-2)},
    \pi^{(-1)}, \pi^{(0)}, \dotsc, \pi^{(t)}$ over $\cX \times \cX$ as follows. Let $\pi^{(-2)}$ be
    zero everywhere. Then for each $i = -1, \dotsc, t$,
    \begin{enumerate}
        \item Define $r^{(i)}, s^{(i)} : \cX \to \bR_{\ge 0}$ by $r^{(i)}(x) \define \mu(x) -
            \pi^{(i-1)}(x, *)$ and $s^{(i)}(y) \define \nu(y) - \pi^{(i-1)}(*, y)$.
        \item Let $\rho^{(i)} : \cX \times \cX \to \bR_{\ge 0}$ be a measure satisfying
            \begin{enumerate}
                \item $\rho^{(i)}(x, y) = 0$ if $\gamma^{(t-i)}(x) \ne \gamma^{(t-i)}(y)$;
                \item For each $z \in \cX$, $\rho^{(i)}(z, *) \le r^{(i)}(z)$ and $\rho^{(i)}(*, z)
                    \le s^{(i)}(z)$; and
                \item For each $\Gamma^{(t-i)}_j \in \Gamma^{(t-i)}$, either $\rho(z, *) =
                    r^{(i)}(z)$ for every $z \in \Gamma^{(t-i)}_j$ or $\rho(*, z) = s^{(i)}(z)$ for
                    every $z \in \Gamma^{(t-i)}_j$.
            \end{enumerate}
            The existence of such an object is given by standard arguments, \eg by seeing
            $\rho^{(i)}$ as the flows in a flow network.
        \item Set $\pi^{(i)}(x, y) \define \pi^{(i-1)}(x, y) + \rho^{(i)}(x, y)$ for each $x, y \in
            \cX$.
    \end{enumerate}

    Let $\pi \define \pi^{(t)}$. We first observe that $\pi \in \Pi(\mu, \nu)$. It is clear by
    construction that $\pi(z, *) \le \mu(z)$ and $\pi(*, z) \le \nu(z)$ for every $z \in \cX$, since
    the contributions $\rho^{(i)}$ are constrained to ensure this. On the other hand, property (c)
    in the $t$-th iteration ensures that either $\rho^{(t)}(z, *) = r^{(t)}(z)$ for every $z \in
    \cX$, or $\rho^{(t)}(*, z) = s^{(t)}(z)$ for every $z \in \cX$. This means that either $\pi(x,
    *) = \mu(x)$ for every $z \in \cX$ or $\pi(*, z) = \nu(z)$ for every $z \in \cX$, and either of
    these conditions (along with the aforementioned upper bounds) implies the other. Thus the
    marginals or $\pi$ are $\mu$ and $\nu$, as needed.

    Then, we show that $\Exu{(\bm{x},\bm{y}) \sim \pi}{\dist(x,y)}$ satisfies the claimed upper
    bound. We will need the following observation: for each $i = 0, \dotsc, t$ and
    $\Gamma^{(t-i)}_j \in \Gamma^{(t-i)}$,
    \[
        \sum_{x, y \in \Gamma^{(t-i)}_j} \rho^{(i)}(x, y)
        \le \frac{1}{2} \sum_{\Gamma^{(t-i+1)}_{j'} \subseteq \Gamma^{(t-i)}_j}
                \abs*{\mu\left[ \Gamma^{(t-i+1)}_{j'} \right] - \nu\left[ \Gamma^{(t-i+1)}_{j'}
                \right]} \,.
    \]
    This follows from properties (a)--(c) of the $(i-1)$-st iteration, which imply that for each
    $\Gamma^{(t-i+1)}_j \in \Gamma^{(t-i+1)}$, either $\sum_{x \in \Gamma^{(t-i+1)}_j} r^{(i)}(x) =
    \abs*{\mu\left[ \Gamma^{(t-i+1)}_{j'} \right] - \nu\left[ \Gamma^{(t-i+1)}_{j'} \right]}$ and
    $s^{(i)}(y) = 0$ for every $y \in \Gamma^{(t-i+1)}_j$, or vice versa. Now,
    \begin{align*}
        \Exu{(\bm{x},\bm{y}) \sim \pi}{\dist(x,y)}
        &= \sum_{i=0}^t \sum_{\Gamma^{(t-i)}_j \in \Gamma^{(t-i)}}
            \sum_{x, y \in \Gamma^{(t-i)}_j} \dist(x, y) \rho^{(i)}(x, y) \\
        &\le \sum_{i=1}^{t} \delta_{t-i} \sum_{\Gamma^{(t-i)}_j \in \Gamma^{(t-i)}}
            \frac{1}{2} \sum_{\Gamma^{(t-i+1)}_{j'} \subseteq \Gamma^{(t-i)}_j}
                \abs*{\mu\left[ \Gamma^{(t-i+1)}_{j'} \right] - \nu\left[ \Gamma^{(t-i+1)}_{j'}
                \right]} \\
            &\qquad + \sum_{\Gamma^{(t)}_j \in \Gamma^{(t)}} \sum_{x, y \in \Gamma^{(t)}_j}
                    \diam\left( \Gamma^{(t)}_j \right) \rho^{(0)}(x, y) \\
        &= \sum_{i=1}^{t} \delta_{t-i}
                \dist_\TV\left( \induced{\mu}{\Gamma^{(t-i+1)}}, \induced{\nu}{\Gamma^{(t-i+1)}}
                \right)
            + \sum_{\Gamma^{(t)}_j \in \Gamma^{(t)}} \sum_{x \in \Gamma^{(t)}_j}
                    \diam\left( \Gamma^{(t)}_j \right) \rho^{(0)}(x, *) \\
        &\le \sum_{i=1}^t \delta_{i-1}
                \dist_\TV\left( \induced{\mu}{\Gamma^{(i)}}, \induced{\nu}{\Gamma^{(i)}} \right)
            + \sum_{x \in \cX} \diam\left( \Gamma^{(t)}_{\gamma^{(t)}(x)} \right) \mu(x) \\
        &= \sum_{i=1}^t \delta_{i-1}
                \dist_\TV\left( \induced{\mu}{\Gamma^{(i)}}, \induced{\nu}{\Gamma^{(i)}} \right)
            + \Exu{\bm{x} \sim \mu}{\diam\left( \Gamma^{(t)}_{\gamma^{(t)}(\bm{x})} \right)} \,.
    \end{align*}
    Note that in the first equality, we could skip $i=-1$ in the summation because $\rho^{(-1)}(x,
    y) > 0$ only when $x=y$, in which case $\dist(x, y) = 0$.
\end{proof}

This lemma may also be extended to cover continuous settings such as $[0,1]^d$, \eg by a
discretization argument:

\begin{lemma}
    \label{lemma:continuous-emd-inequality}
    Let $\cX \subset \bR^d$ be a bounded set in Euclidean space. Let $p \ge 1$ and equip $\cX$ with
    normalized $\ell_p$ metric $\dist(\cdot, \cdot)$. Suppose $(\Gamma^{(1)}, \dotsc, \Gamma^{(t)})$
    is a hierarchical clustering of $\cX$ with diameter bounds $\delta_1, \dotsc, \delta_t$,
    such that every cell $\Gamma^{(i)}_j \in \Gamma^{(i)}$ is $\mu$- and $\nu$-measurable.
    Then for every pair of probability distributions $\mu, \nu$ over $\cX$,
    \[
        \EMD_\dist(\mu, \nu)
        \le \sum_{i=1}^t
                \delta_{i-1} \dist_\TV(\induced{\mu}{\Gamma^{(i)}}, \induced{\nu}{\Gamma^{(i)}})
            + \Exu{\bm{x} \sim \mu}{\diam\left( \Gamma^{(t)}_{\gamma^{(t)}(x)} \right)} \,.
    \]
\end{lemma}
\begin{proof}
    Recall that a metric space is said to be \emph{totally bounded} if it has a finite
    $\epsilon$-net for every $\epsilon > 0$. It is a standard result that every bounded subspace of
    Euclidean space is totally bounded. Thus each $\Gamma^{(t)}_j \in \Gamma^{(t)}$ admits a finite
    $\epsilon$-net under Euclidean distance, and since every $\ell_p$ metric in $\bR^d$ is within a
    finite factor of each other, we obtain a finite $\epsilon$-net under metric $\dist(\cdot,
    \cdot)$. Since clusterings have finitely many clusters by definition, we conclude that for every
    $\epsilon > 0$ there exists a finite set $S_\epsilon \subset \cX$ such that $S_\epsilon \cap
    \Gamma^{(t)}_j$ is an $\epsilon$-net for $\Gamma^{(t)}_j$ for every $\Gamma^{(t)}_j \in
    \Gamma^{(t)}$.

    Fix some $\epsilon > 0$. Let $f : \cX \to S_\epsilon$ map each point $x \in \cX$ to a closest
    point in $S_\epsilon \cap \Gamma^{(t)}_{\gamma^{(t)}(x)}$. Note that, since the clustering is
    hierarchical, it follows that $\gamma^{(i)}(x) = \gamma^{(i)}(f(x))$ for every $i \in [t]$. Let
    $\mu'$ be the probability distribution of $f(\bm{x})$ where $\bm{x} \sim \mu$, and similarly let
    $\nu'$ be the probability distribution of $f(\bm{y})$ where $\bm{y} \sim \nu$.

    We first claim that $\EMD_\dist(\mu, \mu') \le \epsilon$. Indeed, the probability distribution
    over $\cX \times \cX$ given by $\pi(x, f(x)) = \mu(x)$ for all $x \in \cX$ and $\pi(x, y) = 0$
    elsewhere is a coupling in $\Pi(\mu, \mu')$, and $\Exu{(\bm{x}, \bm{y}) \sim \pi}{\dist(\bm{x},
    \bm{y})} = \Exu{\bm{x} \sim \mu}{\dist(\bm{x}, f(\bm{x})} \le \epsilon$, the inequality since
    $S_\epsilon \cap \Gamma^{(t)}_{\gamma^{(t)}(x)}$ is an $\epsilon$-net for
    $\Gamma^{(t)}_{\gamma^{(t)}(x)}$. Similarly, we get that $\EMD_\dist(\nu, \nu') \le \epsilon$.

    We also note that for any $i \in [t]$, $\induced{\mu}{\Gamma^{(i)}} =
    \induced{\mu'}{\Gamma^{(i)}}$ and $\induced{\nu}{\Gamma^{(i)}} = \induced{\nu'}{\Gamma^{(i)}}$.
    This holds because each point $x$ is mapped to another point in the same cluster, which does not
    change the induced distributions. Similarly, $\Exu{\bm{x} \sim \mu'}{\diam\left(
    \Gamma^{(t)}_{\gamma^{(t)}(\bm{x})} \right)} = \Exu{\bm{x} \sim \mu}{\diam\left(
    \Gamma^{(t)}_{\gamma^{(t)}(\bm{x})} \right)}$, since the total probability mass in each cluster
    remains unchanged between $\mu$ and $\mu'$. Now, $(S_\epsilon, \dist)$ is a finite metric space,
    so \cref{lemma:emd-inequality-hierarchical} and the triangle inequality give
    \begin{align*}
        \EMD_\dist(\mu, \nu)
        &\le \EMD_\dist(\mu, \mu') + \EMD_\dist(\mu', \nu') + \EMD_\dist(\nu', \nu) \\
        &\le 2\epsilon + \sum_{i=1}^t
                \delta_{i-1} \dist_\TV(\induced{\mu'}{\Gamma^{(i)}}, \induced{\nu'}{\Gamma^{(i)}})
            + \Exu{\bm{x} \sim \mu'}{\diam\left( \Gamma^{(t)}_{\gamma^{(t)}(x)} \right)} \\
        &= 2\epsilon + \sum_{i=1}^t
                \delta_{i-1} \dist_\TV(\induced{\mu}{\Gamma^{(i)}}, \induced{\nu}{\Gamma^{(i)}})
            + \Exu{\bm{x} \sim \mu}{\diam\left( \Gamma^{(t)}_{\gamma^{(t)}(x)} \right)} \,.
    \end{align*}
    This holds for every $\epsilon > 0$, so the claim follows.
\end{proof}

We will also use the special case where there is a single clustering $\Gamma$; for convenience, we
state this simplified version below.

\lemmaemdtvdiameter*

\subsection{EMD Testing on the Hypergrid}
\label{appendix:emd-testing-hypergrid}

Similar to \cite{DNNR11}, \cref{lemma:emd-inequality-hierarchical} implies algorithms for testing
identity and equivalence of distributions with respect to EMD using a hierarchical clustering, with
sample complexity depending on the clustering, as shown below.

\newcommand{\IdS}{m^{\mathsf{id}}}
\newcommand{\EqS}{m^{\mathsf{eq}}}

Let $\IdS(n, \epsilon) \define n^{1/2}/\epsilon^2$ denote the optimal sample complexity of testing
identity of distributions supported on at most $n$ elements with distance parameter $\epsilon$
\cite{VV17}, and let $\EqS(n, \epsilon) \define \max\{ n^{1/2}/\epsilon^2, n^{2/3}/\epsilon^{4/3}
\}$ denote the same for equivalence testing \cite{CDVV14}.

\begin{proposition}
    \label{prop:testing-with-hierarchical-clustering}
    Let $(\cX, \dist)$ be a finite metric space, let $\epsilon > 0$, and let $(\Gamma^{(1)}, \dotsc,
    \Gamma^{(t)})$ be a hierarchical clustering with diameter bounds $\delta_1, \dotsc, \delta_t$
    such that $\delta_t \le \epsilon/2$. There exists an $\EMD_\dist$ identity tester with sample
    complexity
    \[
        O\left( \log(t) \sum_{i=1}^t \IdS\left( |\Gamma^{(i)}|, \frac{\epsilon}{2 t \delta_{i-1}}
        \right) \right)
    \]
    and an $\EMD_\dist$ equivalence tester with sample complexity
    \[
        O\left( \log(t) \sum_{i=1}^t \EqS\left( |\Gamma^{(i)}|, \frac{\epsilon}{2 t \delta_{i-1}}
        \right) \right) \,.
    \]
\end{proposition}
\begin{proof}[Proof]
    Each algorithm iterates over $i = 1, \dotsc, t$, and in each step runs a standard
    identity/equivalence tester on distributions $\induced{\mu}{\Gamma^{(i)}},
    \induced{\nu}{\Gamma^{(i)}}$ with distance parameter $\frac{\epsilon}{2 t \delta_{i-1}}$, paying
    an extra $O(\log t)$ factor in the sample complexity to amplify the success probability to $1 -
    \frac{1}{3t}$; it rejects if that tester rejects in any step, and accepts otherwise. Correctness
    follows from \cref{lemma:emd-inequality-hierarchical} and an averaging argument: if
    $\EMD_\dist(\mu, \nu) > \epsilon$ and $\delta_t \le \epsilon/2$, then there exists $i \in [t]$
    for which $\delta_{i-1} \dist_\TV(\induced{\mu}{\Gamma^{(i)}}, \induced{\nu}{\Gamma^{(i)}}) >
    \frac{\epsilon}{2t}$, so some step of the algorithm rejects with good probability. On the other
    hand, if $\mu = \nu$ then each step only causes the algorithm to reject with probability at most
    $\frac{1}{3t}$, and an union bound completes the proof.
\end{proof}

Applying this general strategy to the hypergrid $[n]^d$, with the same hierarchical clustering as
\cite{DNNR11}, yields the EMD testing results we use. We note that, by virtue of having used the
optimal sample complexity of equivalence testing from \cite{CDVV14} in the reduction, the bounds we
obtain are slightly better than those explicitly computed in \cite{DNNR11} for some values of $d$.

\lemmaemdtestinghypergrid*
\begin{proof}
    We let $t = \log(2/\epsilon)$ and fix the following hierarchical clustering $(\Gamma^{(1)},
    \dotsc, \Gamma^{(t)})$: for each $i \in [t]$, $\Gamma^{(i)}$ partitions the domain into $2^{id}$
    equal-sized hypercubes, so that each cell of $\Gamma^{(i)}$ is a hypercube whose side length is
    a $1/2^i$ fraction of the side length of $[n]^d$. Recall that $\dist$ is a normalized metric, so
    that $[n]^d$ has unit diameter. It follows that $\delta_i = 2^{-i}$, $i \in [t]$, is a diameter
    bound for $(\Gamma^{(1)}, \dotsc, \Gamma^{(t)})$, and indeed $\delta_t \le \epsilon/2$. Then
    \cref{prop:testing-with-hierarchical-clustering} gives an $\EMD_\dist$ identity tester with
    sample complexity
    \[
        \widetilde O\left( \sum_{i=1}^t 
                \frac{\left( 2^{id} \right)^{1/2}}{(\epsilon / 2^{-i+1})^2} \right)
        = \widetilde O\left( \epsilon^{-2} \sum_{i=1}^t \left( 2^{\frac{d}{2} - 2} \right)^i \right)
        \,.
    \]
    and an $\EMD_\dist$ equivalence tester with sample complexity
    \[
        \widetilde O\left( \sum_{i=1}^t \max\left\{
            \frac{\left( 2^{id} \right)^{1/2}}{(\epsilon / 2^{-i+1})^2},
            \frac{\left( 2^{id} \right)^{2/3}}{(\epsilon / 2^{-i+1})^{4/3}} \right\} \right)
        = \widetilde O\left( \max\left\{
            \epsilon^{-2} \sum_{i=1}^t \left( 2^{\frac{d}{2} - 2} \right)^i,
            \epsilon^{-\frac{4}{3}} \sum_{i=1}^t
                \left( 2^{\frac{2}{3} d - \frac{4}{3}} \right)^i \right\} \right)
        \,.
    \]
    The summand $\left(2^{\frac{d}{2} - 2}\right)^i$ is maximized at $i=1$ when $d \le 4$ and at
    $i=t$ when $d \ge 5$, while $\left(2^{\frac{2}{3}d - \frac{4}{3}}\right)^i$ is maximized at
    $i=1$ when $d \le 2$ and at $i=t$ when $d \ge 3$. Additionally, each sum may be upper bounded by
    $t = \log(2/\epsilon)$ times its maximum summand. The sample complexity claims follow.
\end{proof}

\begin{remark}
    \label{remark:continuous-cube}
    At least for input distributions $\mu, \nu$ that are absolutely continuous with respect to the
    Lebesgue measure, \cref{lemma:emd-testing-hypergrid} also holds when the domain is $[0,1]^d$
    instead of $[n]^d$, since in this case we also have the EMD inequality by
    \cref{lemma:continuous-emd-inequality}. In this case, the $\mu$- and $\nu$-measurability
    requirement on the (Lebesgue measurable) cells follows from the Radon-Nikodym theorem.
\end{remark}

\section{Missing Proofs from
\texorpdfstring{\cref{section:random-expectation}}{Section~\ref{section:random-expectation}}}
\label{appendix:missing-proofs-from-random-expectation}

\propzdeltacrossterm*
\begin{proof}
    When $\phi = \phi^\cyc$, we have that $\nu^\top \phi$ is a constant vector (this is true for any
    circulant matrix), and hence $\nu^\top \phi z = 0$ (since $\sum_i z_i = 0$). Therefore we may
    now assume that $\phi = \phi^\pth$.

    Note that, by symmetry between $z$ and $-z$ in the LHS, it suffices upper bound $\nu^\top \phi z$.
    We expand this expression as follows:
    \[
        \nu^\top \phi z
        = \sum_{i,j \in \bZ_n} \nu_i z_j \phi_{i,j}
        = \frac{1}{n} \sum_{j=0}^{n-1} z_j \left(\sum_{i=0}^{n-1} \phi_{i,j}\right) \,.
    \]
    Hence our goal is to show
    \begin{equation}
        \label{eq:z-j-sum}
        \sum_{j=0}^{n-1} z_j S_j \lequestion \frac{2\delta}{\rho^2} \,,
    \end{equation}
    where $S_j \define \left(\sum_{i=0}^{n-1} \phi_{i,j}\right)$ is the sum of the entries in the $j$-th
    column of $\phi$. Note that $(S_j)_{j=0,\dotsc,n-1}$ is a symmetric unimodal sequence (first
    increasing, then decreasing) with strict inequalities everywhere except for indices
    $\floor{(n-1)/2}$ and $\ceil{(n-1)/2}$ when $n$ is even.
    We will use a ``rearrangement and saturation'' argument
    to construct a vector $z^*$ that upper bounds the LHS of \eqref{eq:z-j-sum} (hereafter called
    the \emph{objective value}).

    Let $z'$ be a vector satisfying the conditions from the statement (hereafter called a
    \emph{feasible solution}) that maximizes the objective value. Let $\sigma$ be a permutation of
    $\{0, \dotsc, n-1\}$ that puts the sequence of column sums in non-decreasing order:
    $S_{\sigma(0)} \le \dotsm \le S_{\sigma(n-1)}$. Then we can also assume that $z'$ respects this
    order: $z'_{\sigma(0)} \le \dotsm \le z'_{\sigma(n-1)}$, since otherwise rearranging the entries
    of $z'$ would yield another feasible solution with equal or larger objective value.

    We now argue that we may assume that, among all nonzero entries of $z'$, all have absolute
    value equal to $\delta$ (which we call \emph{saturated entries}) except for at most one positive
    entry and one negative entry. Indeed, if two consecutive (under $\sigma)$ nonzero entries with
    the same sign are not saturated, \ie they satisfy
    $\abs*{z'_{\sigma(i)}}, \abs*{z'_{\sigma(i+1)}} < \delta$, then we can obtain another feasible
    solution with equal or larger objective value by ``saturating'' this pair of entries, \ie
    making $z'_{\sigma(i)}$ smaller and $z'_{\sigma(i+1)}$ larger until either of them reaches
    a value in $\{-\delta, 0, \delta\}$.

We claim that we may also assume that the multiset of values of the positive entries of $z'$ is
equal to the multiset of absolute values of the negative entries of $z'$.  Suppose $z'$ has $N^+$
entries equal to $\delta$, $N^-$ entries equal to $-\delta$, $M^+ \in \{0,1\}$ entries in the
interval $(0, \delta)$, and $M^- \in \{0,1\}$ entries in the interval $(-\delta, 0)$. If $N^+ =
N^-$, then since $\sum_j z'_j = 0$, we must also have $M^+ = M^-$ and, if this value is $1$, then
the corresponding entries of $z'$ must have the same absolute value so that they add to zero. On the
other hand, if $N^+ \ne N^-$, say $N^+ > N^-$ without loss of generality, then $\sum_i z'_i > 0$
since the sum of the saturated values of $z'$ is at least $\delta$ while the sum of the
non-saturated values must be in $(-\delta, \delta)$. This contradicts the fact that $z'$ is a
feasible solution.

Now we construct $z^*$ by saturating the remaining (zero or two) entries of $z'$:
\[
    z^*_i \define \begin{cases}
        \delta,  & \text{if $z'_i > 0$} \\
        -\delta, & \text{if $z'_i < 0$} \\
        0,       & \text{if $z'_i = 0$.}
    \end{cases}
\]
Then by the same arguments as above, $z^*$ has equal or larger objective value as $z'$.
We now upper bound this objective value by the RHS of \eqref{eq:z-j-sum}, which will
conclude the argument.

    Let $N$ be the number of positive entries of $z^*$. By construction, we have
    \[
        N = |\{i \in \bZ_n : z^*_i = \delta\}| = |\{i \in \bZ_n : z^*_i = -\delta\}| \,.
    \]
    Then our objective value is
    \begin{equation}
        \label{eq:obj-val}
        \sum_j z^*_j S_j
        = \delta \left[ -\sum_{i=0}^{N-1} S_{\sigma(i)} + \sum_{i=0}^{N-1} S_{\sigma(n-1-i)} \right] \,.
    \end{equation}

    Let $h \define (n-1)/2$.
    Since $(S_j)_{j=0,\dotsc,n-1}$ is a symmetric unimodal sequence attaining its maximum in the
    middle, we may say without loss of generality that the indices $\sigma(i)$ in the first
    summation term in the RHS of \eqref{eq:obj-val} are
    $\{0, \dotsc, \ceil{N/2}-1\} \cup \{n-1, \dotsc, n - \floor{N/2}\}$.
    As for the indices $\sigma(n-1-i)$, an exact account depends on the parity of $n$,
    but we can only make the objective value larger by simply using the indices
    $\{\ceil{h}, \dotsc, \ceil{h} + \ceil{N/2} - 1\} \cup
    \{\floor{h}, \dotsc, \floor{h} - \floor{N/2} + 1\}$. Note that when $n$ is odd, this choice
    slightly overestimates the objective value by using the maximum value $S_{h}$ twice,
    but this looser bound suffices for our purposes.

    Therefore, we may finally express and compute our upper bound on the objective value of
    any feasible vector $z$. Recall that $\phi_{i,j} = \eta^{\abs{i-j}}$.
    In the edge case when $\rho=1$ and thus $\eta=0$, we have that $\phi$ is the identity matrix
    and hence $S_j = 1$ for every $j \in \bZ_n$. Therefore we obtain
    \begin{align*}
        \sum_{j=0}^{n-1} z_j S_j = \sum_{j=0}^{n-1} z_j = 0 \,,
    \end{align*}
    which satisfies \eqref{eq:z-j-sum} and we are done. Now, suppose $0 < \eta < 1$.
    Then
    \begin{align*}
        \sum_{j=0}^{n-1} z_j S_j
        &\le \delta\left[ -\sum_{i=0}^{N-1} S_{\sigma(i)} + \sum_{i=0}^{N-1} S_{\sigma(n-1-i)} \right] \\
        &\le \delta\left[
            - \sum_{i=0}^{\ceil{N/2}-1} S_i
            - \sum_{i=0}^{\floor{N/2}-1} S_{n-1-i}
            + \sum_{i=0}^{\ceil{N/2}-1} S_{\ceil{h}+i}
            + \sum_{i=0}^{\floor{N/2}-1} S_{\floor{h}-i}
            \right]
        \le \frac{2\delta}{\rho^2} \,,
    \end{align*}
    where the last inequality is due to \cref{prop:geometric-sum-expression} below.
\end{proof}

\begin{proposition}
    \label{prop:geometric-sum-expression}
    Let $\rho \in (0,1)$. Let $S_i$ denote the sum of the entries in the $i$-th column of
    $\phi^\pth$ for each $i \in \bZ_n$, and let $h \define (n-1)/2$.
    Then for any non-negative integer $N \le n/2$,
    \[
        - \sum_{i=0}^{\ceil{N/2}-1} S_i
            - \sum_{i=0}^{\floor{N/2}-1} S_{n-1-i}
            + \sum_{i=0}^{\ceil{N/2}-1} S_{\ceil{h}+i}
            + \sum_{i=0}^{\floor{N/2}-1} S_{\floor{h}-i}
        <
        \frac{2}{\rho^2} \,.
    \]
\end{proposition}
\begin{proof}
    Recall that $\phi^\pth_{i,j} = \eta^{\abs{i-j}}$ where $\eta = 1-\rho$.
    We express the column sums explicitly and reduce the geometric sums that emerge:
    \begin{align*}
        &\left[
            - \sum_{i=0}^{\ceil{N/2}-1} S_i
            - \sum_{i=0}^{\floor{N/2}-1} S_{n-1-i}
            + \sum_{i=0}^{\ceil{N/2}-1} S_{\ceil{h}+i}
            + \sum_{i=0}^{\floor{N/2}-1} S_{\floor{h}-i}
            \right] \\
        &\quad= \left[
            \begin{array}{l}
                - \sum_{i=0}^{\ceil{N/2}-1} \sum_{j=0}^{n-1} \eta^{\abs{i-j}}
                - \sum_{i=0}^{\floor{N/2}-1} \sum_{j=0}^{n-1} \eta^{\abs{n-1-i-j}} \\
                + \sum_{i=0}^{\ceil{N/2}-1} \sum_{j=0}^{n-1} \eta^{\abs{\ceil{h}+i-j}}
                + \sum_{i=0}^{\floor{N/2}-1} \sum_{j=0}^{n-1} \eta^{\abs{\floor{h}-i-j}}
            \end{array}
            \right] \\
        &\quad= \left[
            \begin{array}{l}
                - \sum_{i=0}^{\ceil{N/2}-1} \left(
                    \sum_{j=0}^i \eta^{i-j}
                    + \sum_{j=i+1}^{n-1} \eta^{j-i} \right)
                \\
                - \sum_{i=0}^{\floor{N/2}-1} \left(
                    \sum_{j=0}^{n-1-i} \eta^{n-1-i-j}
                    + \sum_{j=n-i}^{n-1} \eta^{j+i+1-n} \right)
                \\
                + \sum_{i=0}^{\ceil{N/2}-1} \left(
                    \sum_{j=0}^{\ceil{h}+i} \eta^{\ceil{h}+i-j}
                    + \sum_{j=\ceil{h}+i+1}^{n-1} \eta^{j-i-\ceil{h}} \right)
                \\
                + \sum_{i=0}^{\floor{N/2}-1} \left(
                    \sum_{j=0}^{\floor{h}-i} \eta^{\floor{h}-i-j}
                    + \sum_{j=\floor{h}-i+1}^{n-1} \eta^{j+i-\floor{h}} \right)
            \end{array}
            \right] \\
        &\quad= \left[
            \begin{array}{l}
                - \sum_{i=0}^{\ceil{N/2}-1} \left(
                    \frac{\eta^{i} - \eta^{-1}}{1 - \eta^{-1}}
                    + \frac{\eta - \eta^{n-i}}{1 - \eta} \right)
                - \sum_{i=0}^{\floor{N/2}-1} \left(
                    \frac{\eta^{n-1-i} - \eta^{-1}}{1 - \eta^{-1}}
                    + \frac{\eta - \eta^{i+1}}{1 - \eta} \right)
                \\
                + \sum_{i=0}^{\ceil{N/2}-1} \left(
                    \frac{\eta^{\ceil{h}+i} - \eta^{-1}}{1 - \eta^{-1}}
                    + \frac{\eta - \eta^{n-i-\ceil{h}}}{1 - \eta} \right)
                + \sum_{i=0}^{\floor{N/2}-1} \left(
                    \frac{\eta^{\floor{h}-i} - \eta^{-1}}{1 - \eta^{-1}}
                    + \frac{\eta - \eta^{n+i-\floor{h}}}{1 - \eta} \right)
            \end{array}
            \right] \\
        &\quad= \frac{1}{\rho} \left[
            \begin{array}{l}
                - \sum_{i=0}^{\ceil{N/2}-1} \left(
                    1 - \eta^{i+1}
                    + \eta - \eta^{n-i} \right)
                - \sum_{i=0}^{\floor{N/2}-1} \left(
                    1 - \eta^{n-i}
                    + \eta - \eta^{i+1} \right)
                \\
                + \sum_{i=0}^{\ceil{N/2}-1} \left(
                    1 - \eta^{\ceil{h}+1+i}
                    + \eta - \eta^{n-i-\ceil{h}} \right)
                + \sum_{i=0}^{\floor{N/2}-1} \left(
                    1 - \eta^{\floor{h}+1-i}
                    + \eta - \eta^{n+i-\floor{h}} \right)
            \end{array}
            \right] \\
        &\quad= \frac{1}{\rho} \left[
            \begin{array}{l}
                \sum_{i=0}^{\ceil{N/2}-1} \eta^{i+1}
                + \sum_{i=0}^{\ceil{N/2}-1} \eta^{n-i}
                + \sum_{i=0}^{\floor{N/2}-1} \eta^{n-i}
                + \sum_{i=0}^{\floor{N/2}-1} \eta^{i+1}
                \\
                - \sum_{i=0}^{\ceil{N/2}-1} \eta^{\ceil{h}+1+i}
                - \sum_{i=0}^{\ceil{N/2}-1} \eta^{n-i-\ceil{h}}
                - \sum_{i=0}^{\floor{N/2}-1} \eta^{\floor{h}+1-i}
                - \sum_{i=0}^{\floor{N/2}-1} \eta^{n+i-\floor{h}}
            \end{array}
            \right] \\
        &\quad= \frac{1}{\rho} \left[
            \begin{array}{l}
                \frac{\eta - \eta^{\ceil{N/2}+1}}{1-\eta}
                + \frac{\eta^{n} - \eta^{n-\ceil{N/2}}}{1 - \eta^{-1}}
                + \frac{\eta^{n} - \eta^{n-\floor{N/2}}}{1 - \eta^{-1}}
                + \frac{\eta - \eta^{\floor{N/2}+1}}{1 - \eta}
                \\
                - \frac{\eta^{\ceil{h}+1} - \eta^{\ceil{h}+1+\ceil{N/2}}}{1 - \eta}
                - \frac{\eta^{n-\ceil{h}} - \eta^{n-\ceil{h}-\ceil{N/2}}}{1 - \eta^{-1}}
                - \frac{\eta^{\floor{h}+1} - \eta^{\floor{h}+1-\floor{N/2}}}{1 - \eta^{-1}}
                - \frac{\eta^{n-\floor{h}} - \eta^{n+\floor{N/2}-\floor{h}}}{1 - \eta}
            \end{array}
            \right] \\
        &\quad= \frac{1}{\rho^2} \left[
            \begin{array}{l}
                \eta - \eta^{\ceil{N/2}+1}
                + \eta^{n+1-\ceil{N/2}} - \eta^{n+1}
                \\
                + \eta^{n+1-\floor{N/2}} - \eta^{n+1}
                + \eta - \eta^{\floor{N/2}+1}
                \\
                - \eta^{\ceil{h}+1} + \eta^{\ceil{h}+1+\ceil{N/2}}
                - \eta^{n+1-\ceil{h}-\ceil{N/2}} + \eta^{n+1-\ceil{h}}
                \\
                - \eta^{\floor{h}+2-\floor{N/2}} + \eta^{\floor{h}+2}
                - \eta^{n-\floor{h}} + \eta^{n+\floor{N/2}-\floor{h}}
            \end{array}
            \right] \\
        &\quad= \frac{1}{\rho^2} \left[
            \begin{array}{l}
                (\eta + \eta)
                + (\eta^{n+1-\ceil{N/2}} - \eta^{\ceil{N/2}+1})
                + (\eta^{n+1-\floor{N/2}} - \eta^{\floor{N/2}+1})
                - (\eta^{n+1} + \eta^{n+1})
                \\
                + (\eta^{\ceil{h}+1+\ceil{N/2}} - \eta^{\ceil{h}+1})
                + (\eta^{n+1-\ceil{h}} - \eta^{n+1-\ceil{h}-\ceil{N/2}})
                \\
                + (\eta^{\floor{h}+2} - \eta^{\floor{h}+2-\floor{N/2}})
                + (\eta^{n+\floor{N/2}-\floor{h}} - \eta^{n-\floor{h}})
            \end{array}
            \right] \\
        &\quad< \frac{2}{\rho^2} \,,
    \end{align*}
    where we used the facts that $0 < \eta < 1$ and, in the last step, that $2\ceil{N/2} \le n$
    (which holds because $2N \le n$ by assumption).
\end{proof}

\end{document}